\documentclass[a4paper,reqno]{amsart} 

\usepackage{graphicx}
\usepackage{amsmath,amssymb,amsfonts,latexsym,amsthm}
\usepackage{url}
\usepackage{dsfont}
\usepackage{pstool}
\usepackage{caption}
\usepackage{subcaption}
\usepackage{microtype}
\usepackage{fancyhdr}
\usepackage{textcomp}
\usepackage[active]{srcltx}

\usepackage[round]{natbib}

\usepackage[top=2.7cm, left=2.7cm, right=2.7cm]{geometry}
\usepackage[pdftex,colorlinks=true,urlcolor=blue,citecolor=blue, linkcolor=blue]{hyperref}

\newtheorem{theorem}{Theorem}[section]
\newtheorem{lemma}[theorem]{Lemma}
\newtheorem{example}[theorem]{Example}
\newtheorem{proposition}[theorem]{Proposition}
\newtheorem{definition}[theorem]{Definition}
\newtheorem{corollary}[theorem]{Corollary}
\newtheorem{remark}[theorem]{Remark}

\newtheorem{assumption}[theorem]{Assumption}

\numberwithin{equation}{section}

\newcommand{\filtration}{\left(\mathcal{F}_t\right)_{t\geq0}}
\newcommand{\F}{\mathcal{F}}

\newcommand{\B}{\mathcal{B}}
\newcommand{\G}{\Gamma}
\newcommand{\Pred}{\mathcal{P}}
\newcommand{\h}{\mathcal{H}}
\newcommand{\filt}{\mathbb{F}}
\newcommand{\essinf}{\mathop{\mbox{ess inf}}}
\newcommand{\esssup}{\mathop{\mbox{ess sup}}}

\newcommand\abs[1]{\left\vert {#1} \right\vert}

\newcommand {\R}{\mathbb {R}}
\newcommand {\N}{\mathbb {N}}
\newcommand {\Q}{\mathbb {Q}}

\makeatletter

\date{}
\title[Hedging, good-deal bounds and uncertainty]{Hedging under generalized good-deal bounds\\ and model uncertainty}

\author[D. Becherer]{Dirk Becherer}
\address[D. Becherer]{Institut f\"ur Mathematik, Humboldt-Universit\"at zu Berlin, D-10099 Berlin, Germany}
\email{becherer\,@\,mathematik.hu-berlin.de}
\author[K. Kentia]{Klebert Kentia}
\address[K. Kentia]{Previous: Institut f\"ur Mathematik, Humboldt-Universit\"at zu Berlin, D-10099 Berlin, Germany \linebreak Current: Institut f\"ur Mathematik, Goethe-Universit\"at Frankfurt, D-60054 Frankfurt a.M., Germany}
\email{kentia\,@\,math.uni-frankfurt.de}
\thanks{Support from the German Science Foundation DFG via the Berlin Mathematical School and the Research Training Group 1845 Sto-A is gratefully acknowledged.}

\begin{document}


\begin{abstract}
We  study a notion of good-deal hedging that corresponds to good-deal valuation and is 
described by
 a uniform supermartingale property for the tracking errors of hedging strategies.
For generalized good-deal constraints, defined in  terms of  correspondences for the Girsanov kernels of pricing measures,
 constructive results on good-deal hedges and valuations are derived from backward 
stochastic differential equations, including new examples with explicit formulas. Under model un\-cer\-tainty about the market prices of risk of hedging assets, a robust approach 
leads to a reduction or even elimination of a speculative component in good-deal hedging, which is shown to be equivalent to a global risk-minimization the sense of 
\citet[][]{FollmerSondermann} if uncertainty is sufficiently large. 
\end{abstract}
\maketitle

\section{Introduction}\label{sec:Introduction}

The theory of good-deal bounds offers a valuation approach for contingent 
claims in incomplete markets, where the issuer of derivative contracts cannot eliminate his risk entirely by dynamic hedging but only partially. Good deal bounds are 
defined as economically meaningful valuation bounds by ruling  out not only arbitrage opportunities but also a suitable notion of deals that are 'too good'. The most 
cited work appears to be  \citet{CochraneRequejo}, and for pioneering conceptual contributions we refer also to \citet{CernyHodges,Cerny, JaschkeKuechler}. 
Mathematical theory in continuous time, where earlier contributions have partly relied on heuristic arguments, has been made rigorous and generalized in 
\citet[][]{BjorkSlinko,KloppelSchweizer}.
Inherent to the concept of good-deal valuation bounds is already a certain notion of robustness,
in that it considers
a 'no good-deal' range of
arbitrage-free valuations over a suitable set of risk neutral good-deal pricing measures, instead of using one single 
calibrated pricing measure. Such is conservative for risks that cannot be perfectly hedged, e.g.\  illiquid exotic options, and also with respect to  mark-to-market 
losses that could arise during the regular recalibration of risk-neutral valuation models to the actual market prices of liquid options, if model predictions for price relations differ from reality.  
Good deals have been defined mostly in terms of Sharpe ratios \citep[][]{CochraneRequejo,CernyHodges,Cerny,BjorkSlinko,BayraktarYoung,Delong12,CarassusTemam}
or alternatively by some notion of (quadratic or else) expected utilities \citep[cf.][]{Cerny,KloppelSchweizer,Becherer}. 
Either definition relates to a single objective probability $P$ in terms of market prices of risk. Hence robustness and model uncertainty, being important for finance and decision theory
\citep[cf.][]{Cont, HansenSargent,GilboaS89} in general,
 are relevant problems for good-deal theory. 
This theory has evolved for some time  purely as a pricing theory; in the conclusions of \citet[][]{BjorkSlinko} a quest for a corresponding dual theory of good deal hedging was noted 
as a "challenging open problem". Different notions for good deal hedging have been proposed in \citet[][]{Becherer}, where 
good-deal hedging strategies are defined as minimizers of a-priori  risk measure in spirit of \citet[][]{BarrieuElKaroui}, or \citet[][]{CarassusTemam} who suggest minimum 
variance hedging and demonstrate numerically that both approaches perform comparably well.

The present paper is concerned with approaches to good deal hedging \citep[as in][]{Becherer} under Knightian model uncertainty (ambiguity) about the objective probability, 
with respect to which good deals are defined. To this end, we pose the good-deal valuation and hedging problem in a framework with multiple (uncertain) priors,
 and follow a robust worst-case approach as in \citet{GilboaS89}.
Good deal valuation fits into the theory of dynamic monetary convex, even coherent, risk measures (or monetary utility functionals), for which a rich theory in high generality exists \citep[see e.g.][]{KloeppelS07,BionDiNunno13,DrapeauKupper13}.
On dual representations and time consistency of value functions there 
 seems little to add.  Instead, we contribute  constructive and qualitative results on the (robust, good-deal) hedging strategies.
Hedging strategies arise as minimizers for suitable a-priori coherent risk measures under optimal risk sharing with the market.
Literature on robust utility maximizing {strategies} appears richer than on robust hedging {strategies} (except on superhedging
 \citep[cf.][]{ElKarozuiQuenez95}, which excludes losses but is often too expensive); for  innovative approaches and many more references on the former we refer to 
\citet{DowWerlang,Quenez,Garlappi,Schied08,Seifried10,Biagini}. Results  on robust good deal hedging in the interesting recent work by \citet{BoyarchenkoCerratoCrosbyHodges} are different from ours. They study 
a different uncertainty-penalized preference functional and obtain numerical results in discrete time, whereas we use dynamic coherent risk measures in continuous time and focus mostly on analytical results. 

To derive the optimality equations, we systematically use backward stochastic differential 
equations {(BSDEs)} \citep[see][]{KarouiPengQuenez}, which in some informal way already appear present in \citep[][Sect.II.C.6]{CochraneRequejo}, instead of Markovian dynamical programming methods and PDE \citep[as][]{CochraneRequejo,BjorkSlinko}.
To this end,  Section~\ref{sec:Prelim} formulates an abstract framework of good-deal  constraints, described by predictable correspondences for the Girsanov kernels of the respective pricing 
measures, that is sufficiently general for all later sections. It incorporates the common radial good-deal constraints on Girsanov kernels, defined by (constant) scalar bounds on 
Euclidian norms, that are predominant in the good deal literature
\citep[cf.][]{BjorkSlinko,BayraktarYoung,BondarenkoLongarela,Donnelly,MarroquinMoreno} but
also extensions like ellipsoidal constraints, that still permit explicit analytic generators in the BSDEs of interest, 
being efficient for 
Monte Carlo approximation.  Notably, the generalized constraints  are needed to cover relevant examples  in Section~\ref{sec:GDRestrictionUnderUncertainty} for model uncertainty about  the market 
prices of risk  of the assets that are available  in the (incomplete) market for partial hedging, cf.\ Remarks~\ref{Rem:C0plusTheta}-\ref{Rem:ellipUnc}.
Section~\ref{sec:No-Good-Deal-HedgingApproach} studies good-deal hedging strategies and provides new case studies  with closed-form solutions. 
For an exchange option between tradeable and non-tradeable assets, the good-deal bounds are given by a  Margrabe formula with adjusted 
input parameters. For the stochastic volatility model by Heston, we obtain  semi-explicit formulas under good-deal constraints for pricing measures, which restrict the mean reversion 
level of the stochastic variance process to some interval. This shows, how a conservative valuation for volatility related risks could be obtained  in a (simpler) dominated setup without singularity of  measures. To illustrate to which extend our BSDE solutions can be computed  by 
efficient but generic Monte Carlo methods, complementing numerical results by \citet[][]{CarassusTemam}, we investigate errors between an efficient but generic Monte Carlo approximation 
and our analytic formula in a 4-dimensional example. In a framework without uncertainty,  good deal hedging strategies are shown to comprise a speculative component 
to compensate for unhedgeable risks, see Section~\ref{subsec:Example1withEllipsoids}.

Good deal  hedges and valuations that are robust with respect to uncertainty
are derived in  the main Section~\ref{sec:GDRestrictionUnderUncertainty}. This builds on the analysis of 
previous sections since the problem with uncertain multiple priors can be related to a respective problem without uncertainty for  suitably  enlarged  correspondences, within a (dominated) setup of absolutely continuous probability measures. 
A worst-case approach naturally leads to robust valuation by the 
widest good-deal bounds that are obtained over all (dominated) probabilistic models under consideration. We show that there is also a corresponding notion for robust dynamic hedging.
Indeed, there exists a hedging strategy such that its tracking error, given by the dynamic variation of the robust good-deal bounds 
plus the profits and losses from  hedging, exhibits a supermartingale property uniformly over all a-priori valuation measures with respect to all  priors. 
This means, that hedging strategies are at least mean-self-financing \cite[in the sense of][]{Schweizer} in a uniform sense.
By saddle point arguments we derive a minmax identity, that shows how the robust good-deal hedging strategy is given by the (ordinary) good-deal hedging strategy with respect to a 
suitable worst case measure (for a related result in the context of robust utility maximization cf.\ e.g.\ \cite{Schied08}). Both are identified constructively, see~Remark~\ref{Rem:ellipUnc}. 
As to be expected, a robust approach to uncertainty reduces the speculative component of the good-deal hedging strategy. As a further contribution, we prove that if the uncertainty is 
large enough in relation to the good-deal constraints, then the robust good-deal hedging strategy does no longer include any {speculative} component, but coincides with 
a (globally) risk minimizing strategy in spirit of \citet[][]{FollmerSondermann} \citep[cf. the survey of][]{Schweizer} with respect to a suitable measure. This offers theoretical support to the commonly held perception
that hedging should abstain from speculative objectives
and offers  a new  justification for risk-minimization, which may be criticized at first 
sight for using a quadratic hedging criterion that penalizes gains and losses alike.

\section{Mathematical framework and preliminaries}\label{sec:Prelim}
We work on a filtered probability space $(\Omega, \F,\filt,P)$ with time horizon $T<\infty$; the filtration $\filt=(\F_t)_{t\le T}$  generated by an $n$-dimensional Brownian motion 
$W$, augmented with $P$-null-sets, satisfying the usual conditions. Let 
 $\F=\F_T$. Inequalities between random variables (processes) are meant to hold almost everywhere 
with respect to $P$ (resp.\ $P\otimes dt$). For stopping times $\tau\le T$, the conditional expectation given $\F_\tau$ under a probability 
measure $Q$ is denoted by  $E^Q_\tau\big[\cdot\big]$. We write $E_\tau=E^P_\tau$ if there is no ambiguity about $P$. 
$L^p(\R^m,Q)$, $ p\in[1,\infty)$, 
(or $L^\infty(\R^m,Q)$) denotes the space of $\F_T$-measurable 
$\R^m$-valued random variables $X$ with $\lVert X\rVert_{L^p(Q)}^p = E^Q\left[\lvert X\rvert^p\right]<\infty$ (resp.\ $X$ $Q$-essentially bounded). $\Pred$ denotes the predictable 
$\sigma$-field on $[0,T]\times\Omega$.  Stochastic integrals of predictable integrands $H$ with respect to semimartingales $S$ are denoted $H\cdot S=\int_0^\cdot H_t^{\textrm{tr}}dS_t$. 
Let $\h^p(\R^{m},Q)$ denote the space of predictable $\R^{m}$-valued processes $Z$ with $\lVert Z\rVert_{\h^p(Q)}^p = E^Q\left[\big(\int_0^T\lvert Z_s\rvert^2 ds\big)^{\frac{p}{2}}\right]<\infty$,  and 
$\mathcal{S}^p(Q)$ that of c\`adl\`ag 
semimartingales $Y$ with $\lVert Y\rVert_{\mathcal{S}^p(Q)}=\left\Vert \sup_{t\le T}\abs{Y_t}\right\Vert_{L^p(Q)}<\infty$ 
 If the dimension is clear, we just write $L^p(Q)$ and $\h^p(Q)$, and if $Q=P$ 
just $L^p,\ \h^p$ and $\mathcal{S}^p$, for $p\in[1,\infty]$. The Euclidean norm of a matrix $M\in\R^{n\times d}$ is $\lvert M\rvert:=(\mathrm{Tr}\ MM^{\rm{tr}})^{1/2}$ and its usual operator 
norm is denoted by $\lVert M\rVert$.

We will make use of classical theory of BSDEs \citep{PardouxPeng90,KarouiPengQuenez}. BSDEs are stochastic differential equations of the type
        \begin{equation}\label{eq:BSDEdef}
        	-dY_t = f(t,Y_t,Z_t)dt -Z^{\textrm{tr}}_tdW_t, \textrm{ for }t\le T,\textrm{ and }\ Y_T=X,
        \end{equation}
where the terminal condition $X$ is an $\F_T$-measurable random variable and the generator $f:\Omega\times[0,T]\times\R^{1+n}\rightarrow \R$ a $\Pred\otimes\B(\R^{1+n})$-$\B(\R)$-measurable function.  They are well established in mathematical economics.
A pair $(f,X)$ constitutes standard parameters (also called data) for a BSDE \eqref{eq:BSDEdef} if $X\in L^2$, $f(\cdot,0,0)$ is in $\h^2$ and $f$ is uniformly Lipschitz in $y$ and $z$, i.e.\  there 
exists $L\in(0,\infty)$ such that $P\otimes dt$-a.e.
                \(
                	\ \lvert f(\omega,t,y,z)-f(\omega,t,y',z')\rvert \le L(\lvert y-y'\rvert+\lvert z-z'\rvert)\  
                \)
for all $y,z,y',z'$.
A solution of the BSDE (\ref{eq:BSDEdef}) is a couple $(Y,Z)$ of processes such that $Y$ is real-valued continuous, adapted, and $Z$ is $\R^{n}$-valued predictable and satisfies $\int_0^T\lvert Z_t\rvert^2dt<\infty$. 
For  standard parameters $(f,X)$ there exists a unique solution $(Y,Z)\in\mathcal{S}^2\times\h^2$ to the BSDE (\ref{eq:BSDEdef}), \citep[Thm.2.1]{KarouiPengQuenez}. Let us refer to BSDEs with standard parameters 
as {\em classical} and to the solution to such BSDEs as {\em standard}. A comparison theorem \citep[Prop.3.1]{KarouiPengQuenez} is very useful for optimal control problems stated in terms of classical BSDEs: 
Given standard BSDE solutions  $(Y,Z),(Y^a,Z^a)_{a\in A}$ for a family of standard parameters $(f,X),(f^a,X^a)_{a\in A}$, if there exists $\bar{a}\in A$ such that 
$\displaystyle f(t,Y_t,Z_t) = \essinf_{a\in A}\ f^a(t,Y_t,Z_t)=f^{\bar{a}}(t,Y_t,Z_t),\ P\otimes dt$-a.e., and $\displaystyle X = \essinf_{a\in A}\ X^a = X^{\bar{a}}$, then $\displaystyle Y_t = \essinf_{a\in A}\ Y_t^a = Y_t^{\bar{a}}$ 
holds for all $t\le T.$

Section~\ref{sec:Parameterization} will specify   a financial market with $d$ risky assets whose discounted price processes $S^i$ ($ i\le d$) with respect to a fixed num\'eraire asset (with unit price $S^0=1$) are non-negative locally bounded semimartingales. 
The set of equivalent local martingale measures (risk neutral pricing measures) is denoted by $\mathcal{M}^e:=\mathcal{M}^e(S)$ and we assume $\mathcal{M}^e\neq \emptyset$, i.e.\
there is no free lunch with vanishing risk in the sense of \citet{DelbaenSchachermayer}. The market is incomplete with  $\mathcal{M}^e$ being of infinite cardinality if $d<n$. 
Following \citet{Rockafellar}, we define generalized good-deal bounds by using abstract predictable correspondences (multifunctions) $C$ defined on $[0,T]\times\Omega$ with non-empty compact and convex values $C_t(\omega)\subset \R^n$, 
with predictability meaning that for each closed set $F\subset \R^n$, the set $C^{-1}(F):=\left\lbrace (t,\omega)\in[0,T]\times\Omega:\ C_t(\omega)\cap F\neq \emptyset\right\rbrace$ 
is predictable. More specific examples, e.g.\ for ellipsoidal constraints, will exhibit (semi)explicit solutions for optimizers.
 We write $C:[0,T]\times\Omega\rightsquigarrow \R^n$ 
with ``$\rightsquigarrow$'' for a set-valued mapping $C$, and $\lambda\in C$ means that the predictable function $\lambda$ is a {\em selection} of $C$, i.e.\  $\lambda_t(\omega)\in C_t(\omega)$ holds on  $[0,T]\times \Omega$. 
Note that the seemingly weaker  existence of a $\tilde{\lambda}$ with $\tilde{\lambda}_t(\omega)\in C_t(\omega)$ just $P\otimes dt$-a.e.\ on  $[0,T]\times \Omega$ implies existence of a selection $\lambda\in C$, if $0\in C\neq \emptyset$ (as often later).
 Throughout, a {\em standard} correspondence will refer to a predictable one, whose values are non-empty, compact and convex.
Let $C:[0,T]\times \Omega\rightsquigarrow\R^n$ be a fixed standard correspondence with $0\in C$. 
The set $\mathcal{Q}^{\mathrm{ngd}}:=\mathcal{Q}^{\mathrm{ngd}}(S)$ of (equivalent) no-good-deal measures is given by  
\begin{equation}\label{eq:QngdDefinition}
	      \mathcal{Q}^{\mathrm{ngd}}(S):=\Big\lbrace Q\in\mathcal{M}^e\Big\lvert\ dQ/dP\Big.=\mathcal{E}\left( \lambda\cdot W\right),
															 \lambda\ \textrm{predictable, bounded, } 
                                                                                                                        \lambda\in C\Big\rbrace.
\end{equation}
For good-deal valuation and hedging results later, concrete assumptions (e.g.\ Assumption~\ref{asp:BddCorr}) 
may ensure that all selections $\lambda$ of $C$ are bounded automatically.
We remark that for good time-consistency properties, good-deal constraints should be specified locally in time \citep[][]{KloppelSchweizer}.
For contingent claims $X$ (bounded or sufficiently integrable), upper and lower good-deal valuation bounds 
\begin{equation}\label{eq:definitionPiuandPil}
	\pi^l_t(X):=\essinf_{Q\in\mathcal{Q}^{\mathrm{ngd}}}\ E_t^Q[X]\quad \textrm{and}\quad \pi^u_t(X):=\esssup_{Q\in\mathcal{Q}^{\mathrm{ngd}}}\ E_t^Q[X],\quad t\in [0,T],
\end{equation}
are defined over a suitable  (yet abstract) set of no good deal pricing measures $\mathcal{Q}^{\mathrm{ngd}}$. Hence
 $\pi^u_t(X)$ (respectively $\pi^l_t(X)$) can be seen as the highest (lowest) valuation that does not permit too good deals to the seller (buyer). Since $\pi^l_\cdot(X)=-\pi^u_\cdot(-X)$, further analysis can be restricted to $\pi^u_\cdot(X)$.

 Referring to \citet{ElKarozuiQuenez95}, we note that the good deal bounds in (\ref{eq:definitionPiuandPil}) are  within the interval of no-arbitrage prices, whose upper (lower) bound 
 corresponds to the minimal superreplication cost for the seller (respectively buyer) of the claim $X$; this  follows from (\ref{eq:definitionPiuandPil}) and $\mathcal{Q}^{\mathrm{ngd}}\subset \mathcal{M}^e$. 
 In the same context, one may note that the BSDE (\ref{eq:BSDEPiu}), which will describe the good-deal valuation bound $\pi^u(X)$,  involves a non-negative term that can be interpreted as a penalization which 
 is proportional to the length of the 'non-hedgeable' part $\Pi_t^\bot(Z_t)$ of $Z$ (cf.\ (\ref{eq:CorrespondenceGammaandGammabot})); In the common (radial, cf.\ Section~\ref{subsec:Example1withEllipsoids}) 
 specification of good deals in the literature, this term reads as $k| \Pi_t^\bot(Z_t)|$ for some scalar parameter $k>0$, 
so  it is intuitive  that the good deal bound tends to the superreplication price  for $k\to \infty$.

As mentioned already in the introduction, the definition (\ref{eq:definitionPiuandPil}) in itself could already be viewed as a robust representation in a sense (over $Q$'s).  
For the purposes of the present paper however, the correspondence $C$  and the respective set $\mathcal{Q}^{\mathrm{ngd}}$ of no good deal measures are (at first) given with respect 
to one  objective real world measure $P$ (cf.\ remarks after (\ref{eq:QngdBARDefinition})). 
To be clear in our use of terminology,  we will in the sequel  restrict our use of terms {\em  model uncertainty, ambiguity} or {\em robust hedging/valuation} to situations with Knightian model uncertainty about $P$. 
Note that the use of terminology in some literature  \citep[e.g.][]{Delong12} is different, where those terms 
may instead refer to representations like (\ref{eq:definitionPiuandPil}).
Definition (\ref{eq:QngdDefinition}) implies that density processes 
of measures $Q\in\mathcal{Q}^{\mathrm{ngd}}$ are in $\mathcal{S}^p$, $p\in [1,\infty)$. Hence $X\in L^2=L^2(P)\subset L^1(Q)$.
In particular for $X\in L^\infty\subset L^2$, we have (cf.\ \citep[][Thm.3.7]{KentiaPhD} and Prop.~\ref{pro:OptimizationGeneratorPiuasBSDEsolution})  
that $\pi^u_t(X) = \esssup_{Q\in\overline{\mathcal{Q}^{\mathrm{ngd}}}}\ E_t^Q[X]$, where 
\begin{equation}\label{eq:QngdBARDefinition}
	      \overline{\mathcal{Q}^{\mathrm{ngd}}}:=\Big\lbrace Q\in\mathcal{M}^e\Big\lvert\ dQ/dP\Big.=\mathcal{E}\left( \lambda\cdot W\right),
															\ \lambda\ \textrm{predictable and } 
          \lambda\in C\Big\rbrace
\end{equation}
is a larger set than $\mathcal{Q}^{\mathrm{ngd}}$, containing measures with Girsanov kernels that are not necessarily bounded.
In the definition (\ref{eq:QngdBARDefinition}) and in  subsequent definitions of sets of equivalent measures, we tacitly assume (if not automatically satisfied) that Girsanov kernels 
$\lambda$ are such that stochastic exponentials $\mathcal{E}\left( \lambda\cdot W\right)$ are uniformly integrable martingales.
We recall that for radial constraints $C$  (like in  (\ref{Constrellip}) with $A\equiv\mathrm{Id}_{\R^n}$ and constant $h\in (0,\infty)$), common in the good-deal literature, one has a known financial 
justification. By a direct duality argument, one can see \citep[e.g.\ in a semimartingale framework,][Sect.3]{Becherer} that any (arbitrage-free) extension ${\bar S}=(S,S')$ of the 
market $S$
 by derivative price processes $S':=E_t^Q[X]$ for contingent claims $X$ (with $Q\in \mathcal{M}^e$, $X^-\in L^{\infty}$, $X^+\in L^1(Q)$)
 does  permit only for wealth processes $V>0$ of self-financing trading strategy (in $\bar{S}$) whose 
expected growth rates (log utilities) over any time period  $0\le t<\tau \le T$ satisfy the (sharp) estimate
  $E^P_{t}\left[ \log \frac{V_{\tau}}{V_t}\right] \le E^P_t\left[-\log  \frac{Z_{\tau}}{Z_t}\right]$,
 where $Z$ is the density process of $Q$. For $Q\in \mathcal{Q}^{\mathrm{ngd}}$ with  radial constraint, 
this estimate is bounded by ${h^2}(\tau-t)/2$, ensuring a bound $h^2/2$ to expected growth rates (good deals) for {\em any
market extension}  \citep[ideas going back at least to][]{CochraneRequejo,CernyHodges}. 

For the good-deal bounds to have nice dynamic properties, multiplicative stability (m-stability) of the set of no-good-deal measures is important. M-stability of dominated families of 
probability measures in dual representations (like e.g.\ (\ref{eq:definitionPiuandPil})) for dynamic coherent risk measures \citep[see e.g.][]{ArtznerDelbaenEberHeathKu} ensures in 
particular time consistency (recursiveness) and has been studied in a general context by  \citet{Delbaen}. 
\nocite{CheriditoK09}
In economics, it is known as rectangularity \citep{ChenEpstein}. A set 
$\mathcal{Q}$ of measures $Q\sim P$ is called m-stable if for all $Q^1,Q^2\in \mathcal{Q}$ with density processes $Z^1,Z^2$ and for all stopping times $\tau\le T$, 
the process $Z := I_{[0,\tau]}Z^1_\cdot+ I_{]\tau,T]}Z^1_\tau Z^2_\cdot/Z^2_\tau $ is the density process of a measure in $\mathcal{Q}$, where 
$\left[0,\tau\right]:=\left\lbrace(t,\omega)\in[0,T]\times\Omega\ \lvert \ t\le\tau(\omega)\right\rbrace$ denotes the stochastic interval
and $I_A$ is the indicator function on a set $A$. As noted in \citep[Rem.\ 6]{Delbaen}, by closure this definition  extends to sets of measures that are 
absolutely continuous but not necessarily equivalent; such is formally achieved by setting $Z^2_T/Z^2_\tau=1$ on $\{Z^2_\tau=0\}$. 
The role of m-stability shows in results due to \citet{Delbaen}, stated in 
Lemma~\ref{lem:PropDynamicRiskMeasMeandQngdareMstable}, Part a); for details cf.\ \citep[][Thm.2.7]{KloppelSchweizer} or \citep[][Prop.2.6]{Becherer}. Proof for part b) is provided in the appendix.

\begin{lemma}\label{lem:PropDynamicRiskMeasMeandQngdareMstable}
Let $\mathcal{Q}$ be a convex and m-stable set of probability measures $Q\sim P$ and $ \pi^{u,\mathcal{Q}}_t(X):=\esssup_{Q\in \mathcal{Q}}\ E^Q_t[X],$ for $X\in L^\infty$.
\\
\phantom{x}
a) 
 There exists a c\`adl\`ag process $Y$ such that for all stopping times $\tau\le T$,
                          $\displaystyle Y_\tau = \esssup\nolimits_{Q\in \mathcal{Q}}E^Q_\tau[X]=:\pi^{u,\mathcal{Q}}_\tau(X)$.
Moreover $\pi^{u,\mathcal{Q}}_\cdot(\cdot)$ has the properties of a dynamic coherent risk measure. It is recursive and stopping time consistent: 
For stopping times $\sigma\le\tau\le T$  holds $\pi^{u,\mathcal{Q}}_\sigma(X^1)= \pi^{u,\mathcal{Q}}_\sigma\left(\pi^{u,\mathcal{Q}}_\tau(X^1)\right)$,
and  $\pi^{u,\mathcal{Q}}_\tau(X^1)\ge \pi^{u,\mathcal{Q}}_\tau(X^2)$ for $X^1,X^2\in L^\infty$  implies $\pi^{u,\mathcal{Q}}_\sigma(X^1)\ge \pi^{u,\mathcal{Q}}_\sigma(X^2).$
Finally, a supermartingale property holds: For all stopping times $\sigma\le \tau\le T$ and $Q\in \mathcal{Q}$,  $\pi^{u,\mathcal{Q}}_\sigma(X)\ge E^Q_\sigma\big[\pi^{u,\mathcal{Q}}_\tau(X)\big]$,
                          and $\pi^{u,\mathcal{Q}}_\cdot(X)$ is a supermartingale under any $Q\in \mathcal{Q}$.
\\
\phantom{x}
b)
The sets $\mathcal{M}^e$ and $\mathcal{Q}^{\mathrm{ngd}}$ are m-stable and convex and hence for $\mathcal{Q}=\mathcal{Q}^{\mathrm{ngd}}$, 
$\pi^u_\cdot(X)=\pi^{u,\mathcal{Q}}_\cdot(X)$ satisfies the properties of Part a).
\end{lemma}

\subsection{Parametrizations  in an  It\^o process model}
\label{sec:Parameterization}
This section describes the It\^o process framework for the financial market, and details the parametrizations for dynamic trading strategies and no-good-deal constraints. 
The latter are specified at this stage by abstract correspondences (\ref{eq:QngdDefinition}) such that respective dynamic no-good-deal valuation bounds for contingent claims can be conveniently described in terms of  (super-)solutions to BSDEs (Sections~\ref{sec:NoGoodDealValuation_bdcorr}-\ref{sec:NoGoodDealValuation_unbdcorr}) within a convenient framework  that is sufficiently general for {all} later Sections~\ref{sec:No-Good-Deal-HedgingApproach}-\ref{sec:GDRestrictionUnderUncertainty}. 

We consider models for financial markets where prices $(S^i)_{i=1 \ldots d}$ of $d$ risky assets evolve according to a stochastic differential equation (SDE)
\begin{equation*}
	dS_t = \mathrm{diag}(S_t)\sigma_t\left(\xi_tdt+dW_t\right) =:\mathrm{diag}(S_t)\sigma_td\widehat{W}_t,\ t\in [0,T],\quad S_0\in (0,\infty)^d,
\end{equation*}
for predictable $\R^d$- and $\R^{d\times n}$-valued coefficients $\xi$ and $\sigma$, with $d\le n$. This includes basically all
examples of continuous price and state evolutions in (typically incomplete) markets of the good-deal literature, and permits also for non-Markovian evolutions.
Risky asset prices $S$ are given in units of some riskless num\'eraire asset whose discounted price $S^0\equiv 1$ is 
constant. 
We assume that $\sigma$ is of maximal rank $d\le n$ (i.e.\ $\det(\sigma_t\sigma_t^{\textrm{tr}})\neq 0$, that means no locally redundant assets) and that the market price of risk process $\xi$, satisfying 
$\xi_t\in\mathrm{Im}\ \sigma^{\textrm{tr}}_t$, is bounded. This ensures that the market is free of arbitrage but typically incomplete (if $d<n$)
 in the sense that $\mathcal{M}^e\neq  \emptyset$, as the minimal local martingale 
measure $\widehat{Q}$ given by $d\widehat{Q}=\mathcal{E}\left(-\xi\cdot W\right)dP$ \citep[see][]{Schweizer} is in $\mathcal{M}^e$, which however is typically not a singleton. Trading strategies are represented by the amount of 
wealth $\varphi=(\varphi^i_t)_i$ invested in the risky assets $(S^i)_i$. A self-financing trading strategy is described by a pair $(V_0,\varphi)$, where $V_0$ is the initial capital 
while $\varphi=(\varphi^i_t)_i$ describes the amount of wealth invested in the risky assets $(S^i)_i$ at any time $t$. The set $\Phi_\varphi$ of permitted strategies consists of 
$\R^d$-valued predictable processes $\varphi$  satisfying $E^P\left[\int_0^T\lvert\varphi^{\textrm{tr}}_t\sigma_t\rvert^2dt\right]<\infty.$ For a permitted strategy $\varphi$, 
the associated wealth process $V$ from initial capital $V_0$ has dynamics $dV_t = \varphi^{\textrm{tr}}_t\sigma_t d\widehat{W}_t$. To ease notation, we re-parametrize 
strategies in $\Phi_\varphi$ in terms of integrands $\phi:=\sigma^{\textrm{tr}}\varphi$ with respect to $\widehat{W}$. Indeed, equalities $\phi = \sigma^{\textrm{tr}}\varphi$ and 
$\varphi = (\sigma^{\textrm{tr}})^{-1}\phi$, where $(\sigma^{\textrm{tr}})^{-1}:=(\sigma\sigma^{\textrm{tr}})^{-1}\sigma$ is the pseudo-inverse of $\sigma^{\textrm{tr}}$, provide a 
one-to-one relation between $\varphi$ and $\phi$. Define the correspondences 
\begin{equation}\label{eq:CorrespondenceGammaandGammabot}
     \G_t(\omega) := \mathrm{Im}\ \sigma_t^{\textrm{tr}}(\omega)\ \textrm{ and }\ \G^\bot_t(\omega):= \mathrm{Ker}\ \sigma_t(\omega),\quad (t,\omega)\in [0,T]\times \Omega,
\end{equation}
where $\mathrm{Im}\ \sigma_t^{\textrm{tr}}$ and $\mathrm{Ker}\ \sigma_t$ denote the range (image) and the kernel of the respective matrices. Clearly, $\R^n =\G_t \oplus \G_t^\bot$ and any 
$z\in\R^n$ decomposes uniquely into its orthogonal projections as $z = \Pi_{\G_t}(z)\oplus \Pi_{\G^{\bot}_t}(z)=:\Pi_{t}(z)\oplus \Pi^\bot_{t}(z)$. Let 
\begin{equation*}
	\Phi = \Phi_\phi:= \left\lbrace \phi\ \Big\lvert\ \phi\ \textrm{is predictable},\ \phi\in \G\ \textrm{and}\  E\Big[\int_0^T\lvert\phi_t\rvert^2dt\Big]<\infty\Big.\right\rbrace
\end{equation*}
denote the (re-parametrized) set of permitted trading strategies. Proving the claims of the next proposition is routine (using \citet{Rockafellar} for Part~\ref{ImAndKerCorrespondencesPredictable}).
\begin{proposition}\label{pro:ImAndKerCorrespondencesPredictableCharactOfMe}
\hspace{2em}\begin{enumerate}
 \item\label{ImAndKerCorrespondencesPredictable} The correspondences $\G$, $\G^\bot$ are closed-convex-valued and predictable.
 \item \label{CharactOfMe} $Q$ is in $\mathcal{M}^e$ if and only if $Q\sim P$ with $dQ =\mathcal{E}(\lambda\cdot W)dP$, where $\lambda$ is predictable 
and $\lambda=-\xi +\eta$, with  $-\xi_t = \Pi_{t}(\lambda_t)\in\mathrm{Im}\,\sigma^{\textrm{tr}}_t$ and $\eta_t=\Pi^\bot_{t}(\lambda_t)\in\mathrm{Ker}\,\sigma_t$, $t\le T$.
 \end{enumerate}
\end{proposition}
By Part \ref{CharactOfMe} of Proposition \ref{pro:ImAndKerCorrespondencesPredictableCharactOfMe}, the set $\mathcal{Q}^{\mathrm{ngd}}$ defined in (\ref{eq:QngdDefinition}) can be written as
\begin{equation}\label{eq:QngdDefinitionlambda}
	\mathcal{Q}^{\mathrm{ngd}}=\Big\lbrace Q\sim P\ \Big\lvert\ dQ/dP\Big.=\mathcal{E}\left(\lambda\cdot W\right),\,\lambda\ \textrm{predictable, bounded and }\lambda\in \Lambda\Big\rbrace,
\end{equation}
where $\Lambda:[0,T]\times \Omega\rightsquigarrow \R^n$ is defined by 
\(
\Lambda_t(\omega) := C_t(\omega)\cap (-\xi_t(\omega)+\mathrm{Ker}\ \sigma_t(\omega)).   
\)
By Part~\ref{ImAndKerCorrespondencesPredictable} of Proposition~\ref{pro:ImAndKerCorrespondencesPredictableCharactOfMe} and \citep[][Cor.1.K and Thm.1.M]{Rockafellar}, $\Lambda$ is a 
compact-convex-valued predictable correspondence. Slightly beyond the no-free-lunch with vanishing risk condition, we assume that $\mathcal{Q}^{\mathrm{ngd}}$ contains the measure 
$\widehat{Q}$, or equivalently $-\xi\in C$. This implies that $\Lambda$ is non-empty valued, hence standard.

\subsection{Good-deal valuation with uniformly bounded correspondences}
\label{sec:NoGoodDealValuation_bdcorr}
We here consider the case where the no-good-deal restriction is described by a uniformly bounded correspondence;  a more general case is studied afterwards. We say that a correspondence 
$C$ is uniformly bounded if it satisfies  
\begin{assumption}\label{asp:BddCorr}
 $\sup_{(t,\omega)}\sup_{x\in C_t(\omega)}\abs{x}<\infty$.
\end{assumption}
Let $C:[0,T]\times\Omega\rightsquigarrow \R^n$ be a standard correspondence satisfying Assumption \ref{asp:BddCorr} and $0\in C$. Under Assumption \ref{asp:BddCorr}, selections of $C$ are uniformly 
bounded processes. In particular, the Girsanov kernels of no-good-deal measures are uniformly bounded, and hence boundedness in the definition (\ref{eq:QngdDefinition}) (see also (\ref{eq:QngdDefinitionlambda}))
of $\mathcal{Q}^{\mathrm{ngd}}$ is not necessary. The good-deal valuation bound $\pi^u_t(X):=\esssup_{Q\in\mathcal{Q}^{\mathrm{ngd}}}\ E^Q_t[X]$ is well-defined 
for a contingent claim   $X\in L^2\supset L^\infty$, that may be path-dependent, and one can check that in this case an analog of Part a) of Lemma \ref{lem:PropDynamicRiskMeasMeandQngdareMstable} still holds. 
Though Assumption \ref{asp:BddCorr} fits well with the classical theory, it would be too restrictive to impose it in general since it may not hold in some interesting practical situations; see for instance 
the example in Section \ref{subsec:GDVinStochVolModel}. Let us recall a fact about linear BSDEs \citep[cf.][]{KarouiPengQuenez} which explains their role for valuation purposes. 
\begin{lemma}\label{lem:ClosedMartandBSDEs}
	For $Q\sim P$ with bounded Girsanov kernel $\lambda$, the linear BSDE 
        \begin{equation}\label{eq:BSDEYlambda}
        	-dY_t = Z^{\textrm{tr}}_t\lambda_tdt -Z_t^{\textrm{tr}}dW_t,\: t\le T,\quad \text{with}\quad Y_T=X  \;\text{in}\; L^2,
        \end{equation}
has a unique standard solution $(Y^\lambda,Z^\lambda)$ with $Y^\lambda_t = E^Q_t[X] = Y^\lambda_0+Z\cdot W^Q_t,$ and
$W^Q:=W-\int_0^\cdot \lambda_tdt$. If $X\in L^\infty$ then $Y$ is bounded. 
\end{lemma}
Boundedness of $\lambda$ in Lemma \ref{lem:ClosedMartandBSDEs} clearly implies that the parameters of the BSDE \eqref{eq:BSDEYlambda} are standard. For unbounded $\lambda$, the classical 
BSDE theory no longer applies and one needs different results to characterize the good-deal bounds in terms of BSDEs. Under Assumption \ref{asp:BddCorr}, $\Lambda$ is uniformly bounded 
and thus Girsanov kernels $\lambda^Q$  for all $Q\in \mathcal{Q}^{\mathrm{ngd}}$ are bounded by the same constant. One has the following 

\begin{proposition}\label{pro:OptimizationGeneratorPiuasBSDEsolution}
 Let Assumption~\ref{asp:BddCorr} hold. For $X\in L^2$ and  $\lambda=\lambda^Q\in\Lambda$, $Q\in\mathcal{Q}^{\mathrm{ngd}}$, let $(Y^\lambda,Z^\lambda)$  and  $(Y,Z)$ be the standard solutions to 
  the classical BSDEs (\ref{eq:BSDEYlambda}) and 
	\begin{equation}\label{eq:BSDEPiu}
        	-dY_t = g(t,Z_t)dt -Z_t^{\textrm{tr}}dW_t,\ t\le T,\ \text{and}\ \quad Y_T=X,
        \end{equation}
with generator $g$ given by $g(t,\cdot,z) := \sup_{\lambda\in\Lambda}\lambda_t^{\text{tr}}(\cdot)z$ for $t\in[0,T]$,  $z\in\R^n$. Then: 
  \begin{enumerate}
  \item \label{OptimizationGenerator} There exists $\bar{\lambda}\in\Lambda$ such that $g(t,Z_t) = \bar{\lambda}_t^{\text{tr}}Z_t,\ P\otimes dt$-a.e..
  \item \label{PiuasBSDEsolution} $\pi^u_t(X) =\esssup_{Q\in\mathcal{Q}^{\mathrm{ngd}}}\ E^{Q}_t[X]= E^{\bar{Q}}_t[X] = Y_t$ holds for 
$\bar{Q}\in\mathcal{Q}^{\mathrm{ngd}}$ given by $d\bar{Q}=\mathcal{E}(\bar{\lambda}\cdot W)dP$, and $Y_t=\esssup_{\lambda\in\Lambda}\ Y_t^\lambda = Y^{\bar{\lambda}}_t, \ t\in[0,T]$.
 \end{enumerate}
\end{proposition}
\begin{proof}
Because $\Lambda$ is a predictable correspondence,
 by  measurable maximum and measurable selection results \citep[][Thms. 2.K and 1.C]{Rockafellar}, for each $z\in\R^n$ the function $g(\cdot,\cdot,z)$ is predictable and 
 there exists $\lambda^z\in\Lambda$ such that $g(t,\omega, z) = \lambda^z_t(\omega)z$ for all $(t,\omega)$. 
By Assumption \ref{asp:BddCorr}, the generators of the BSDEs (\ref{eq:BSDEYlambda}),(\ref{eq:BSDEPiu}) are uniformly Lipschitz in $z$. 
Being a Carath\'eodory function, 
 $g$ is $\Pred\otimes\mathcal{B}(\R^n)$-measurable. 
So the parameters (or data) for the BSDEs (\ref{eq:BSDEYlambda}),(\ref{eq:BSDEPiu}) are indeed standard.
Part \ref{OptimizationGenerator} follows again by the measurable maximum and measurable selection theorems, since $Z$ is predictable. For Part \ref{PiuasBSDEsolution}, the measure 
$\bar{Q}$ is in $\mathcal{Q}^{\mathrm{ngd}}$ since $\bar{\lambda}\in\Lambda$. The remainder of Part \ref{PiuasBSDEsolution} follows by existence, uniqueness  and comparison results for classical BSDEs, cf.\  \citep[][Sect.2-3]{KarouiPengQuenez}
\end{proof}
\subsection{Good-deal valuation with non-uniformly bounded correspondences}
\label{sec:NoGoodDealValuation_unbdcorr}
To relax the Assumption \ref{asp:BddCorr} of uniform boundedness, we now admit for a non-uniformly bounded 
standard  correspondence $C$, with $0\in C$, which satisfies 
\begin{equation}\label{eq:2IntCorr}
  \exists\ R  \text{ predictable with } \sup_{x\in C_t(\omega)}\lvert x\rvert \le R_t(\omega)<\infty \ \forall (t,\omega)\;\text{ and }  \int_0^T\lvert R_t\rvert^2dt<\infty. 
\end{equation}
It is relevant to look beyond Assumption~\ref{asp:BddCorr}, because examples of practical interest {require} to do so, see Section \ref{subsec:GDVinStochVolModel} where quasi-explicit 
formulas of good-deal bounds are obtained in a stochastic volatility model, with $C$ not being uniformly bounded but satisfying (\ref{eq:2IntCorr}). 
Classical BSDE results do not apply as before to characterize good-deal bounds directly by standard BSDE solutions. 
Yet, one can still \citep[][Thm.3.7]{KentiaPhD}
approximate $\pi^u_\cdot(X)$ for  $X\in L^\infty$ by solutions to classical BSDEs for suitable truncations of $C$, and prove that $\pi^u_t(X)$ coincides with the essential supremum 
over the larger set $\overline{\mathcal{Q}^{\text{ngd}}}\subseteq \mathcal{M}^e$ given in (\ref{eq:QngdBARDefinition}). Given (\ref{eq:2IntCorr}), we show that $\pi^u_\cdot(X)$ is 
the minimal supersolution of the BSDE (\ref{eq:BSDEPiu}), and that it is the minimal solution to  (\ref{eq:BSDEPiu}) if a worst-case measure $\bar{Q}$ for $\pi^u_0(X)$ exists. 
Obviously, a maximizing $\bar Q$ may be attained rather in the larger set $\overline{\mathcal{Q}^{\text{ngd}}}$. 

Let $g:[0,T]\times\Omega\times\R^n\to \R$ be the generator function defined by
\begin{equation}\label{eq:generatorsupersol}
g(t,\cdot,z) = \sup_{\lambda\in \Lambda}\lambda_t^{\textrm{tr}}(\cdot)z,\quad \text{for }t\in[0,T],\ z\in\R^n.
\end{equation}

By condition (\ref{eq:2IntCorr}), the function $g$ above is finitely defined since it satisfies $\lvert g(t,\omega,z)\rvert\le R_t(\omega) \lvert z\rvert<\infty$ for all $(t,\omega,z)$.
Note   for each $(t,\omega)$ that  $g(t,\omega,\cdot)$ is Lipschitz-continuous since $\Lambda_t(\omega)$ is compact. In addition by measurable selection arguments 
analogous to those in the proof of Prop.\ref{pro:OptimizationGeneratorPiuasBSDEsolution} it follows that $g$ is indeed $\Pred\otimes\mathcal{B}(\R^n)$-measurable.
Since $g$ may not be uniformly Lipschitz if $C$ does not satisfy Assumption \ref{asp:BddCorr}, then $\pi^u_\cdot(X)$ cannot directly be characterized by classical BSDEs. But one can still obtain a 
characterization by the minimal supersolution to the BSDE with data $(g,X)$. 
\begin{definition}\label{def:SupersolutionsBSDEs}
$(Y,Z,K)$  is a supersolution of the BSDE with parameters $(f,X)$ if 
\begin{equation*}
 -dY_t = f(t,Y_t,Z_t)dt - Z_t^{\textrm{tr}}dW_t+dK_t\quad \textrm{for }t\le T,\textrm{ and } Y_T=X,
\end{equation*}
with $K$ non-decreasing c\`adl\`ag adapted, $K_0=0$, and $\int_0^T\abs{Z_t}^2dt<\infty$. A supersolution with $K\equiv 0$ is a BSDE solution. A (super)solution 
$(Y,Z,K)$ 
is minimal 
if $Y_t\le \bar{Y}_t$, $t\in[0,T]$ holds for any other (super)solution $(\bar{Y},\bar{Z},\bar{K})$. 
\end{definition}
Note that a minimal supersolution when it exists is unique, as minimality implies uniqueness of the $Y$-components; since continuous local martingales of finite variation are trivial, 
identity of the $Z$- and $K$-components follows. Existence of the minimal supersolution is sometimes investigated under the condition that there exists at least one supersolution to the 
BSDE \citep[cf.][]{DrapeauHeyneKupper}. This condition is satisfied for the BSDE with parameters $(g,X)$, $X\in L^\infty$ since $g(\cdot,0)=0$ and thus $(Y,Z,K):=(\abs{X}_\infty-(\abs{X}_\infty-X)I_{\{T\}},0,(\abs{X}_\infty-X)I_{\{T\}})$ 
is a supersolution. Note that $g$ satisfies $g_t(z)\ge -\xi_t^\text{tr}z,\ P\otimes dt$-a.e.\ and moreover $(g,X)$ satisfies the hypotheses of \citep[][Thm.4.17]{DrapeauHeyneKupper} 
which implies existence of the minimal supersolution to the BSDE with parameter $(g,X)$. We show that $\pi^u_\cdot(X)$ can be identified with the $Y$-component of this minimal supersolution. 
Condition (\ref{eq:2IntCorr}) ensures that the process $\int_0^\cdot g_t(Z_t)dt$ for $g$ in (\ref{eq:generatorsupersol}) and $Z$ satisfying $\int_0^T\abs{Z_t}^2dt<\infty$ is real-valued, 
since the Cauchy-Schwarz inequality implies $\int_0^T\lvert g_t(Z_t)\rvert dt\le(\int_0^T\lvert Z_t\rvert^2dt)^{\frac{1}{2}}(\int_0^T\lvert R_t\rvert^2 dt)^{\frac{1}{2}}<\infty$.
\begin{theorem}\label{thm:MinimalSupersolGDB}
 Let (\ref{eq:2IntCorr}) hold and $X\in L^\infty$. There exists $Z\in\h^2(\widehat{Q})$ and a non-decreasing predictable process $K$ with $K_0=0$ such that $(\pi^u_\cdot(X),Z,K)$
is the minimal supersolution to the BSDE for data $(g,X)$ with $g$ from (\ref{eq:generatorsupersol}), and $\pi^u_\cdot(X)\in \mathcal{S}^\infty$. 
\end{theorem}
The proof for this theorem is given in Appendix~\ref{app:AppendixA}, while for details on the proof of the next corollary we refer to \cite[][Corollary 3.10]{KentiaPhD}.

\begin{corollary}\label{cor:GDBsolvesBSDEunderOptimalMeas}
Let (\ref{eq:2IntCorr}) hold and $X\in L^\infty$. If there exists a measure $\bar{Q}\in \overline{\mathcal{Q}^{\mathrm{ngd}}}$ such that $\pi^u_0(X)=\sup_{Q\in \mathcal{Q}^{\mathrm{ngd}}}E^Q[X] = E^{\bar{Q}}[X]$,
then $\pi^u_\cdot(X)$ is a $\bar{Q}$-martingale and there exists $Z\in\h^2(\widehat{Q})$ such that $(\pi^u_\cdot(X),Z)$ is the minimal solution to the BSDE with parameters $(g,X)$ for 
$g$ defined in (\ref{eq:generatorsupersol}). The Girsanov kernel $\bar{\lambda}$ of $\bar{Q}$ satisfies $\esssup_{\lambda\in \Lambda}\lambda_t^{\textrm{tr}}Z_t =\bar{\lambda}_t^{\textrm{tr}}Z_t,\ P\otimes dt\text{-a.e.}$. 
\end{corollary}

In concrete cases, existence of $\bar{Q}\in \overline{\mathcal{Q}^{\mathrm{ngd}}}$ as in Corollary \ref{cor:GDBsolvesBSDEunderOptimalMeas} may be shown by direct considerations, see  
  Section \ref{subsec:ExampleNGDValuations} for examples. If one  specifies the no-good-deal restriction ($C$) such that the set $\overline{\mathcal{Q}^{\mathrm{ngd}}}$ becomes weakly compact in $L^1$, then
  $\bar{Q}$ would exist for any $X\in L^\infty$ as maximizer of a  bounded linear functional over a weakly compact subset of $L^1$. 
  Note that Assumption \ref{asp:BddCorr} only implies (by Dunford-Pettis theorem) that $\overline{\mathcal{Q}^{\mathrm{ngd}}}$ is weakly 
relatively compact in $L^1$. Yet if $\mathcal{Q}^{\mathrm{ngd}}$ is not weakly relatively compact in $L^1$, then by James' theorem \citep[cf.][Thm.6.36]{AliprantisBorder} 
there exists $X\in L^\infty$ such that the supremum in $\pi^u_0(X)=\sup_{Q\in \mathcal{Q}^{\mathrm{ngd}}}E^Q[X]$ is not attained in the $L^1$-closure of 
$\overline{\mathcal{Q}^{\mathrm{ngd}}}$ (note that $\overline{\mathcal{Q}^{\mathrm{ngd}}}$ is convex), and in particular not in $\overline{\mathcal{Q}^{\mathrm{ngd}}}$. Let us give an example where $\bar{Q}$ does not exist in $\overline{\mathcal{Q}^{\mathrm{ngd}}}$ for some contingent claim and $C$ does neither 
satisfy Assumption \ref{asp:BddCorr} nor (\ref{eq:2IntCorr}). Section \ref{subsec:GDVinStochVolModel} will furthermore give an example in a stochastic volatility model where $\bar{Q}$ 
exists and $C$ is not uniformly bounded but satisfies (\ref{eq:2IntCorr}). 
\begin{example}
Let $n=2$ with $W=(W^1,W^2)$, $d=1$ with $dS_t=S_t\sigma^SdW^1_t,\ S_0>0$, $\sigma^S>0$, and $\xi=0$. Let $h>0$ be a deterministic predictable process with $\int_0^Th_tdt=\infty$ and 
$C_t(\omega):=\{0\}\times[-h_t,h_t],\ (t,\omega)\in [0,T]\times\Omega$. Now let $X:=I_{\{W^2_T\ge 0\}}\in L^\infty$, then $\pi_0:=\sup_{n\in\N}Q^n[\{W^2_T\ge 0\}]\le\pi^u_0(X)\le 1$, 
where $dQ^n=\mathcal{E}(\lambda^n\cdot W^2)dP$ with $\lambda^n_t=h_t\wedge n,\ t\in [0,T],\ n\in\N$. The process $W^{2,n}:=W^2-\int_0^\cdot\lambda^n_tdt$ is a $Q^n$-Brownian motion. 
Hence $W^{2,n}_T\sim \mathcal{N}(0,T)$ under $Q^n$. We have $\int_0^T\lambda^n_tdt\nearrow \int_0^Th_tdt=\infty$ as $n\nearrow \infty$. Hence $\pi_0=\sup_{n\in\N}Q^n[\{W^{2,n}_T\ge -\int_0^T\lambda^n_tdt\}]=1$. 
Therefore $\pi^u_0(X)=1$. But there exists no measure $\bar{Q}\in \overline{\mathcal{Q}^{\mathrm{ngd}}}$ such that $\pi^u_0(X)=E^{\bar{Q}}_0[X]$. Indeed for such a measure, one would have $\bar{Q}[\{W^2_T\ge0\}]=1$ 
which is not possible since $\bar{Q}\sim P$.
\end{example}

\section{Dynamic good-deal hedging}\label{sec:No-Good-Deal-HedgingApproach}
Let again $C$ be a standard correspondence satisfying $0\in C$, and define the family of a-priori valuation measures 
\begin{equation}\label{eq:PngdDefinition}
	\mathcal{P}^{\mathrm{ngd}}:=\Big\lbrace Q\sim P\ \Big\lvert\ dQ/dP\Big.=\mathcal{E}\left(\lambda\cdot W\right),\ 
                             \lambda\ \textrm{predictable, bounded, } \lambda\in C\Big\rbrace 
\end{equation}
which satisfy the same no-good-deal constraint as those in $\mathcal{Q}^{\mathrm{ngd}}$, except that the local martingale condition for $S$ is omitted. 
\begin{remark}
One could view $\mathcal{P}^{\mathrm{ngd}}$ as the no-good-deal measures for a market consisting only of the riskless asset $S^0\equiv 1$, i.e.\  $\mathcal{P}^{\mathrm{ngd}}=\mathcal{Q}^{\mathrm{ngd}}(1)$. 
It is natural to define (\ref{eq:PngdDefinition}) as a-priori valuation measures, as the idea of no-good-deal valuation is to consider those risk neutral valuation measures $Q$, 
for which any extension of the financial market by additional derivatives' price processes (being $Q$-martingales) would not give rise to 'good deals'; see e.g.\ \citet{BjorkSlinko,KloeppelS07,Becherer} 
for rigorous detail in continuous time for Sharpe ratios, utilities or growth rates; for concepts cf.\ \citep[][]{Cerny}. 
\end{remark}
Like $\mathcal{Q}^{\mathrm{ngd}}$, the set $\mathcal{P}^{\mathrm{ngd}}$ clearly is again  m-stable and convex. We define the a-priori dynamic coherent risk measure (in the sense of Lemma~\ref{lem:PropDynamicRiskMeasMeandQngdareMstable}) 
\begin{equation}\label{eq:DefRhoPngdB}
	\rho_t(X):=\esssup_{Q\in\mathcal{P}^{\mathrm{ngd}}}\ E^Q_t[X],\quad t\in[0,T],
\end{equation}
for suitable contingent claims $X$  for which $\rho_t(X)$ is finitely defined (e.g.\ $X$ bounded, or just in $L^2$ with $C$ satisfying Assumption \ref{asp:BddCorr}).
For a bounded correspondence $C$, 
one can describe $\rho(X)$ (analogous to $\pi^u(X)$ in Prop.\ref{pro:OptimizationGeneratorPiuasBSDEsolution}) by  classical 
BSDEs: 
\begin{proposition}\label{pro:RhoasBSDESolution}
	Let Assumption \ref{asp:BddCorr} hold. For $X\in L^2$, let $(\widetilde{Y},\widetilde{Z})$ and $(Y^\lambda,Z^\lambda)$ (for $\lambda\in C$) be the respective standard solutions to the BSDEs
	\begin{align*}
			-dY_t &= Z^{\textrm{tr}}_t\tilde{\lambda}_tdt -Z_t^{\textrm{tr}}dW_t,\ t\le T,\ \text{with}\ Y_T=X,\quad \text{and}\\
			-dY_t &= Z^{\textrm{tr}}_t\lambda_tdt -Z_t^{\textrm{tr}}dW_t,\ t\le T,\ \text{with}\ \ Y_T=X,
	\end{align*}
where $\tilde{\lambda}\in C$ is a predictable process satisfying $\tilde{\lambda}^{\textrm{tr}}_tZ_t=\esssup_{\lambda\in C}\lambda^{\textrm{tr}}_tZ_t,\ P\otimes dt\text{-a.e.}$.
Then the measure $\widetilde{Q}$ with Girsanov kernel $\lambda^{\widetilde{Q}}=\tilde{\lambda}$ is in $\mathcal{P}^{\mathrm{ngd}}$, and $\displaystyle \rho_t(X) = \esssup_{\lambda\in C}Y_t^\lambda =E^{\widetilde{Q}}_t[X]= \tilde{Y}_t,\ t\in [0,T]$. 
\end{proposition}
Elements $Q$ of $\mathcal{P}^{\mathrm{ngd}}$ or $\mathcal{Q}^{\mathrm{ngd}}$  could be seen as generalized scenarios 
\citep[as in][]{ArtznerDelbaenEberHeathKu} for the dynamic coherent risk measures $\pi^u$ or $\rho$ (cf.\ Lemma \ref{lem:PropDynamicRiskMeasMeandQngdareMstable}).
By $\mathcal{P}^{\mathrm{ngd}}\cap\mathcal{M}^e = \mathcal{Q}^{\mathrm{ngd}}$ one has $\rho_t(X)\ge \pi^u_t(X)$ 
for $t\le T$. An investor holding a liability $X$ and trading in the market according to a permitted trading strategy $\phi$ would assign  at time $t$ a residual risk 
$\rho_t (X-\int_t^T\phi^{\textrm{tr}}_sd\widehat{W}_s)$ to his position. The investor's objective is to hedge his position by a trading strategy $\bar{\phi}$ that minimizes his residual risk 
at any time $t\le T$. To justify a premium $\pi^u_\cdot(X)$  for selling $X$, the minimal capital requirement to make his position $\rho$-acceptable should coincide with 
$\pi^u_\cdot(X)$. Thus, his hedging problem is to  find a strategy $\bar{\phi}\in\Phi$ such that 
\begin{equation}\label{eq:HedgingProblem}
	\pi^u_t(X)=\rho_t\Big(X-\int_t^T\bar{\phi}^{\textrm{tr}}_sd\widehat{W}_s\Big)=\essinf_{\phi\in\Phi}\ \rho_t\Big(X-\int_t^T\phi^{\textrm{tr}}_sd\widehat{W}_s\Big),\quad t\in[0,T].
\end{equation}
The good-deal hedging strategy will be defined as a minimizer $\bar{\phi}$ in (\ref{eq:HedgingProblem}); If such a minimizer exists (as e.g. in \cite{Becherer} or in  later sections) the good-deal valuation $\pi^u_\cdot(\cdot)$ coincides with the market 
consistent risk measure corresponding to $\rho$, in the spirit of \citet{BarrieuElKaroui}. For contingent claim $X$, the tracking error $R^\phi_t(X)$ of a strategy $\phi\in\Phi$ 
 is defined as the difference between the dynamic variations in the capital requirement and the profit/loss from trading (hedging) according to $\phi$, i.e.\
\(
	  	  R^\phi_t(X):=\pi^u_t(X)-\pi^u_0(X)- \phi\cdot \widehat{W}_t, \; t\in [0,T].
\)
\begin{proposition}\label{pro:RobustTrackingError}
Let $X \in L^2$ and $C$ satisfy Assumption \ref{asp:BddCorr}.
Then the tracking error $R^{\bar{\phi}}(X)$ for a strategy $\bar{\phi}\in\Phi$ solving (\ref{eq:HedgingProblem}) is a $Q$-supermartingale 
for all $Q\in \mathcal{P}^{\textrm{ngd}}$.
\end{proposition}
\begin{proof}
By the first equality of (\ref{eq:HedgingProblem}) and the definition of the tracking error one has $R^{\bar{\phi}}_t(X)=-\pi^u_0(X)+\rho_t\big(X-\int_0^T\bar{\phi}^{\textrm{tr}}_sd\widehat{W}_s\big),\ t\in [0,T]$. 
Note that $\int_0^T\bar{\phi}^{\text{tr}}_sd\widehat{W}_s\in L^2(P)$ holds since $\xi$ is bounded and $\bar{\phi}$ is in $L^2(P\otimes dt)$. Now let $Q\in \mathcal{P}^{\textrm{ngd}}$ with $dQ/dP = \mathcal{E}(\lambda^Q\cdot W)$.   
By Proposition \ref{pro:RhoasBSDESolution} this implies after a change of measures that 
$-dR^{\bar{\phi}}_t(X) =\big(\esssup_{\lambda\in C}Z^{\textrm{tr}}_t\lambda_t -Z^{\textrm{tr}}_t\lambda^Q_t\big)dt-Z_t^{\textrm{tr}}dW^Q_t,\ t\in[0,T].$
Because $R^{\bar{\phi}}(X)\in \mathcal{S}^2(P)$ and $dQ/dP\in L^p(P)$ for all $p<\infty$ (since $\lambda^Q$ is bounded), then H\"older's inequality implies that 
$R^{\bar{\phi}}(X)\in \mathcal{S}^{2-\epsilon}(Q)$ for $\epsilon\in(0,1)$. Furthermore under $Q$ the finite variation part of $R^{\bar{\phi}}_\cdot(X)$ is non-increasing, and hence 
$R^{\bar{\phi}}(X)$ is a $Q$-supermartingale.
\end{proof}
\begin{remark}\label{rem:atleastmeanselffinancing}
For a self-financing strategy $\phi\in \Phi$ replicating $X=x_0+\int_0^T\phi^{\textrm{tr}} d\widehat{W}$, with $x_0\in \mathbb{R}$,  the tracking error vanishes, i.e.\ $R^{{\phi}}(X)=0$. One says that a strategy is mean-self-financing 
(like risk minimizing strategies studied in \citet[][]{Schweizer}, Sect.2, with $E^{\widehat{Q}}_t[X]$ taking the role of $\pi^u_t(X)$) if its tracking 
error is a martingale (under $P$). Hence by Proposition  \ref{pro:RobustTrackingError}, one can view the good-deal hedging strategy as being ``at least mean-self-financing'' under any 
$Q\in \mathcal{P}^{\textrm{ngd}}$, since its tracking error is a supermartingale under all measures in $\mathcal{P}^{\textrm{ngd}}$.
Holding uniformly over all measures in $\mathcal{P}^{\textrm{ngd}}$, this could be seen as a robustness property of $\bar{\phi}$.  
\end{remark}
To describe solutions to the hedging problem (\ref{eq:HedgingProblem}), we will often assume that $C$ has further structure and is uniformly bounded. Section \ref{subsec:GDVinStochVolModel} also contains an example 
with a  semi-explicit solution to the hedging problem in the Heston model
 for a correspondence $C$ that is not uniformly bounded but satisfies (\ref{eq:2IntCorr}). 

\subsection{Results for ellipsoidal no-good-deal constraints}
\label{subsec:Example1withEllipsoids}
 
This section derives more explicit BSDE results to describe the solution to the valuation and the hedging problem (\ref{eq:HedgingProblem}) for  (predictable) ellipsoidal no-good-deal constraints. This generalizes the important special case of  radial  constraints \citep[as e.g.\ in][]{Becherer}, which is common to the good-deal literature and  justified by bounds (uniform in $(t,\omega)$) on optimal growth rates or instantaneous Sharpe ratios, while still permitting comparably explicit results.  The generalization could
 be interpreted as imposing different bounds on growth rates (or Sharpe ratios) for the risk factors 
associated to the principal axes. While such might appear as technical at this stage,  in the subsequent context of  model uncertainty (cf.\ Remark~\ref{Rem:C0plusTheta}~b)) non-radial 
constraints will appear naturally.

To this end, let $h$ be a positive bounded predictable process, and $A$ be a predictable $\R^{n\times n}$-matrix-valued process with symmetric values and uniformly elliptic i.e.\  $A^\text{tr}=A\text{ and } x^{\textrm{tr}}A x\ge c \abs{x}^2$ 
for all $x\in\R^n$, for some $c\in (0,\infty)$. The latter is a tractable and convenient (not necessarily most general) 
generalization of common good-deal constraints in the radial case, where $A\equiv{\mathrm{Id}}_{\R^n}$. It can be seen as setting  constraints by scalar bounds on the \cite{Mahalanobis36}-distance  of Girsanov kernels for the no-good-deal  measures,  instead of 
on their Euclidian norm.  

We define the standard \citep[see][Cor.1.Q]{Rockafellar} correspondence
\begin{equation}\label{Constrellip}
	C_t(\omega) = \left\lbrace x\in \R^n\ \lvert\ x^{\textrm{tr}}A_t(\omega)x\le h^2_t(\omega)\right\rbrace,\quad (t,\omega)\in[0,T]\times\Omega,
\end{equation}
that satisfies Assumption \ref{asp:BddCorr} due to ellipticity and boundedness of $h$. Assume that the kernel of the volatility matrix $\sigma$ is spanned by eigenvectors of $A$, i.e.\ 
        \begin{equation}\label{eq:AssumptionOptimizationSimplify}
        	A^{-1}_t(\mathrm{Ker}\ \sigma_t) = \mathrm{Ker}\ \sigma_t,\quad t\in[0,T].
        \end{equation}
As the eigenvectors of $A$ are orthogonal and $(\mathrm{Ker}\ \sigma)^{\bot}=\mathrm{Im}\ \sigma^{\textrm{tr}} $, then (\ref{eq:AssumptionOptimizationSimplify}) can be interpreted as 
separability of $\mathrm{Im}\ \sigma^{\textrm{tr}}$ and $\mathrm{Ker}\ \sigma$ in the sense that each of these subspaces has a basis of eigenvectors of $A$. Given (\ref{eq:AssumptionOptimizationSimplify}), 
the subspaces $\mathrm{Im}\ \sigma^{\textrm{tr}}$ and $\mathrm{Ker}\ \sigma$ are orthogonal 
under the scalar product defined by $A$, one can re-write 
\begin{equation*}
	\mathcal{Q}^{\mathrm{ngd}}=\Big\lbrace Q\sim P\ \Big\lvert\ dQ/dP\Big.=\mathcal{E}\left(\lambda\cdot W\right),\ 
                             \lambda\textrm{ predictable, }\lambda=-\xi+\eta,\textrm{ } \eta\in C^\xi\cap \mathrm{Ker}\ \sigma  \Big\rbrace,
\end{equation*}
with 
$
	C^\xi_t(\omega) = \left\lbrace x\in\R^n\ \lvert\ x^{\textrm{tr}}A_t(\omega)x\le h^2_t(\omega)-\xi_t(\omega)^{\textrm{tr}}A_t(\omega)\xi_t(\omega)\right\rbrace,
$
also satisfying Assumption \ref{asp:BddCorr}. The correspondence $C^\xi$ is standard if 
     \begin{equation}\label{eq:WeakerAssumption}
     	h^2 > \xi^{\textrm{tr}}A\xi.
     \end{equation}
The separability condition (\ref{eq:AssumptionOptimizationSimplify}) ensures that $-\xi+\eta\in C$ is equivalent to $\eta\in C^\xi$, for $\eta\in \mathrm{Ker}\ \sigma$. This way the 
ellipsoidal constraint on the Girsanov kernels transfers to one on their $\eta$-component, which permits to formulate the no-good-deal restriction only with respect to non-traded risk 
factors in the market. In this setup, it is straightforward to obtain an expression for $\bar{\lambda}$ from Part \ref{OptimizationGenerator} of Proposition \ref{pro:OptimizationGeneratorPiuasBSDEsolution} 
via 
\begin{lemma}\label{lem:OptimizationGeneratorApplication}
	For $z\in \R^n\setminus\{0\}$, $h>0$ and a symmetric positive definite $n\times n$-matrix $A$, the unique maximizer of $y^{\text{tr}}z$ subject to $y^{\textrm{tr}}Ay\le h^2$
	is $\bar{y}=h(z^{\textrm{tr}}A^{-1}z)^{-1/2}A^{-1}z$.
\end{lemma}
For $X\in L^2$ with $C$ satisfying Assumption \ref{asp:BddCorr}, the classical BSDE 
\begin{align}        	
-dY_t &= \Big(-\xi_t^{\textrm{tr}}\Pi_t(Z_t) + \sqrt{h_t^2-\xi_t^{\textrm{tr}}A_t\xi_t}\sqrt{{\Pi_t^\bot(Z_t)}^\textrm{tr}A^{-1}_t\Pi_t^\bot(Z_t)}\Big)dt
-Z_t^{\textrm{tr}}dW_t
\label{eq:BSDEPiuApplicationfinal}
\end{align}
with terminal condition $Y_T=X$ has a unique standard solution $(Y,Z)$.\\
We will see that $\pi^u_\cdot(X)=Y$ holds and that the optimal 
kernel $\bar{\lambda}$ from  Proposition \ref{pro:OptimizationGeneratorPiuasBSDEsolution}, Part \ref{OptimizationGenerator}, 
takes the form $\bar{\lambda}=-\xi+\bar{\eta}$, with $\bar{\eta}\in\mathrm{Ker}\ \sigma$ given by 
        \begin{equation}\label{eq:EquationEta}
		\bar{\eta}_t =\sqrt{\frac{{h_t^2-\xi_t^{\textrm{tr}}A_t\xi_t}}{{{\Pi_t^\bot(Z_t)}^\textrm{tr}A^{-1}_t\Pi_t^\bot(Z_t)}}}A^{-1}_t\Pi_t^\bot(Z_t)\,,\quad  t\in [0,T].
	\end{equation}
By Lemma \ref{lem:OptimizationGeneratorApplication} and  (\ref{eq:AssumptionOptimizationSimplify}), $\bar{\eta}_t^{\textrm{tr}}\Pi^\bot_t(Z_t)=\esssup_{\eta_\cdot\in C^\xi_\cdot\cap\mathrm{Ker}\ \sigma_\cdot} \eta_t^{\textrm{tr}}\Pi^\bot_t(Z_t)$, $P\otimes dt\text{-a.e.}$,
thus $\bar{\lambda}^{\textrm{tr}}_tZ_t = -\xi_t^{\textrm{tr}}\Pi_t(Z_t)+\big(h_t^2-\xi_t^{\textrm{tr}}A_t\xi_t\big)^{1/2}\big({\Pi_t^\bot(Z_t)}^\textrm{tr}A^{-1}_t\Pi_t^\bot(Z_t)\big)^{1/2}$, $P\otimes dt\text{-a.e.}$.
Therefore Part \ref{PiuasBSDEsolution} of Proposition \ref{pro:OptimizationGeneratorPiuasBSDEsolution} yields
\begin{theorem}\label{thm:NGDboundasBSDEsol}
Assume (\ref{eq:AssumptionOptimizationSimplify}) and (\ref{eq:WeakerAssumption}) hold. For $X\in L^2$, let $(Y,Z)$ be the standard solution to the 
BSDE (\ref{eq:BSDEPiuApplicationfinal}). Then $\pi_t^u(X)=Y_t=E^{\bar{Q}}_t[X],\ t\in[0,T]$, where $d\bar{Q}=\mathcal{E}\left((-\xi+\bar{\eta})\cdot W\right)dP$ with $\bar{\eta}$ given explicitly by (\ref{eq:EquationEta}).
\end{theorem}
One can interpret the dynamics of the no-good-deal valuation   (\ref{eq:BSDEPiuApplicationfinal}) 
as follows. Writing it as
 $dY_t = -a_t dt + \xi_t^{\text{tr}}\Pi_t(Z_t)\,dt+ Z_t^{\text{tr}} dW_t= -a_t dt + \Pi_t(Z_t)^{\text{tr}} d\widehat{W}_t +  \Pi_t^\bot(Z_t)^{\text{tr}}d\widehat{W}_t$ (cf.\ Section~\ref{sec:Parameterization}), with 
 $a_t:= \sqrt{h_t^2-\xi_t^{\textrm{tr}}A_t\xi_t}\sqrt{{\Pi_t^\bot(Z_t)}^\textrm{tr}A^{-1}_t\Pi_t^\bot(Z_t)}$, it decomposes 
 into a hedgeable part $\Pi_t(Z_t)^{\text{tr}}(\xi_t dt+dW_t)= \Pi_t(Z_t)^{\text{tr}}d\widehat{W}$, that is dynamically spanned by tradeable assets, 
an orthogonal part  $\Pi_t^\bot(Z_t)^{\text{tr}}d\widehat{W}$, being a martingale under $P$ (and $\widehat{Q}$), and
a premium part, where the rate $a_t\ge 0$ inherent to the upper good deal bound can be seen as compensation for the seller of the claim for non-tradeable risk. Note that $a>0$ on $\{(\omega,t):\Pi^\bot(Z_t)\neq 0\}$ by (\ref{eq:WeakerAssumption}). 

The observation of the following lemma is straightforward.
\begin{lemma}\label{lem:qPrimePositiveDefinitePtwise}
	The matrices $A_t^{-1}(\omega)$, for $(t,\omega)\in [0,T]\times\Omega$, are positive-definite 
	and satisfy $x^{\textrm{tr}}A^{-1}_t(\omega)x \ge \alpha'_t(\omega) \abs{x}^2$ for all $x,t$,
where $\alpha'_t(\omega) = c\lVert A_t(\omega)\rVert^{-2}>0$ for $c$ being the constant of uniform ellipticity of $A$. 
Moreover $\lVert A\rVert \ge c$ holds.
\end{lemma}
By Lemma \ref{lem:OptimizationGeneratorApplication}, the process $\tilde{\lambda} = h(Z^\textrm{tr}A^{-1}Z)^{-1/2}\ A^{-1}Z$ satisfies 
$\tilde{\lambda}^{\textrm{tr}}_tZ_t = \esssup_{\lambda_\cdot^{\textrm{tr}}A_\cdot\lambda_\cdot\le h^2_\cdot}\ \lambda^{\textrm{tr}}_tZ_t=h_t(Z_t^\textrm{tr}A^{-1}_tZ_t)^{1/2}$, $P\otimes dt\text{-a.e.}$. 
Hence Proposition \ref{pro:RhoasBSDESolution} gives $\rho_t(X)=Y_t,\ t\in [0,T]$, where $(Y,Z)$ uniquely solves the classical BSDE with terminal condition $ Y_T=X$ and
	\begin{equation}\label{eq:BSDErhoApplication}
		 -dY_t = h_t(Z_t^\textrm{tr}A^{-1}_tZ_t)^{1/2}dt -Z_t^{\textrm{tr}}dW_t.
	\end{equation}
Thanks to Lemma \ref{lem:qPrimePositiveDefinitePtwise}, a sufficient condition to ensure (\ref{eq:WeakerAssumption}) is
        \begin{equation}\label{eq:AssumptionOnAlphaPrime}
     	\abs{\xi}<h\sqrt{\alpha'}.
        \end{equation}
In addition (\ref{eq:AssumptionOnAlphaPrime}) is used to verify for Lemma \ref{lem:OptimizationWithPhiSolved} (stated in the appendix) the Kuhn-Tucker conditions 
before applying the Kuhn-Tucker theorem \citep[see][Section 28]{RockafellarConvexAnalysis}, after which comparison results for BSDE yield the result of 
Theorem \ref{thm:FullHedgingProblemSolved} below. The proof is omitted as it is analogous to that of \citep[][Thm.5.4]{Becherer}, 
using now Lemma~\ref{lem:OptimizationWithPhiSolved} instead of Lem.6.1 therein. 
For $\phi\in\Phi$, let $(Y^\phi,Z^\phi)$ denote the standard solution to the BSDE with terminal condition $Y_T=X$ and, for $t\le T$,
\begin{align}
 		-dY_t &= \Big(-\xi_t^{\textrm{tr}}\phi_t +h_t\big(\left(Z_t-\phi_t\right)^{\textrm{tr}}A^{-1}_t\left(Z_t-\phi_t\right)\big)^{1/2}\Big)dt -Z_t^{\textrm{tr}}dW_t.
\label{eq:BSDEforYphi}
	\end{align}
\begin{theorem}\label{thm:FullHedgingProblemSolved}
 Assume (\ref{eq:AssumptionOptimizationSimplify}),(\ref{eq:AssumptionOnAlphaPrime}) hold. For $X\in L^2$, let $(Y,Z)$ and $(Y^\phi,Z^\phi)$ (for $\phi\in\Phi$) be standard solutions to the 
 BSDEs (\ref{eq:BSDEPiuApplicationfinal}),(\ref{eq:BSDEforYphi}). Then $Y^\phi_t=\rho_t (X-\int_t^T\phi^{\textrm{tr}}_sd\widehat{W}_s)$, $ t\le T$, and the strategy
	\begin{equation}\label{eq:PhibarExpression}
 		\bar{\phi}_t = \frac{\sqrt{\Pi_t^\bot(Z_t)^{\textrm{tr}}A^{-1}_t\Pi_t^\bot(Z_t)}}{\sqrt{h_t^2-\xi_t^{\textrm{tr}}A_t\xi_t}}\ A_t\xi_t + \Pi_t(Z_t)
	\end{equation}
is in $\Phi$ and satisfies $\displaystyle Y^{\bar{\phi}}_t = \essinf_{\phi\in\Phi}\ Y^{\phi}_t = Y_t$ for any $t\in [0,T]$, that is  
\begin{equation*}
	\pi^u_t(X) =\essinf_{\phi\in\Phi}\ \rho_t\Big(X-\int_t^T\phi^{\textrm{tr}}_sd\widehat{W}_s\Big)= \rho_t\Big(X-\int_t^T\bar{\phi}^{\textrm{tr}}_sd\widehat{W}_s\Big)=Y^{\bar{\phi}}_t.
\end{equation*}
Moreover, the tracking error $R^{\bar{\phi}}(X)$ is a supermartingale under all measures $Q\in \mathcal{P}^{\textrm{ngd}}$ and a martingale under the measure $Q^\lambda\in \mathcal{P}^{\textrm{ngd}}$ 
with Girsanov kernel $\lambda_t := h_t\big((Z_t-\bar{\phi}_t)A^{-1}_t(Z_t-\bar{\phi}_t)\big)^{-1/2} A_t^{-1}\left(Z_t-\bar{\phi}_t\right)$, $t\in [0,T].$
\end{theorem}
The summands in the expression (\ref{eq:PhibarExpression}) for the strategy $\bar{\phi}$  play different roles from the perspective of hedging. 
The second summand
is a {\em non-speculative} component that hedges locally tradeable risk by replication, while the first is a {\em speculative} component that  compensates (``hedges'') for unspanned  
non-tradeable risk by taking favorable bets on the market price of risk. Clearly, good deal bounds fit into  the rich theory
of $g$-expectations and market-consistent risk measures \citep[cf.][and more references therein]{BarrieuElKaroui}. See  \citet{Leitner} for closely related ideas about  instantaneous  measurement of risk.  

\subsection{Examples for good-deal valuation and hedging with closed-form solutions}\label{subsec:ExampleNGDValuations}
Explicit formulas, if available, facilitate intuition and enable fast computation of valuations, hedges and comparative statics. To this end, several concrete case studies are provided, starting with European options 
with monotone payoff profiles (e.g.\  call options)  on non-traded assets in a multidimensional model of Black-Scholes type, in which tradeable assets only permit for partial hedging.
In parallel to \citep[][Prop.3, Sect.5.3]{CarassusTemam} and \citet{BayraktarYoung}, who employ SDE respectively
PDE methods, this demonstrates how 
previous BSDE analysis can be applied in concrete case studies and we contribute some slight generalizations as well (e.g.\ higher dimensions, ellipsoidal 
constraints). As a further example,  we  contribute new explicit 
formulas for an option to exchange (geometric averages of)  non-traded assets into traded assets. As before, the no-good-deal approach here gives rise to a familiar option pricing 
formula (by Margrabe) but suitable adjustments of parameter inputs are required, showing the difference to a simple no-arbitrage valuation approach that uses only one (given) single risk neutral measure.
A further  example derives semi-explicit good-deal solutions for the stochastic volatility model by Heston, for no-good-deal constraints on market prices of (unspanned) stochastic volatility 
risk which impose an interval range on the  mean reversion level of the stochastic variance process under any valuation measure $Q \in Q^{{\rm ngd}}$. Technically, this corresponds to 
imposing bounds on the instantaneous Sharpe ratio which are inversely proportional to the stochastic volatility.
This is different to  a related result by \citet{BondarenkoLongarela},
in that their example imposes no good deal  constraints in terms of  bounds on simultaneous changes in the level of mean-reversion combined with opposite changes in reversion speed. We emphasize that, in addition to valuation 
formulas, all our examples provide explicit results for good-deal  hedging strategies as well.
Detailed derivations of the formulas in Sections \ref{subsec:OptionOnNonTradedAsset}-\ref{subsec:GDVinStochVolModel} are given in Appendix~\ref{app:AppendixA}

\subsubsection{Closed-form formulas for options in a generalized Black-Scholes model}\label{subsec:OptionOnNonTradedAsset}
\label{sec:CFformulas}
The market information $\filt=\filtration$ is generated by an $n$-dimensional $P$-Brownian motion $W:=(W^1,\ldots,W^n)^{\text{tr}}$ with $W^S=(W^1,\ldots,W^d)^{\text{tr}}$, $d< n$ for $n,d\in\N$, 
and is augmented by null-sets. The  financial market consists of $d\le n$ (incomplete if $d<n$) stocks with (discounted) prices $S=(S^k)_{k=1}^d$ and further $n-d$ non-traded assets with 
values $H=(H^l)_{l=1}^{n-d}$. We consider a risk neutral model ($P=\widehat{Q}\in\mathcal{M}^e,\ \xi=0$) 
where the processes $S$ and $H$ evolve as
\begin{equation*}
	dS_t=\textrm{diag}(S_t)\sigma^S dW^S_t\quad \text{and}\quad dH_t=\textrm{diag}(H_t)\big(\gamma dt + \beta dW_t\big), \quad  t\in [0,T],
\end{equation*}
with $S_0\in(0,\infty)^d,\ H_0\in (0,\infty)^{n-d}$, constant coefficients $\sigma^S=(\sigma^S_{ki})_{k,i}\in \R^{d\times d}$ invertible, $\beta=(\beta_{li})_{l,i}\in\R^{{(n-d)}\times n}$ and $\gamma\in \R^{n-d}$.
The volatility matrix of $S$ is $\sigma := (\sigma^S,0)\in\R^{d\times n}$ and  is clearly of maximal rank $d\le n$. For $z\in\R^n$, we have ${\Pi(z)=(z^1,\ldots,z^d,0,\ldots,0)^{\textrm{tr}}\in\R^n}$ 
and $\Pi^{\bot}(z)=\left(0,\ldots,0,z^{d+1},\ldots,z^{n}\right)^{\textrm{tr}}\in\R^n$. We assume the ellipsoidal framework of Section \ref{subsec:Example1withEllipsoids}, with $h\equiv\textrm{const}>0$ 
and $A \equiv \mathrm{diag}(a)$, with $a\in(0,\infty)^n$. Clearly $A$  satisfies the 
assumption~(\ref{eq:AssumptionOptimizationSimplify}).
From Theorem \ref{thm:NGDboundasBSDEsol} we know that $\pi^u_t(X)=Y_t,\ t\in[0,T]$, 
for $(Y,Z)$ being the unique standard solution to the BSDE 
\begin{equation}\label{eq:BSDEvalExCloseFormSol}
 -dY_t= h\big(\sum_{i=d+1}^n(Z^i_t)^2/a_{i}\big)^{1/2}dt - Z^{\textrm{tr}}_t dW_t,\ t\le T,\quad\text{and}\quad Y_T=X.
\end{equation}
By Theorem \ref{thm:FullHedgingProblemSolved} the good-deal hedging strategy is $\bar{\phi}_t=\Pi(Z_t)\ P\otimes dt$-a.e.. Define  the geometric averages $\tilde{S}_t=\big(\prod_{k=1}^dS^k_t\big)^{1/d}$ and 
$\tilde{H}_t=\big(\prod_{l=1}^{n-d}H^l_t\big)^{1/(n-d)}$. One has $\tilde{S}_t=\tilde{S}_0\exp\Big(\tilde{\sigma}^{\textrm{tr}} W^S_t+\big(\tilde{\mu}-\frac{1}{2}\vert\tilde{\sigma}\vert^2\big)t\Big)$ 
and $\tilde{H}_t = \tilde{H}_0\exp\Big(\tilde{\beta}^{\textrm{tr}} W_t+\big(\tilde{\gamma}-\frac{1}{2}\vert\tilde{\beta}\vert^2\big)t\Big)$, where 
$\tilde{\sigma}:=\frac{1}{d}(\sigma^S)^{\textrm{tr}}\mathds{1},\quad \tilde{\mu}:=\frac{1}{2}\vert\tilde{\sigma}\vert^2-\frac{1}{2d}\vert\sigma^S\vert^2,\quad\tilde{\beta} = \frac{1}{n-d}\beta^{\textrm{tr}}\mathds{1}$
and $\tilde{\gamma}:=\frac{1}{n-d}\gamma^{\textrm{tr}}\mathds{1}+\frac{1}{2}\vert\tilde{\beta}\vert^2-\frac{1}{2(n-d)}\vert\beta\vert^2$, with $\mathds{1}=(1,\ldots,1)^\textrm{tr}$.
We treat the following two examples. 
\paragraph{European option on non-traded assets:} Consider a European option $X=G(\tilde{H}_T)\in L^2$ on the geometric average $\tilde{H}$, where the payoff $ x\mapsto G(x)$ is non-decreasing, 
measurable and of polynomial growth in $x^{\pm1}$, i.e.\  $\vert G(x)\vert\le k(1+x^n+x^{-n})$ for all $x$ in $(0,\infty)$, for some $k>0$ and $n\in\N$.  
Let us consider the constant process $\lambda$ with value $h\big(\sum_{i=d+1}^n\tilde{\beta}_i^2/a_{i}\big)^{-1/2}(0,\ldots,0,\tilde{\beta}_{d+1}/a_{d+1},\ldots,\tilde{\beta}_n/a_n)^{\textrm{tr}}$ 
to identify the solution 
$(Y,Z)$ of the BSDE (\ref{eq:BSDEvalExCloseFormSol}) as follows:
We show that the standard solution $(Y^\lambda,Z^\lambda)$ to the linear BSDE (\ref{eq:BSDEYlambda}) corresponding to $\lambda$ and with terminal condition $G(\tilde{H}_T)$ coincides with 
$(Y,Z)$. By Lemma \ref{lem:ClosedMartandBSDEs} we know that $Y^\lambda_t=E^{Q^\lambda}_t[G(\tilde{H}_T)],\ t\in[0,T]$, with $dQ^\lambda/dP = \mathcal{E}(\lambda^\text{tr}W)$.
Moreover the Feynman-Kac formula yields $Y^\lambda_t=u^\lambda(t,\tilde{H}_t)$ and $Z^\lambda_t= \tilde{H}_t\partial_xu^\lambda(t,\tilde{H}_t)\tilde{\beta}$ for a function $u^\lambda\in\mathcal{C}^{1,2}\big([0,T)\times(0,\infty)\big)$ 
solution to a Black-Scholes type PDE (after coordinate transformations that reduce the PDE into the heat equation \citep[using][Sect.4.3]{KaratzasShreve}). Since $G$ is 
non-decreasing, then $\partial_xu^\lambda\ge0$. Hence one has $P\otimes dt$-a.e. (omitting the argument $(t,\tilde{H}_t)$ of $\partial_xu^\lambda$ for simplicity)
\begin{equation*}
\lambda^{\text{tr}} Z^\lambda_t = h\,\tilde{H}_t\,\partial_xu^\lambda\,\sqrt{\sum_{i=d+1}^n\frac{\tilde{\beta}_i^2}{a_{i}}}
				=h\,\sqrt{\sum_{i=d+1}^n\frac{\big(\tilde{H}_t\,\tilde{\beta}_i\,\partial_xu^\lambda\big)^2}{a_{i}}}= h\,\sqrt{\sum_{i=d+1}^n\frac{(Z^{\lambda,i}_t)^2}{a_{i}}},
\end{equation*}
where the second equality uses $\partial_xu^\lambda\ge0$. This implies that $(Y^\lambda,Z^\lambda)$ solves the BSDE (\ref{eq:BSDEvalExCloseFormSol}), and therefore coincides with the latter's unique 
standard solution $(Y,Z)$.  In particular we obtain $\pi^u_t(X) = Y^\lambda_t=E^{Q^\lambda}_t[G(\tilde{H}_T)],\ t\in[0,T]$. Moreover $\tilde{H}_t = \tilde{H}_0\ e^{\alpha_+ t} \exp\big(\tilde{\beta}^{\textrm{tr}} W^\lambda_t -\frac{1}{2}\vert \tilde{\beta}\vert^2t\big)$, $t\in [0,T]$, where 
$\alpha_\pm :=\tilde{\gamma} \pm h\big(\sum_{i=d+1}^n\tilde{\beta}_i^2/a_{i}\big)^{1/2}$ and $W^\lambda$ is an $n$-dimensional $Q^\lambda$-Brownian motion. Specifically for 
$G(x):= (x-K)^+$, $X$ is a call option on $\tilde{H}$ with strike $K$ and maturity $T$. Its upper good-deal bound at time $t\in [0,T]$ is therefore given by a Black-Scholes type formula (with ``vol'' abbreviating volatility)
\begin{align}
	\pi^u_t(X) &=N(d_{+})\tilde{H}_te^{\alpha_+ (T-t)}-KN(d_{-})\notag\\
		    =&\text{B/S-call-price}\big(\text{time: } t,\ \text{spot: }\tilde{H}_te^{\alpha_+ (T-t)},\ \text{strike: } K, \text{vol: } \vert\tilde{\beta}\vert\big),\label{eq:CloFormUpperGDBCall}
\end{align}
where $d_{\pm}:=\big(\ln \big(\tilde{H}_t/K\big)+\big(\alpha_+\pm \frac{1}{2}\vert\tilde{\beta}\vert ^2\big)(T-t)\big)\big(\vert\tilde{\beta}\vert\sqrt{T-t}\big)^{-1}$ and $N$ is the cdf of the standard normal law. Analogously, the lower good-deal bound turns out as 
\begin{align*}
	\pi^l_t(X) &=\text{B/S-call-price}\big(\text{time: } t,\ \text{spot: }\tilde{H}_te^{\alpha_- (T-t)},\ \text{strike: } K, \text{vol: } \vert\tilde{\beta}\vert\big).
\end{align*}
The difference between the good-deal valuation formulas above and the standard Black-Scholes formula for risk-neutral valuation under measure $P=\widehat{Q}$ for a call option 
$(\tilde{H}_T-K)^+$ shows by the additional factors $e^{\alpha_{\pm} (T-t)}$ multiplying the spot price. This corresponds to an additional premium an option trader (selling at $\pi^u_\cdot(X)$ or buying 
at $\pi^l_\cdot(X)$) would require, if using the no-good-deal approach instead of  simple arbitrage free valuation under a given risk neutral measure   $P=\widehat{Q}$ (being an element of 
$ \mathcal{Q}^{\text{ngd}}$). The hedging strategy for the seller of $X$ is 
\begin{equation}\label{eq:CloFormHedgStraUpperGDBCall}
\bar{\phi}_t = e^{\alpha_+ (T-t)}N(d_+)\tilde{H}_t\big(\tilde{\beta}_1,\ldots,\tilde{\beta}_d,0,\ldots ,0\big)^{\textrm{tr}}\quad P\otimes dt\text{a.e.},
\end{equation}
which coincides with the delta (as computed under $P=\widehat{Q}$) for the call option only if $\alpha_+=0$. The hedging strategy of the buyer is derived analogously .

\paragraph{Exchange option of traded and non-traded assets:} Consider an European option to exchange the traded asset $\tilde{S}$ for the non-traded asset $\tilde{H}$ at maturity $T$ with payoff $X=(\tilde{H}_T-\tilde{S}_T)^+\in L^2$.
The upper bound $\pi^u_t(X)=E^{\bar{Q}}_t[X]$ can be explicitly derived using arguments from the previous example in combination with a change of num\'eraire \cite[see Derivation of 3.33 in Section 3.4 of][]{KentiaPhD}.  
We thereby obtain a Margrabe type formula 
\begin{align}
	\pi^u_t(X) &=N(d_{+})\tilde{H}_te^{\alpha_+ (T-t)}-N(d_{-})\tilde{S}_te^{\tilde{\mu}(T-t)}\label{eq:MargrabeFormNGD}\\
	    &=\text{B/S-call-price}\big(\text{time: } t,\ \text{spot: }\tilde{H}_te^{\alpha_+ (T-t)},\ \text{strike: }\tilde{S}_te^{\tilde{\mu}(T-t)} ,\ \text{vol: }\delta\big),\notag
\end{align}
where $d_{\pm}:=\big(\ln \big(\tilde{H}_t/\tilde{S}_t\big)+(\alpha_++\tilde{\mu}\pm\frac{\delta^2}{2})(T-t)\big)\big(\delta\sqrt{T-t}\big)^{-1}$. Analogously, the corresponding lower good-deal bound is 
\begin{align*}
	\pi^l_t(X) &=\text{B/S-call-price}\big(\text{time: } t,\ \text{spot: }\tilde{H}_te^{\alpha_- (T-t)},\ \text{strike: }\tilde{S}_te^{\tilde{\mu}(T-t)} ,\ \text{vol: }\delta\big).
\end{align*}
The good-deal hedging strategy $  \bar{\phi}_t$ for the seller of the exchange option equals
\begin{equation*}
 N(d_+)\tilde{H}_te^{\alpha_+ (T-t)}\big(\tilde{\beta}_1,\ldots,\tilde{\beta}_d,0,\ldots ,0\big)^{\textrm{tr}}-N(d_-)\tilde{S}_te^{\tilde{\mu} (T-t)}\big(\tilde{\sigma}_1,\ldots,\tilde{\sigma}_d,0,\ldots ,0\big)^{\textrm{tr}}
\end{equation*}
$P\otimes dt$-a.e.. Again, the difference between the good-deal valuation formula and the classical Margrabe formula, as computed by standard no-arbitrage valuation under risk neutral measure $P=\widehat{Q}$, for the 
exchange option $(\tilde{H}_T-\tilde{S}_T)^+$ shows by the presence of the factors involving the term $\alpha_\pm$, which depends only on the parameters $A$ and $h$ for no-good-deal restrictions. 

\subsubsection{Computational results by Monte Carlo} \label{sec:MC}
To demonstrate that good-deal bounds and hedging strategies can be computed numerically in moderately high dimensions by generic simulation methods available for classical BSDE, we apply a (generic) 
multilevel Monte Carlo algorithm from \citet{BechererTurkedjiev} \citep[that builds on][{reducing variances i.p.\ for $Z$}]{GobetTu} to approximate the solution $(Y,Z)$ of the BSDE (\ref{eq:BSDEvalExCloseFormSol}) 
in dimension $n=4$, and compare with the known analytical solution for the  exchange option $X:=(\tilde{H}_T-\tilde{S}_T)^+$. Using parameters $d=2$, $T=1$ and
\begin{align*}
&\quad H_0=\begin{pmatrix}1\\1 \end{pmatrix},\quad S_0=\begin{pmatrix}1\\1 \end{pmatrix},\quad 
\sigma^S=\begin{pmatrix}0.5\ &0.2\\0\ &0.4 \end{pmatrix},\quad  
\gamma = (0.1,\ 0.3)^{\text{tr}},\\&\beta=\begin{pmatrix}0.3\ &0.4\ &0.2\ &0.5\\ 0.5\ &0.7\ &0.3\ &0.4\end{pmatrix},\quad
h=0.3,\quad \text{and}\quad  A=\text{diag}(0.5,\ 0.65,\ 0.8,\ 0.95),
\end{align*}
 we compare the approximate values at time $t=0$ to the known theoretical values obtained from Section \ref{subsec:OptionOnNonTradedAsset}. The exact value of the good-deal bound at 
 time $t=0$ according to the formula (\ref{eq:MargrabeFormNGD}) is then $\pi^u_0(X)={0.5494}$, up to four digits, while for the hedging strategy it is $\bar{\phi}_0={(0.3049,0.4440,0,0)}$,  
 the exact value of $Z_0$ being $(0.3049,0.4440,0.2792,0.5025)$. We use a 4-level algorithm on an equidistant time grid with $N=2^4$  steps, a number of sample paths $M=3\times 10^{6}$ and with $K=50^4$ 
 regression functions, being indicator functions on a hypercube partition of $\R^4$, the state space of the forward process $(S,H)$. Table \ref{tab:GDBapprox-and-Phi-Approx} provides the 
numerical simulation results, summarized by the approximation means for the good-deal bound and  the hedging strategy at time $0$, the empirical root-mean-square errors (RMSE) 
computed coordinate-wise and the corresponding relative values (Rel.RMSE), based on $80$ independent simulation runs. 
 Simulation in Matlab for one run took 153sec on a core-i7 cpu laptop, showing relative errors (in terms of maximal coordinates in Rel.RMSE) of about $0.07\%$ for 
 valuation and $0.34\%$ for hedging.
\begin{table}[h]
\centering
{\renewcommand{\arraystretch}{1.5}
\begin{tabular}{l|c|c|c|}
\cline{2-4}
&$Y_0$ approx& $Z_0$ approx&$\bar{\phi}_0$ approx\\\hline
\multicolumn{1}{|l|}{Mean}
 &$0.5499$&$(0.3052,0.4462,0.2852,0.5137)$&$(0.3052,0.4462,0,0)$\\\hline
 \multicolumn{1}{|l|}{RMSE}
 &$10^{-4}\times 4$&$10^{-4}\times(10,13,12,13)$&$10^{-4}\times(10,13,0,0)$\\\hline
 \multicolumn{1}{|l|}{Rel.RMSE}
 &$10^{-4}\times 7$&$10^{-4}\times(34,29,41,27)$&$10^{-4}\times(34,29,0,0)$\\\hline
\end{tabular}}
\caption{Mean  and (relative) root-mean-square errors of approximations}
\label{tab:GDBapprox-and-Phi-Approx}
\end{table}

\subsubsection{Semi-explicit formulas in the Heston stochastic volatility model}\label{subsec:GDVinStochVolModel}
The market information is generated by a two-dimensional $P$-Brownian motion $W=(W^S,W^\nu)$, and is augmented by null-sets. We are going to consider a European put option  $X=(K-S_T)^+$ on $S$ with strike $K$ in the Heston model    
\begin{equation*}
 dS_t = S_t \sqrt{\nu_t}\ dW^S_t\quad\text{and}\quad d\nu_t = b(\frac{a}{b}-\nu_t)dt +\beta\sqrt{\nu_t}\big(\rho dW^S_t + \sqrt{1-\rho^2}dW^\nu_t\big),\ t\le T,
\end{equation*}
that is specified directly under a risk neutral measure $P=\widehat{Q}$,
with $S_0,\nu_0>0$, $a,b,\beta>0$ and $\rho\in(-1,1)$. Here $b$ represents the mean-reversion speed, $a/b$ the mean-reversion level and $\beta/2$ the volatility of the variance process 
$\nu$. Assume that the Feller condition $\beta^2\le2a$ is satisfied, such that $\nu>0$. The equivalent local martingale measures $Q\in\mathcal{M}^e$ in this model are specified by Girsanov kernels $\lambda$ 
such that $dQ/dP=\mathcal{E}(\lambda\cdot W^\nu)$ is a uniformly integrable martingale. Indeed, we parametrize the pricing measures only by the second component of their Girsanov kernels 
(i.e.\  w.r.t.\  $W^\nu$) since the first component is always zero. We consider the no-good-deal constraint correspondence
\begin{equation}\label{eq:StochVolNGDRestriction}
 C_t(\omega) = \big\lbrace x\in\R^2: \lvert x\rvert\le \varepsilon/\sqrt{\nu_t(\omega)}\big\rbrace \quad (t,\omega)\in[0,T]\times \Omega, 
\end{equation}
for a constant $\varepsilon>0$. One observes that $C$ is standard with $0\in C$, non-uniformly bounded 
and satisfies (\ref{eq:2IntCorr}) for 
$R=\varepsilon/\sqrt{\nu}$ (since $\nu>0$ is  continuous). Hence good-deal valuation results for uniformly bounded correspondences may not apply. We obtain a  Heston-type formula (semi-explicit, computation requiring only 
1-dim.\ integration) for the good-deal 
bound of the put option $X=(K-S_T)^+$,
\begin{equation}\label{eq:GoodDealUpperForPut}
	\pi^u_t\left(X\right) =\textrm{ Heston-put-price}(\text{time: }t, \bar{a}:=a+\beta\varepsilon\sqrt{1-\rho^2},b,\beta),
\end{equation}
just like the ordinary Heston put price, associated to parameters $(t,a,b,\beta)$, but where the parameter $a$ has to be adjusted to $\bar{a}:=a+\beta\varepsilon\sqrt{1-\rho^2}$. 
The formula for the lower bound $\pi^l_t(X)$ is similar, but here $\bar{a}$ is replaced by $\underline{a}:=a-\beta\varepsilon\sqrt{1-\rho^2}$, for which the Feller condition still holds if $\varepsilon\le \frac{1}{2}\beta^{-1}(2a-\beta^2)(1-\rho^2)^{-1/2}$.
In particular, $\pi^u_t(X) = E^{\bar{Q}}_t[X]$ holds with $d\bar{Q}/dP = \mathcal{E}\big((\varepsilon/\sqrt{\nu})\cdot W^\nu\big)$. By Corollary~\ref{cor:GDBsolvesBSDEunderOptimalMeas} 
this yields $\bar{Y}=\pi^u_\cdot\left(X\right)$ for the minimal solution $(\bar{Y},\bar{Z})\in\mathcal{S}^\infty\times\h^2$ of the BSDE
\begin{equation}\label{eq:BSDEmINsOL}
	-dY_t =\frac{\varepsilon}{\sqrt{\nu_t}}\big\lvert Z^2_t\big\rvert- Z_t^{\text{tr}}dW_t,\ t\in[0,T],\quad Y_T= (K-S_T)^+.
\end{equation}
The (seller's) good-deal hedging strategy $\bar{\phi}$ is given by the semi-explicit formula 
\begin{equation}\label{eq:GoodDealHedgingForPut}
 \bar{\phi}_t=S_t\sqrt{\nu_t}\ \Delta_t+\frac{\beta\rho}{2}\ \mathcal{V}_t\quad P\otimes dt\text{-a.e.},
\end{equation}
where $\Delta_t$ and $\mathcal{V}_t$ denote the delta and the vega of the put option at time $t$ in the Heston model with parameters $(\bar{a},b,\beta)$. Derivations are provided in Appendix~\ref{app:AppendixA}.
We note that (\ref{eq:GoodDealHedgingForPut}) coincides \citep[cf.][]{PoulsenetAl2009} with the risk-minimizing strategy  \citep[in the sense of][]{Schweizer} for the put in a Heston model, not with respect 
to the probability $P$ but with respect to the measure $\bar Q$ (derived just before) under which also Heston dynamics but with modified parameters prevail. This shows, how the strategy (\ref{eq:GoodDealHedgingForPut}) 
differs from the standard risk minimizing strategy under $P$ \citep[as in][]{PoulsenetAl2009,HeathPS01}.
Good-deal valuation bounds for a put option in the Heston model are thus given by a Heston type formula but for a mean-reversion level increased by 
$\beta\varepsilon\sqrt{1-\rho^2}/b>0$. Similar to earlier examples, this difference constitutes an increase in the premium that an issuer selling at $\pi^u_\cdot(X)$ would require 
according to good-deal valuation, in comparison to a standard arbitrage free valuation under one given risk neutral measure $P=\widehat{Q}$, when $S$ is the only risky asset 
available for hedging and  stochastic volatility risk is otherwise taken to be unspanned. 
Figure~\ref{fig:plotHestonGBDvsS0} illustrates this, showing good deal valuations $\pi^u_0(X)$, $\pi^l_0(X)$ (at $t=0$) for a long-dated put option with maturity $T=10$ in relation to the underlying 
$S_0$, for different no-good-deal constraint parameters $\varepsilon$. Other parameters are $K=100$, $a=0.12,\ b=3,\ \beta=0.3,\ \rho=-0.7,\ \nu_0=0.04$.
The standard Heston price computed directly under a given risk neutral (minimal martingale) measure $P=\widehat{Q}$ (i.e.\ for $\varepsilon=0$) lies between the upper and lower good-deal bounds, whose 
spread  increases with $\varepsilon>0$. 
\begin{figure}[h!]
        \centering
               \includegraphics[scale = 0.33]{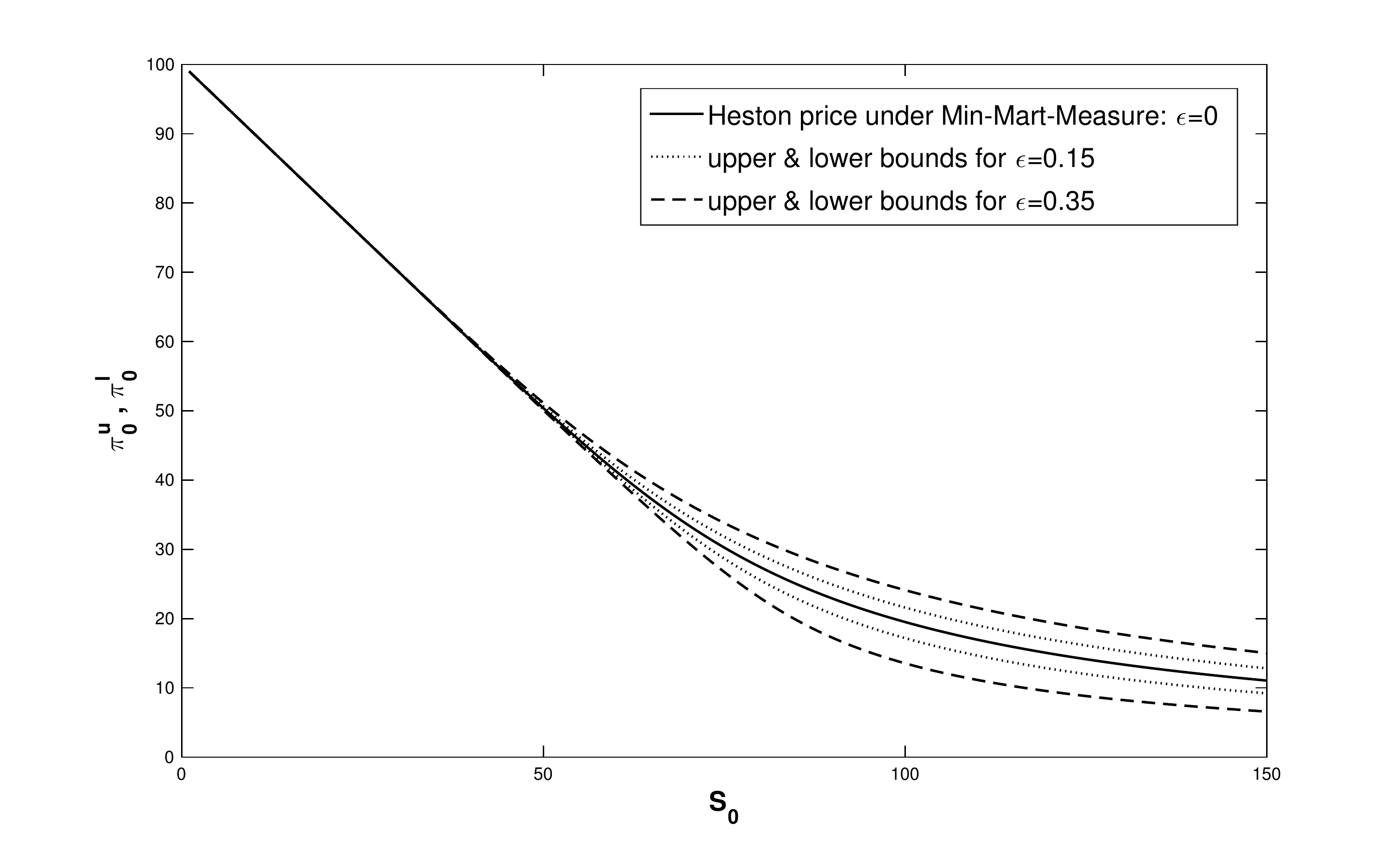} 
        \caption{Dependence of $\pi^u_0(X), \pi^l_0(X)$ on $S_0$ and $\varepsilon$ in the Heston model for $K=100$, $a=0.12,\ b=3,\ \beta=0.3,\ \rho=-0.7,\ \nu_0=0.04,\ T=10$.}       
        \label{fig:plotHestonGBDvsS0}
\end{figure}
\section{Good-deal valuation and hedging under model uncertainty}\label{sec:GDRestrictionUnderUncertainty}
In preceding sections, good-deal bounds and hedging strategies have been described by classical BSDEs under the probability measure $P$, expressing the objects of 
interest in terms of the market price of risk $\xi$ with respect to $P$. In reality, the objective real world probability measure is not precisely known, hence there is ambiguity 
about the market price of risk. To include model uncertainty (ambiguity)  into the analysis,  we follow  a multiple priors
approach in spirit of \cite{GilboaS89,ChenEpstein,EpsteinS03}, 
by specifying a confidence region of reference probability measures  $\{P^{\theta}:\theta \in \Theta\}$ (multiple priors,
interpreted as potential real world probabilities of equal right), centered around some measure $P_0$. In practice, an investor facing model uncertainty may 
first extract an estimate $P_0$ for the true but uncertain $P$ from data, but then consider a
class $\mathcal{R}$ of potential reference measures in some confidence region around $P_0$  to acknowledge the statistical uncertainty of estimation.

The starting point for the good-deal valuation approach under uncertainty is then to associate to each 
model $P^\theta$ its own family of (a-priori) no-good-deal measures  $\mathcal{Q}^{\mathrm{ngd}}(P^\theta)$ (resp.\  $\mathcal{P}^{\mathrm{ngd}}(P^\theta)$). 
A robust worst-case approach  requires the seller of 
a derivative to consider the (worst-case) model $P^{\bar{\theta}}$  that provides the largest upper good-deal valuation {bound},  to be conservative against model misspecification (see~(\ref{eq:InfSupNGDBoundRobust})).
Such leads to wider good-deal  bounds, corresponding to a larger overall set of no-good-deal measures under uncertainty. Notably,
it will simultaneously also give rise to a suitable robust notion of good-deal hedging, which is uniform with respect to all $P^{\theta}$, 
by means of a saddle point result that ensures a minmax identity (see Theorem \ref{thm:CharactOfSaddlePt}). 
We 
associate to each model $P^\theta$ a correspondence $C^\theta$ that defines the set of no-good-deal measures in this model. 
The aggregate set of no-good-deal measures will be described then by single correspondence $\widetilde{C}$, which incorporates also the uncertainty. Technically, this makes is 
possible to apply analysis obtained in the framework of  previous  sections, with $P_0$ taking the role of $P$.

Let us note that the worst-case and the multiple prior approach to uncertainty are classical in the theory, and so is the setup for our problem. Yet,  they may be perceived as overly  conservative, as results could be much influenced by just few extreme priors, what could appear even less desirable if one would not be ready to consider all priors to be of equal right. 
This may suggest directions for further research, e.g.\ to study some weighting or penalization of priors. It may also be interesting to study non-dominated priors (as e.g.\ in articles on superreplication under volatility uncertainty by 2nd-order BSDE) which do not fit within Section~\ref{sec:GDRestrictionUnderUncertainty}.

\subsection{Model uncertainty framework}
Let $(\Omega,\F,P_0,\mathbb{F})$ be a probability space with a usual filtration $\mathbb{F}=(\F_t)_{t\le T}$  generated by an $n$-dimensional $P_0$-Brownian motion $W^0$. 
We assume that all reference measures $P^\theta$ are equivalent to $P_0$ with corresponding Girsanov kernels $\theta$ evolving in some given confidence region $\Theta$.  More precisely, 
we define 
\begin{equation*}
	\mathcal{R}:=\Big\lbrace P^\theta\sim P_0\ \Big\lvert\ dP^{\theta}/dP_0\Big.=\ \mathcal{E}\left(\theta\cdot W^0\right)
                             ,\ \mathrm{with}\ \theta\ \text{predictable\ and } \theta\in \Theta\Big\rbrace,
\end{equation*}
where $\Theta:[0,T]\times\Omega\rightsquigarrow\R^n$ is a standard correspondence satisfying Assumption~\ref{asp:BddCorr} and  $0\in \Theta$, hence $P_0\in \mathcal{R}\neq \emptyset$. 
A similar framework has been considered for example in \citet{ChenEpstein,Quenez} for solving the robust utility maximization problem under Knightian uncertainty about drift coefficients.  

The financial market consists of $d\le n$ tradeable risky assets whose discounted prices $(S^i)_{i=1}^{d}$ under $P^\theta$ (for $\theta\in \Theta$) evolve as It\^o processes, solving the SDEs
\begin{equation}\label{eq:SDEStockPricesUncertainty}
	dS_t = \textrm{diag}(S_t)\sigma_t(\xi^\theta_tdt+dW^\theta_t)=:\textrm{diag}(S_t)\sigma_td\widehat{W}^\theta_t,\quad t\le T,
\end{equation}
with $S_0\in (0,\infty)^d$, for $\R^n$-valued predictable $\xi^\theta$ and $\R^{d\times n}$-valued predictable volatility $\sigma$ of full rank, and $W^\theta:=W^0-\int_0^\cdot\theta_sds$ a 
$P^\theta$-Brownian motion. Noting that market prices of risk, $\xi^\theta_t$ and $\xi^0_t$, canonically take values in $\mathrm{Im}\ \sigma^{\textrm{tr}}_t$,  we assume that market prices 
of risk $\xi^\theta$ (under  $P^\theta$ for $\theta\in \Theta$) have the form  
\begin{equation}\label{eq:NiceConditions}
 \xi^\theta_t = \xi^0_t+\Pi_t(\theta_t)\in\mathrm{Im}\ \sigma^{\textrm{tr}}_t,\ t\in [0,T],
\end{equation}
and that $\xi^0$ is bounded. By (\ref{eq:NiceConditions}), the solutions of the SDEs (\ref{eq:SDEStockPricesUncertainty}) coincide $P_0\textrm{-a.s.}$ for all $\theta\in \Theta$.
The process $\xi^\theta$ (for $\theta\in \Theta$) is the market price of risk in the model $P^\theta$ and is also bounded (since $\xi^0$ is bounded and $\Theta$ satisfies Assumption \ref{asp:BddCorr}).
Hence, the minimal martingale measure \citep{Schweizer} $\widehat{Q}^\theta$ with respect to $P^\theta$ is $d\widehat{Q}^\theta=\mathcal{E}(-\xi^\theta\cdot W^\theta)dP^\theta$. In addition 
$d\widehat{Q}^\theta=\mathcal{E}\left(\Pi^\bot(\theta)\cdot \widehat{W}^0\right)d\widehat{Q}^0$ and $\widehat{W}^\theta = \widehat{W}^0-\int_0^\cdot\Pi^\bot_t(\theta_t)dt$, for all $\theta\in \Theta.$
We recall from Section~\ref{sec:Parameterization} how dynamic trading strategies are defined and re-parametrized in terms of integrands $(\phi^i)_{i=1}^d$ with respect to $\widehat{W}^0$. 
The set of permitted trading strategies is 
\begin{equation*}
	\Phi := \Big\lbrace \phi\ \Big\lvert\ \phi\ \textrm{is predictable},\ \phi\in \mathrm{Im}\ \sigma^{\textrm{tr}}\ \textrm{and}\  E^{P_0}\Big[\int_0^T\lvert\phi_t\rvert^2dt\Big]<\infty\Big.\Big\rbrace.
\end{equation*}
Since $\phi^\textrm{tr}\Pi^\bot(\theta)=0$ for $\theta\in \Theta$,  the wealth process $V^\phi$ of strategy $\phi\in \Phi$ with initial capital $V_0$ is $V^\phi =V_0 +\phi\cdot \widehat{W}^\theta = V_0 +\phi\cdot \widehat{W}^0$, 
for all $\theta\in \Theta$. Let $\mathcal{M}^e(P^\theta):=\mathcal{M}^e(S,P^\theta)$ denote the set of equivalent local martingale measures for $S$ in the model $P^\theta$.
Noting  $P^{\theta}\sim P^0$ and recalling Proposition \ref{pro:ImAndKerCorrespondencesPredictableCharactOfMe} one easily obtains

\begin{proposition}\label{pro:EMMsetsCoincide}
	 $\mathcal{M}^e(P^\theta)=\mathcal{M}^e(P_0)$ for all  $\theta\in \Theta$. In addition, every $Q\in\mathcal{M}^e(P^\theta)$ satisfies $dQ = \mathcal{E}\left(\lambda^{\theta}\cdot W^\theta\right)dP^\theta$ 
	 and $dQ = \mathcal{E}\left(\lambda^{0}\cdot W^0\right)dP_0$, with $\lambda^{\theta}=-\xi^\theta+\eta^{\theta}$ and $\lambda^{0}=-\xi^0+\eta^{0}$, where $\Pi^{\bot}(\lambda^{\theta})= \eta^{\theta}, \quad \Pi^{\bot}(\lambda^{0})=\eta^{0}$ 
	 and $\eta^{\theta} = \eta^{0}-\Pi^\bot(\theta)$.
\end{proposition}
Thus, we simply write $\mathcal{M}^e=\mathcal{M}^e(S)$ for the set of equivalent martingale measures. 
\subsection{No-good-deal restriction and good-deal bounds under uncertainty}\label{sec:GeneralNGDRestriction}
Let $\{C^\theta\ \vert\ \theta\in \Theta\}$ be a family of standard correspondences satisfying 
\begin{equation}\label{eq:xithetainCtheta}
 -\xi^\theta\in C^\theta\quad  \text{for all }\theta\in\Theta.
\end{equation}
In the model $P^\theta$, $\theta\in \Theta$, let the no-good-deal restriction be such that the Girsanov kernels of measures in $\mathcal{M}^e$ are selections of $C^\theta$. The resulting 
set $\mathcal{Q}^{\mathrm{ngd}}(P^\theta)$ of no-good-deal measures is equal to
\begin{equation*}
\Big\lbrace Q\sim P^\theta\ \Big\lvert \Big.\ dQ/dP^\theta=\mathcal{E}\big(\lambda\cdot W^\theta\big),\ \lambda\textrm{ predictable, bounded},  
												    \ \lambda\in (-\xi^\theta+\mathrm{Ker}\ \sigma)\cap C^\theta\Big\rbrace.
\end{equation*}
By (\ref{eq:xithetainCtheta}), then $\widehat{Q}^\theta\in \mathcal{Q}^{\mathrm{ngd}}(P^\theta)\neq \emptyset$ for all $\theta\in\Theta$. By Proposition \ref{pro:EMMsetsCoincide} holds 
\begin{equation}\label{eq:QngdPnuCharact}
                \mathcal{Q}^{\mathrm{ngd}}(P^\theta)=\Big\lbrace Q\sim P_0\Big\lvert \Big.\ dQ/dP_0=\mathcal{E}\big(\lambda\cdot W^0\big),\ 
												    \ \lambda\in-\xi^0+(\widetilde{C}^\theta\cap \mathrm{Ker}\ \sigma)  
\Big\rbrace
\end{equation}
where $\lambda$ is predictable and bounded, and for all $\theta\in \Theta$ the correspondences 
are 
\begin{equation}\label{eq:DefCtildeNu}
 \widetilde{C}^\theta:=C^\theta+\xi^\theta+\Pi^\bot(\theta)=C^\theta+\xi^0+\theta.
\end{equation}
Following a worst-case approach, we take the (robust) upper good-deal valuation $\pi^u_\cdot(\cdot)$ under uncertainty as being the largest of all good-deal bounds $\pi^{u,\theta}_\cdot(\cdot)$ over 
all models $P^\theta,\ \theta\in\Theta$. The respective set $\mathcal{Q}^{\mathrm{ngd}}$ of no-good-deal valuation measures corresponding to $\pi^u_\cdot(\cdot)$ can be described in terms of the sets 
$\mathcal{Q}^{\mathrm{ngd}}(P^\theta),\ \theta\in \Theta$. At first, one might guess that $\mathcal{Q}^{\mathrm{ngd}}$ should be the union of all $\mathcal{Q}^{\mathrm{ngd}}(P^\theta)$.
However, to have m-stability and convexity of $\mathcal{Q}^{\mathrm{ngd}}$ for good dynamic properties of the resulting good-deal bounds (as in Lemma~\ref{lem:PropDynamicRiskMeasMeandQngdareMstable}), 
one has to  define $\mathcal{Q}^{\mathrm{ngd}}$ as the smallest m-stable and convex set containing all $\mathcal{Q}^{\mathrm{ngd}}(P^\theta),\ \theta\in \Theta$. 
\begin{definition}\label{defn:SetOfNGDMeas}
$\mathcal{Q}^{\mathrm{ngd}}$ is the smallest m-stable convex subset of $\mathcal{M}^e$ containing all $\mathcal{Q}^{\mathrm{ngd}}(P^\theta)$, $\theta\in \Theta$. For sufficiently integrable claims $X$ (e.g.\ in $L^\infty$), the worst-case upper good-deal bound under uncertainty is $ \pi^u_t(X):= \esssup_{Q\in \mathcal{Q}^{\mathrm{ngd}}} E^Q_t[X]$. 
\end{definition}
We characterize the set $\mathcal{Q}^{\mathrm{ngd}}$ from Definition \ref{defn:SetOfNGDMeas} using a suitable single correspondence $\widetilde{C}$ which is derived from all 
$C^\theta,\ \theta\in \Theta$. To this end, we impose the 
\begin{assumption}\label{asp:AssumptiononCnus}
	The correspondence with values $\bigcup_{\theta\in \Theta} \widetilde{C}^\theta_t(\omega),\ (t,\omega)\in[0,T]\times\Omega,$ is compact-valued and predictable.
\end{assumption}
The theory of measurable correspondences is well-developed for closed-valued correspondences \citep[see][]{Rockafellar}. Assumption \ref{asp:AssumptiononCnus} ensures closed-valuedness and 
predictability of $\widetilde{C}$ for the proposition below. If all $C^\theta$ ($\theta\in \Theta$) are equal to some given $C^0$, as in the following example, such an assumption will
automatically hold in the setting required for Section \ref{sec:GDValandHedgingUnderUncertaintyEllipsoids}, where $\widetilde{C}^\theta$ $(\theta\in\Theta)$ are ellipsoidal.
\begin{example}\label{exa:ExampleExplicitNGDCorr}
 For a standard correspondence $C^0$ with $\xi^\theta\in C^0$, $\theta\in\Theta$, let $C^\theta := C^0$, $\theta\in\Theta$. Then $\widetilde{C}^\theta=C^0+\xi^0+\theta$ and $\bigcup_{\theta\in\Theta}\widetilde{C}^\theta=C^0+\xi^0+\Theta$ satisfies Assumption~\ref{asp:AssumptiononCnus}.
\end{example}
\begin{proposition}\label{pro:SetofNGDMeasCharact}
Let Assumption \ref{asp:AssumptiononCnus} hold. Then $  \mathcal{Q}^{\mathrm{ngd}}$ equals
    	\begin{equation}\label{eq:QngdUncertCharact}
               \Big\lbrace Q\sim P_0\ \Big\lvert \Big.\ dQ/dP_0=\mathcal{E}\big(\lambda\cdot W^0\big),\ \lambda=-\xi^0+\eta\textrm{ predictable, bounded, }
												    \eta\in \widetilde{C} \Big\rbrace,
         \end{equation}
for the standard correspondence $ \widetilde{C}_t(\omega) :=  \mathrm{Ker}\ \sigma_t(\omega) \cap \mathrm{Conv}\Big(\bigcup_{\theta\in \Theta}\widetilde{C}^\theta_t(\omega)\Big)$.
\end{proposition}
\begin{proof}
With Assumption \ref{asp:AssumptiononCnus}, \citep[][Thm.1.M and Prop.1.H]{Rockafellar} imply that $\widetilde{C}$ is standard. Note that $\widetilde{C}$ is non-empty-valued since 
$-\xi^0\in C^0$ and hence $0\in \widetilde{C}^0_t(\omega)\cap \mathrm{Ker}\ \sigma_t(\omega)\subset \widetilde{C}_t(\omega)$. Denote by $\mathcal{Q}$ the set in (\ref{eq:QngdUncertCharact}). 
By definition $\widetilde{C}_t(\omega)\subset \mathrm{Ker}\ \sigma_t(\omega)$, implying $\mathcal{Q}\subseteq\mathcal{M}^e$. We first prove that $\mathcal{Q}^{\mathrm{ngd}}\subseteq \mathcal{Q}$. 
Applying \citep[][Thm.1]{Delbaen} or following the steps of the proof for Lemma \ref{lem:PropDynamicRiskMeasMeandQngdareMstable},  Part b), one sees that $\mathcal{Q}$ is m-stable and convex. 
By (\ref{eq:QngdPnuCharact}) and since $ \widetilde{C}^\theta_t(\omega)\cap \mathrm{Ker}\ \sigma_t(\omega) \subseteq \widetilde{C}_t(\omega)$ for all $\theta\in\Theta,$ then $\mathcal{Q}$ contains 
the union of all $\mathcal{Q}^{\mathrm{ngd}}(P^\theta)$, $\theta\in \Theta$. By definition $\mathcal{Q}^{\mathrm{ngd}}$ is the smallest m-stable convex subset of $\mathcal{M}^e$ with this property, 
hence $\mathcal{Q}^{\mathrm{ngd}}\subseteq \mathcal{Q}$.

Let us show $\mathcal{Q}\subseteq \mathcal{Q}^{\mathrm{ngd}}$. The $L^1$-closure of $\mathcal{Q}^{\mathrm{ngd}}$ is an m-stable closed and convex set of measures $Q\ll P_0$, and 
$\mathcal{Q}^{\mathrm{ngd}}$ comprises exactly those elements of its closure that are equivalent to $P_0$.
Closedness and convexity of the closure of $\mathcal{Q}^{\mathrm{ngd}}$ are clear. We now show its m-stability. To this end, let $Z^1_T,Z^2_T$ be in the closure of $\mathcal{Q}^{\mathrm{ngd}}$,
$\tau\le T$ be a stopping time and $Z_T:=Z^1_\tau Z^2_T/Z^2_\tau I_{\{Z^2_\tau>0\}} + Z^1_\tau I_{\{Z^2_\tau=0\}}$. There exist $\big(Z^{1,n}_T\big)_n,\big(Z^{2,n}_T\big)_n\subseteq \mathcal{Q}^{\mathrm{ngd}}$
such that $Z^{1,n}_T\to Z^1_T$ and $Z^{2,n}_T\to Z^2_T$ in $L^1$. By m-stability of $\mathcal{Q}^{\mathrm{ngd}}$ holds $Z^n_T:=Z^{1,n}_\tau Z^{2,n}_T/Z^{2,n}_\tau\in \mathcal{Q}^{\mathrm{ngd}}$
for each $n\in\N$. Now $E[Z^n_T]=1$ for all $n\in\N$, and $Z^n_T\to Z_T$ in probability as $n\to \infty$. In addition, 
\begin{align*}
E[Z_T] &= E[Z^1_\tau Z^2_T/Z^2_\tau I_{\{Z^2_\tau>0\}}] + E[Z^1_\tau I_{\{Z^2_\tau=0\}}]\\
	& = E\big[ E_\tau[Z^2_T/Z^2_\tau ]\ Z^1_\tau I_{\{Z^2_\tau>0\}} \big] + E[Z^1_\tau I_{\{Z^2_\tau=0\}}]
           =E[ Z^1_\tau]=1.
\end{align*}
By Scheff\'e's lemma one obtains $Z^{n}_T\to Z_T$ in $L^1$ as $n\to \infty$, and m-stability of the closure of $\mathcal{Q}^{\mathrm{ngd}}$ follows. As $W^0$ is a continuous 
$P_0$-martingale with the predictable representation property, it satisfies the hypotheses of \citep[][Thm.2]{Delbaen}, implying by Definition \ref{defn:SetOfNGDMeas} the existence 
of a closed-convex-valued predictable correspondence $C^1$ such that the no-good-deal measure set  $\mathcal{Q}^{\mathrm{ngd}}$ is equal to
\begin{align*}
	 \Big\lbrace Q\sim P_0 \Big\lvert \Big.\Big.&\ dQ/dP_0=\mathcal{E}\left(\lambda\cdot W^0\right),\ \lambda=-\xi^0+\eta \textrm{ predictable, }
												    \eta\in C^1\cap\mathrm{Ker}\ 
\sigma \Big.\Big\rbrace\label{eq:QngdUncertCharactDelbaen}.
\end{align*}
To prove the claim, it suffices to show that all predictable selections of $\widetilde{C}$ are also predictable selections of $C^1\cap\mathrm{Ker}\ \sigma$. To this end it suffices to 
show that for all $\theta\in \Theta$,  any predictable selection of $\widetilde{C}^\theta\cap\mathrm{Ker}\ \sigma$ is a predictable selection of $C^1\cap\mathrm{Ker}\ \sigma$. Assume 
the contrary that there exists $\theta\in \Theta$ and a predictable process $\eta$ such that $\eta\in \widetilde{C}^\theta\cap\mathrm{Ker}\ \sigma$ and $\eta$ is not selection of 
$C^1\cap\mathrm{Ker}\ \sigma$. Then $\mathcal{E}\left((-\xi^0+\eta)\cdot W^0\right)$ is in $\mathcal{Q}^{\mathrm{ngd}}(P^\theta)$ but not in $\mathcal{Q}^{\mathrm{ngd}}$, which contradicts 
$\mathcal{Q}^{\mathrm{ngd}}(P^\theta)\subseteq \mathcal{Q}^{\mathrm{ngd}}$.
\end{proof}
Using the characterization of $\mathcal{Q}^{\text{ngd}}$ in Proposition \ref{pro:SetofNGDMeasCharact} we can apply the results of Sections~\ref{sec:Prelim}-\ref{sec:No-Good-Deal-HedgingApproach} in order to 
derive worst-case good-deal bounds and hedging strategies under uncertainty like in the absence of uncertainty, with the center $P_0$ of the set of reference measures $\mathcal{R}$ taking 
the role of $P$ (in Sections~\ref{sec:Prelim}-\ref{sec:No-Good-Deal-HedgingApproach}) and the enlarged correspondence $\widetilde{C}$ taking the role of $C$ there.
\begin{example}\label{explQngd}
For $C^\theta, \theta\in\Theta,$ as in Example~\ref{exa:ExampleExplicitNGDCorr} holds $\widetilde{C}=(C^0+\xi^0+\Theta)\cap \mathrm{Ker}\ \sigma$ and
\begin{equation}\label{eq:forExaUnionQthetaisQ}
  \mathcal{Q}^{\mathrm{ngd}}= \Big\lbrace Q\sim P_0\ \Big\lvert \Big.\ dQ/dP_0=\mathcal{E}\left(\lambda\cdot W^0\right),\ \lambda\in(-\xi^0+\mathrm{Ker}\ \sigma)\cap (C^0+\Theta) 
    \Big\rbrace
\end{equation}
 with $\lambda$ denoting bounded predictable selections, by Proposition \ref{pro:SetofNGDMeasCharact}. Moreover the union $\bigcup_{\theta\in\Theta}\mathcal{Q}^{\mathrm{ngd}}(P^\theta)$ is convex, m-stable (cf.\ Lemma~\ref{lem:UnionPngdQngd-mstable}) and equals  $\mathcal{Q}^{\mathrm{ngd}}$.
\end{example}
\begin{remark}\label{Rem:C0plusTheta}
a) Equation (\ref{eq:QngdUncertCharact}) shows, how the good-deal valuation and hedging problem 
under model uncertainty can technically  be embedded into the mathematical framework of  Sections~\ref{sec:Prelim}-\ref{sec:No-Good-Deal-HedgingApproach} without uncertainty,
by considering  an enlarged  no-good-deal constraint correspondence  $C$ as $\mathrm{Conv}(\cup_{\theta\in \Theta}({C}^\theta +\theta))$  in (\ref{eq:QngdDefinitionlambda})
  with $P_0$  taking the role of $P$. 
In Example~\ref{explQngd}, (\ref{eq:forExaUnionQthetaisQ}),  it simply means to take $C$ as $C^0+\Theta$.  

b) Typical examples for good-deal constraints are radial, i.e.\ $C^0$ is a ball. This case is predominant in the literature and justified from a finance point of view by ensuring a constant bound on instantaneous Sharpe ratios (or growth rates). But typical examples for uncertainty (ambiguity) constraints $\Theta$ can well be non-radial \citep[see][]{ChenEpstein,EpsteinS03}. 
For instance, $\Theta$ may arise from a confidence region for some unknown drift parameters in a multivariate (log-)normal model; such would in general be ellipsoidal but not radial, and the sum $C^0+\Theta$ can even be non-ellipsoidal. To offer a suitable framework for such and other examples,   Section~\ref{sec:Prelim} treats abstract correspondences.
A constructive method to solve for such a typical parametrization of $C^0+\Theta$ is described in Remark~\ref{Rem:ellipUnc}.
\end{remark}

\subsection{Robust approach to good-deal hedging under model uncertainty}
As in Section \ref{sec:No-Good-Deal-HedgingApproach} (cf.\ (\ref{eq:PngdDefinition}) and the definition of $\mathcal{Q}^{\textrm{ngd}}(P^\theta)$), we define for $\theta\in \Theta$ the 
set $$\mathcal{P}^{\mathrm{ngd}}(P^\theta) :=\left\lbrace Q\sim P^\theta\ \,\lvert \,\ dQ/dP^\theta=\mathcal{E}\left(\lambda\cdot W^{\theta}\right),\  \lambda\in C^\theta \textrm{ predictable, bounded}\right\rbrace$$ 
in order to introduce a robust notion of good-deal hedging. Let $\mathcal{P}^{\mathrm{ngd}}$ be the smallest m-stable convex set of measures $Q\sim P_0$ containing all $\mathcal{P}^{\mathrm{ngd}}(P^\theta)$,  
$\theta\in \Theta$. Then 
\begin{equation*}
	\rho_t(X) := \esssup_{Q\in \mathcal{P}^{\mathrm{ngd}}} E^Q_t[X],\quad t\in [0,T], 
\end{equation*}
 defines a  time-consistent  dynamic coherent risk measure  by Lemma \ref{lem:PropDynamicRiskMeasMeandQngdareMstable}
(for $X\in L^2(P_0)$ if the correspondence {from} Assumption \ref{asp:AssumptiononCnus} is uniformly bounded, or for $X\in L^\infty$ {otherwise}). As in Section~\ref{sec:No-Good-Deal-HedgingApproach}, the task of good-deal hedging under uncertainty is posed as a minimization problem 
(\ref{eq:HedgingProblemUncertainty}) for a-priori risk measures $\rho$ of hedging errors: For a 
contingent
 claim $X$, find a strategy $\phi^*\in\Phi$ such that 
\begin{equation}\label{eq:HedgingProblemUncertainty}
	\pi^u_t(X)=\rho_t\Big(X-\int_t^T{\phi^*_s}^{\textrm{tr}}d\widehat{W}^0_s\Big)=\essinf_{\phi\in\Phi}\ \rho_t\Big(X-\int_t^T\phi^{\textrm{tr}}_sd\widehat{W}^0_s\Big) ],\quad t\in [0,T].
\end{equation}
The good-deal hedging strategy under uncertainty is defined as this minimizer (if it exists) $\phi^*\in\Phi$. For $X\in L^2(P_0)$, one can prove (as in Prop.\ref{pro:RobustTrackingError}) 
that the tracking error $R^{\phi^*}(X)$ of the strategy $\phi^*$ is a supermartingale under every measure in $\mathcal{P}^{\textrm{ngd}}$: 
\begin{proposition}\label{pro:RobustTrackingErrorUncertainty}
Let $X\in L^2(P_0)$ {with} the correspondence defined in Assumption \ref{asp:AssumptiononCnus} being uniformly bounded. 
Then the tracking error $R^{\phi^*}(X)$ of a strategy $\phi^*$ solving (\ref{eq:HedgingProblemUncertainty}) is a $Q$-supermartingale for all $Q\in \mathcal{P}^{\textrm{ngd}}$.
\end{proposition}
A strategy solving the good-deal hedging problem under uncertainty and whose tracking error satisfies the supermartingale property under all measures in $\mathcal{P}^{\textrm{ngd}}$
(as in Proposition \ref{pro:RobustTrackingErrorUncertainty}) will be qualified as {\em robust} with respect to uncertainty. Note that this is a different notion of robustness compared to the 
one in Remark~\ref{rem:atleastmeanselffinancing}, because the supermartingale property has to hold for measures in $\mathcal{P}^{\textrm{ngd}}(P^\theta)$ uniformly for all models 
$P^\theta\in\mathcal{R}$ (since $\bigcup_{\theta\in\Theta}\mathcal{P}^{\textrm{ngd}}(P^\theta)$ is a subset of $\mathcal{P}^{\textrm{ngd}}$). 
More concrete results under uncertainty will be derived next under additional conditions.

\subsection{Hedging under model uncertainty for ellipsoidal good-deal constraints
}\label{sec:GDValandHedgingUnderUncertaintyEllipsoids}

In this section we consider ellipsoidal good-deal constraints.  To this end, let
\begin{equation}\label{eq:ellipsoidalC0}
 C^0_t(\omega) = \left\lbrace x\in\R^n\ \Big\lvert\Big.\ x^{\textrm{tr}}A_t(\omega)x\le h^2_t(\omega)\right\rbrace,\quad (t,\omega)\in[0,T]\times \Omega,
\end{equation}
where $A$ is a uniformly elliptic and predictable matrix-valued process, and $h$ some positive bounded and predictable process. We assume that $A$ satisfies the separability condition 
(\ref{eq:AssumptionOptimizationSimplify}) with respect to  $\sigma$. 
Let $\Theta$ be an arbitrary standard correspondence satisfying the uniform boundedness Assumption~\ref{asp:BddCorr} and  $0\in\Theta$. As in Example~\ref{exa:ExampleExplicitNGDCorr}, we 
let $C^\theta := C^0$, for all $\theta\in \Theta$, yielding by (\ref{eq:DefCtildeNu}) that 
\begin{equation*}
 \widetilde{C}^\theta_t\cap \mathrm{Ker}\ \sigma_t =  \left\lbrace x\in\R^n\Big\lvert\Big.\ x^{\textrm{tr}}A_tx\le h^2_t-{\xi_t^\theta}^\textrm{tr}A_t\xi_t^\theta\right\rbrace\bigcap \mathrm{Ker}\ \sigma_t +\Pi^\bot_t(\theta_t).
\end{equation*}
Clearly, $C^\theta$ is standard and satisfies Assumption \ref{asp:BddCorr} for $\theta\in\Theta$. Similarly to (\ref{eq:AssumptionOnAlphaPrime}), to derive explicit BSDE formulations 
for solving the hedging problem we will assume
\begin{equation}\label{eq:AssumptionOnAlphaPrimeUncertainty}
         	\lvert \xi^\theta\rvert<h\sqrt{\alpha'} \quad \text{for all }\theta\in\Theta,
\end{equation}
where the process $\alpha'$ is the constant of ellipticity of $A^{-1}$ as in Lemma \ref{lem:qPrimePositiveDefinitePtwise}. Recall that, thanks to Lemma \ref{lem:qPrimePositiveDefinitePtwise}, the 
inequality (\ref{eq:AssumptionOnAlphaPrimeUncertainty}) implies in particular that $-\xi^\theta\in C^0,\ \theta\in\Theta$; hence (\ref{eq:xithetainCtheta}) holds and  the correspondences 
$\widetilde{C}^\theta\cap \mathrm{Ker}\ \sigma$ are standard, $\theta\in \Theta$. Note that condition (\ref{eq:AssumptionOnAlphaPrimeUncertainty}) ensures applicability of 
Lemma \ref{lem:OptimizationWithPhiSolved} in our current setup for any model $P^\theta$.
Since $C^\theta$ is equal to $C^0$ and satisfies Assumption \ref{asp:BddCorr}, one has 
\begin{equation}\label{eq:PngdNUdefinition}
\mathcal{P}^{\mathrm{ngd}}(P^\theta)=\Big\lbrace Q\sim P_0\ \Big\lvert \Big.\ dQ/dP_0=\mathcal{E}\big(\lambda\cdot W^0\big),\  
                                                                         \lambda\ \textrm{predictable,}\ \lambda\in C^0+\theta\Big\rbrace.
\end{equation}
The following lemma can be shown similarly to the proof of part b) of Lemma \ref{lem:PropDynamicRiskMeasMeandQngdareMstable} using (\ref{eq:PngdNUdefinition}). For details see \cite[][Lemma 3.26 and its proof]{KentiaPhD}.
\begin{lemma}\label{lem:UnionPngdQngd-mstable}
\hspace{2em}
\begin{enumerate}
 \item The set $\bigcup_{\theta\in \Theta}\mathcal{P}^{\mathrm{ngd}}(P^\theta)$ is m-stable, convex and equal to $\mathcal{P}^{\mathrm{ngd}}$.
 \item The set $\bigcup_{\theta\in \Theta}\mathcal{Q}^{\mathrm{ngd}}(P^\theta)$ is m-stable, convex and equal to $\mathcal{Q}^{\mathrm{ngd}}$.
\end{enumerate}
\end{lemma}
Thanks to Lemma \ref{lem:UnionPngdQngd-mstable}, the dynamic risk measure $\rho$ satisfies for $X\in L^2(P_0)$
\begin{equation*}
\rho_t(X) := \esssup_{Q\in \mathcal{P}^{\mathrm{ngd}}} E^Q_t[X]= \esssup_{\theta\in \Theta}\rho^{\theta}_t(X), \quad   t\in[0, T],
\end{equation*}
with $\rho^{\theta}_t(X) := \esssup_{Q\in \mathcal{P}^{\mathrm{ngd}}(P^\theta)} E^Q_t[X]$. The worst-case upper good-deal bound $\pi^u_t(X)$ for $X\in L^2(P_0)$ rewrites from 
Definition \ref{defn:SetOfNGDMeas} as 
\begin{equation}\label{eq:DefinitionOfPiu}
	\pi^u_t(X):= \esssup_{\theta\in \Theta}\esssup_{Q\in \mathcal{Q}^{\mathrm{ngd}}(P^\theta)} E^Q_t[X]=\esssup_{\theta\in \Theta}\pi^{u,\theta}_t(X),\  t\in[0,T],
\end{equation}
where $\pi^{u,\theta}_t(X) = \esssup_{Q\in \mathcal{Q}^{\mathrm{ngd}}(P^\theta)} E^Q_t[X]$. The corresponding lower bound $\pi^l_\cdot(X)$ is obtained via $\pi^l_\cdot(X) =-\pi^u_\cdot(-X)$. 
For a worst-case approach to uncertainty we will investigate valuation of claims according to $\pi^u_\cdot(\cdot)$ and hedging with the optimal trading strategy solution to (\ref{eq:HedgingProblemUncertainty}). 
We employ results from Section \ref{subsec:Example1withEllipsoids} (under $P=P^\theta$) to characterize $\pi^{u,\theta}_\cdot(X)$ as well as the associated hedging strategies $\bar{\phi}^{\theta}$ in $\Phi$. 
For $\theta\in \Theta$ and $\phi\in\Phi$ let us consider the classical BSDEs 
 \begin{align}
  -dY_t &= f^{\phi,\theta}(t,Z_t)dt -Z_t^\textrm{tr}dW^0_t,\quad t\le T,\quad Y_T=X,\label{eq:BSDEYPhiNu}\\
  -dY_t &= f^\theta(t,Z_t)dt-Z_t^{\textrm{tr}}dW^0_t,\quad  t\le T,\quad Y_T=X,\label{eq:BSDEPiuNu}
 \end{align}
with respective generators 
\begin{align}\label{eq:generatorfphitheta}
	f^{\phi,\theta}(t,z) :=\:\,&\theta_t^\textrm{tr}(z-\phi_t)-{\xi^0_t}^\textrm{tr}\phi_t + h_t\big((z-\phi_t)^\textrm{tr}A^{-1}_t(z-\phi_t)\big)^{1/2}, \\
\label{eq:generatorfTheta}
  f^\theta(t,z)  := \:\,&\quad\; \Pi^\bot_t(\theta_t)^\textrm{tr}\Pi^\bot_t(z)-{\xi^0_t}^{\textrm{tr}}\Pi_t(z) \\
 \nonumber & + \big(h_t^2-{\xi_t^{\theta}}^{\textrm{tr}}A_t\xi_t^{\theta}\big)^{1/2}\big({\Pi_t^\bot(z)}^{\textrm{tr}}A_t^{-1} \Pi_t^\bot(z)\big)^{1/2}.
 \end{align}
It is straightforward to derive the BSDE descriptions for $\pi^{u,\theta}_\cdot(X)$ and $\rho^\theta_\cdot(X)$ stated in the subsequent proposition. The proof is analogous to that for Theorem~\ref{thm:FullHedgingProblemSolved},
using (\ref{eq:AssumptionOnAlphaPrimeUncertainty}) instead of (\ref{eq:AssumptionOnAlphaPrime}), replacing $P$ by $P^\theta$ and changing measure from $P^\theta$ to $P_0$.
\begin{proposition}\label{pro:BSDEspiuNuandRhoNuPhi}
 Assume (\ref{eq:AssumptionOptimizationSimplify}) and (\ref{eq:AssumptionOnAlphaPrimeUncertainty}) hold. For $X\in L^2(P_0)$, $\theta\in \Theta$ and $\phi\in\Phi$, let $(Y^{\phi,\theta},Z^{\phi,\theta})$ and 
 $(Y^\theta,Z^\theta)$ be the standard solutions to the BSDEs (\ref{eq:BSDEYPhiNu}) and (\ref{eq:BSDEPiuNu}) respectively.
Then 
\( \pi^{u,\theta}_t(X) = Y^\theta_t= E^{\bar{Q}^\theta}_t[X],\ t\in [0,T],
\)
holds with $\bar{Q}^\theta\in \mathcal{Q}^{\mathrm{ngd}}(P^{\theta})$ given by $d\bar{Q}^\theta/dP_0=\mathcal{E}\left((-\xi^0+{\bar{\eta}}^{\theta})\cdot W^0\right)$ for
\begin{equation*}
 {\bar{\eta}}^{\theta}_t = \big(h_t^2-{\xi_t^{\theta}}^{\textrm{tr}}A_t\xi_t^{\theta}\big)^{1/2}\big({\Pi_t^\bot(Z^\theta_t)}^{\textrm{tr}}A_t^{-1} \Pi_t^\bot(Z^\theta_t)\big)^{-1/2}A_t^{-1}\Pi_t^\bot(Z^\theta_t)+\Pi^\bot_t(\theta_t).
\end{equation*}
Moreover $Y^{\phi,\theta}_t = \rho^\theta_t\left(X-\int_t^T\phi^\textrm{tr}_sd\widehat{W}^0_s\right)$ holds, and the strategy $\bar{\phi}^{\theta}$ (in $\Phi$) with
\begin{equation*}
 \bar{\phi}^{\theta}_t := \Pi_t(Z^\theta_t)+\big(\Pi^\bot_t(Z^\theta_t)^{\textrm{tr}}A^{-1}_t\Pi^\bot_t(Z^\theta_t)\big)^{1/2}\big(h^2_t-{\xi^{\theta}_t}^{\textrm{tr}}A_t \xi^{\theta}_t\big)^{-1/2}A_t\xi^{\theta}_t\ P\otimes dt\text{-a.e.} 
\end{equation*}
 satisfies $\displaystyle \pi^{u,\theta}_t(X)= \rho^\theta_t\Big(X-\int_t^T(\bar{\phi}^\theta_s)^{\textrm{tr}}d\widehat{W}^0_s\Big)=\essinf_{\phi\in\Phi}\rho^\theta_t\Big(X-\int_t^T\phi^\textrm{tr}_sd\widehat{W}^0_s\Big).$
\end{proposition}
By Proposition \ref{pro:BSDEspiuNuandRhoNuPhi}, we can write $\pi^u_\cdot(X)$ from (\ref{eq:DefinitionOfPiu}) as  
\begin{equation}\label{eq:SupInfNGDBoundNaive}
	\pi^{u}_t(X)= \esssup_{\theta\in \Theta}\ \essinf_{\phi\in \Phi} \rho^\theta_t\Big(X-\int_t^T\phi^{\textrm{tr}}_sd\widehat{W}^0_s\Big),\qquad  t\in [0,T].
\end{equation}
This permits to describe $\pi^u_\cdot(X)$ and the associated hedging strategy $\bar{\phi}$ in the next theorem  by the solution to the classical BSDE 
\begin{equation}\label{eq:BSDEsforPiuunc}
-dY_t= f(t,Z_t)dt - Z^{\textrm{tr}}_t dW^{0}_t,\ t\le T\quad \textrm{and}\quad  Y_T=X,
\end{equation}
with generator $f(t,\omega,z):=\sup_{\theta\in\Theta}f^\theta(t,\omega, z)$, $z\in\R^n$, $(t,\omega)\in [0,T]\times\Omega$, with $f^\theta$ given by (\ref{eq:generatorfTheta}). The theorem moreover identifies by $\bar{\theta}$ the worst-case model $P^{\bar{\theta}}\in \mathcal{R}$. 
\begin{theorem}\label{thm:SolveUncertainValuationpb}
	Assume (\ref{eq:AssumptionOptimizationSimplify}) and (\ref{eq:AssumptionOnAlphaPrimeUncertainty}) hold. For $X\in L^2(P_0)$, let $(Y,Z)$ be the standard solution to the
	BSDE (\ref{eq:BSDEsforPiuunc}). Then there exists a unique predictable selection  $\bar{\theta}:=\bar{\theta}(X)$ of $\Theta$ satisfying $f^{\bar{\theta}}(t,Z_t) = \esssup_{\theta\in\Theta}f^\theta(t,Z_t),\ P\otimes dt$-a.e.\ 
	such that for all $t\in [0,T]$
\begin{equation}\label{eq:StandardGDBoundandHedging}
    		\pi^{u}_t(X) =\rho^{\bar{\theta}}_t\Big(X-\int_t^T{\bar{\phi}_s}^{\textrm{tr}}d\widehat{W}^0_s\Big)= \pi^{u,\bar{\theta}}_t(X) = Y_t
\end{equation}
 holds with $\bar{\phi}=(\bar{\phi}_t)_{t\in[0,T]}:=\bar{\phi}^{\bar{\theta}}(X)\in\Phi$ given $P\otimes dt$-a.e.\ by 
\begin{equation}\label{eq:PhiBarExpressionGeneralTheta}
 \bar{\phi}_t = \Pi_t(Z_t)+\big(\Pi^\bot_t(Z_t)^{\textrm{tr}}A^{-1}_t\Pi^\bot_t(Z_t)\big)^{1/2}\big(h^2_t-{\xi^{\bar{\theta}}_t}^{\textrm{tr}}A_t \xi^{\bar{\theta}}_t\big)^{-1/2}A_t\xi^{\bar{\theta}}_t.
\end{equation}
The tracking error $R^{\bar{\phi}}(X):=\pi^{u}_\cdot(X)-\pi^{u}_0(X)  - \bar{\phi}\cdot\widehat{W}^0$ of the strategy $\bar{\phi}$ is a supermartingale under any $Q$ in $\mathcal{P}^{\textrm{ngd}}(P^{\bar{\theta}})$, and 
is a martingale under $\bar{Q}$ in $\mathcal{P}^{\textrm{ngd}}(P^{\bar{\theta}})$ given by $d\bar{Q}/dP_0 = \mathcal{E}(\bar{\lambda}\cdot W^0)$ with 
\begin{equation*}
\bar{\lambda} := h\big((Z-\bar{\phi})A^{-1}(Z-\bar{\phi})\big)^{-1/2} A^{-1}\left(Z-\bar{\phi}\right)+\bar{\theta}. 
\end{equation*}
\end{theorem}
\begin{proof}
Pointwise existence and uniqueness of $\bar{\theta}\in\Theta$ follow by the continuity and strict concavity of $f^{\theta}$ as a function of $\theta\in\R^n$, and the uniform boundedness of 
$\Theta$. Predictability of $\bar{\theta}$ follows by \citet{Rockafellar}. The claims (\ref{eq:StandardGDBoundandHedging}) and (\ref{eq:PhiBarExpressionGeneralTheta}) are corollaries of 
Proposition \ref{pro:BSDEspiuNuandRhoNuPhi}. The remaining claims are similar to those of Theorem~\ref{thm:FullHedgingProblemSolved}, hence their proof goes likewise,
making again use of Lemma~\ref{lem:OptimizationWithPhiSolved} \citep[instead of][Lem.6.1]{Becherer} and (\ref{eq:AssumptionOptimizationSimplify}) 
and (\ref{eq:AssumptionOnAlphaPrimeUncertainty}).
\end{proof}

The process $\bar{\phi}:=\bar{\phi}^{\bar{\theta}}$ in Theorem \ref{thm:SolveUncertainValuationpb} is the good-deal hedging strategy of $X$ for the worst-case model  $P^{\bar{\theta}}\in\mathcal{R}$ which yields that highest good-deal valuation with $\pi^u_\cdot(X)=\pi^{u,\bar{\theta}}_\cdot(X)$. The tracking error of $\bar{\phi} $ is therefore a supermartingale under any measure in $\mathcal{P}^{\textrm{ngd}}(P^{\bar{\theta}})$ (cf.\  Proposition \ref{pro:RobustTrackingError}), i.e.\  
$\bar{\phi}$ is ``at least mean-self-financing'' under any measure in $\mathcal{P}^{\textrm{ngd}}(P^{\bar{\theta}})$. However, it is not clear at this stage whether the supermartingale 
property of the tracking error of $\bar{\phi}$ holds simultaneously under all measures in $\mathcal{P}^{\textrm{ngd}}(P^{\theta})$ for all models $\mathcal{R}=\{P^\theta : \theta\in\Theta\}$.
We will show that this is the case, and that $\bar{\phi}$ and its associated valuation bound $\pi^u_\cdot(X)$ are indeed robust with respect to uncertainty. The idea is first to find an 
alternative bound $\pi^{u,*}_\cdot(\cdot)$ and an associated  strategy $\phi^*$ that satisfy the supermartingale property of the tracking error simultaneously under all measures in 
$\bigcup_{\theta\in\Theta}\mathcal{P}^{\textrm{ngd}}(P^{\theta})$ and are therefore robust. After this, we show that $\pi^{u,*}_\cdot(X)$ coincides with the worst-case bound $\pi^u_\cdot(X)$, 
and that the same holds for the hedging strategies $\bar{\phi}(X)$ and $\phi^*(X)$ for any contingent claim $X$. In general the good-deal bound $\pi^u_\cdot(X)$ is dominated by $\pi^{u,*}_\cdot(X)$, 
but thanks to a saddle point result (Theorem \ref{thm:CharactOfSaddlePt}) one can actually prove that the two bounds are identical. Exchanging the order between $\esssup$ and 
$\essinf$ in the expression (\ref{eq:SupInfNGDBoundNaive}) for $\pi^u_\cdot(X)$, we define for $X\in L^2(P_0)$ and $t\le T$ 
\begin{equation}\label{eq:InfSupNGDBoundRobust}
	\pi^{u,*}_t(X):= \essinf_{\phi\in \Phi}\ \esssup_{\theta\in \Theta} \rho^\theta_t\Big(X-\int_t^T\phi^{\textrm{tr}}_sd\widehat{W}^0_s\Big).
\end{equation}
From this it is clear that in general $\pi^{u,*}_t(X)\ge\pi^u_t(X)$, for all $X\in L^2(P_0)$. We will show that in fact the minimax identity holds in the sense that the expressions in 
(\ref{eq:SupInfNGDBoundNaive}) and (\ref{eq:InfSupNGDBoundRobust}) coincide, and that a saddle point exists, giving equality of $\pi^u_\cdot(X)$ and $\pi^{u,*}_\cdot(X)$. 
 To this end, we describe $\pi^{u,*}_\cdot(X)$ and $\phi^*$ in terms of the standard solution $(Y,Z)$ for the classical BSDE 
\begin{equation}\label{eq:BSDEsforPiBaru}
-dY_t= f^*(t,Z_t)dt - Z^{\textrm{tr}}_t dW^{0}_t,\quad t\le T\quad \textrm{and}\quad  Y_T=X,
\end{equation}
where the generator $f^*$ satisfies $P\otimes dt$-a.e.\ $f^*(t,z)=\essinf_{\phi\in\Phi}f^\phi(t,z)$ for all $z\in\R^n,$ with $f^\phi(\cdot,\cdot,z):=\sup_{\theta\in\Theta} f^{\phi,\theta}(\cdot,\cdot,z)$ for $f^{\phi,\theta}$ from (\ref{eq:generatorfphitheta}).
Indeed such a generator function $f^*$ can be defined at first $P\otimes dt$-a.e.\ for each $z\in\Q^n$, and by Lipschitz continuous extension then $P\otimes dt$-a.e.\ for all $z\in\R^n$.
We have 
\begin{equation}\label{eq:GeneratorfPhi}
f^\phi(t,z) =-{\xi^0_t}^\textrm{tr}\phi_t+\sup_{\theta\in\Theta}\theta_t^{\text{tr}}(z-\phi_t)+ h_t\big((z-\phi_t)^\textrm{tr}A^{-1}_t(z-\phi_t)\big)^{1/2},
\end{equation}
and we can identify the robust good-deal hedging strategy $\phi^*$ by
\begin{proposition}\label{pro:RobustValandHedging}
Assume (\ref{eq:AssumptionOptimizationSimplify}) and (\ref{eq:AssumptionOnAlphaPrimeUncertainty}) hold. For $X\in L^2(P_0)$, let $(Y,Z)$ be the standard solution to the BSDE (\ref{eq:BSDEsforPiBaru}). 
Then  there exists a unique $\phi^*\in\Phi$ satisfying $f^{\phi^*}(t,Z_t) = \essinf_{\phi\in\Phi}f^\phi(t,Z_t)\ P\otimes dt$-a.e.\ such that, for $ t\in [0,T]$,
\begin{equation}\label{eq:RobustGDBoundandHedging}
    		\pi^{u,*}_t(X) = \essinf _{\phi\in \Phi}\rho_t\Big(X-\int_t^T\phi_s^{\textrm{tr}}d\widehat{W}^0_s\Big)=\rho_t\Big(X-\int_t^T{\phi^*_s}^{\textrm{tr}}d\widehat{W}^0_s\Big) = Y_t.
\end{equation}
Moreover $R^{\phi^*}(X):=\pi^{u,*}_\cdot(X)-\pi^{u,*}_0(X)  - \phi^*\cdot\widehat{W}^0$ is a $Q$-supermartingale for all $Q\in\mathcal{P}^{\textrm{ngd}}$, and a $Q^*$-martingale for 
$Q^*\in \mathcal{P}^{\textrm{ngd}}$ with $dQ^*/dP_0 = \mathcal{E}(\lambda^*\cdot W^0)$, where
\begin{equation*}
\lambda^* = h\big((Z-\phi^*)^{\textrm{tr}}A^{-1}(Z-\phi^*)\big)^{-1/2}A^{-1}\left(Z-\phi^*\right)+\theta^*,
\end{equation*}
with $\theta^*:=\theta^*(\phi^*)\in\Theta$ satisfying $ {\theta^*_t}^{\text{tr}}(Z_t-\phi^*_t)=\esssup_{\theta\in\Theta}\theta_t^{\text{tr}}(Z_t-\phi^*_t)\ P\otimes dt$-a.e.\ such that $f^*(t,Z_t)=f^{\phi^*,\theta^*}(t,Z_t),\ P\otimes dt$-a.e..
\end{proposition}
\begin{proof}
By \cite[][Thms.2.K and 1.C]{Rockafellar} one has for any $\phi\in\Phi$ and $z\in\R^n$ that there exists $\theta^*(\phi)=\theta^*(\phi,z)\in\Theta$ such that ${\theta^*_t(\phi)}^\textrm{tr}(z-\phi_t) = \sup_{\theta\in \Theta}\theta_t^\textrm{tr}(z-\phi_t)$
for all $(t,\omega)$ and hence $f^{\phi}(t,\omega,z)=f^{\phi,\theta^*(\phi)}(t,\omega,z)$ for all $(t,\omega)$. Consider the convex continuous function 
$$\R^n\ni \phi\mapsto F(\phi):=  -{\xi^0}^{\textrm{tr}}\phi+\sup_{\theta\in\Theta}\theta^{\text{tr}}(z-\phi)+h\big((z-\phi)^{\textrm{tr}}A^{-1} (z-\phi)\big)^{1/2},$$ for constant 
$h,\phi$,$z,\xi^0$,$\sigma$ and $A$ satisfying the notations of Lemma \ref{lem:OptimizationWithPhiSolved} and for a compact set $\Theta\subset \R^n$ containing the origin. The function 
$F$ is also coercive on $\mathrm{Im}\ \sigma^{\text{tr}}$, i.e.\  $F(\phi)\to +\infty$ as $\abs{\phi}\to +\infty$ for $\Pi^\bot(\phi)=0$ because $\abs{\xi^0}<h\sqrt{\alpha'}$ and 
$\sup_{\theta\in\Theta}\theta^{\text{tr}}(z-\phi)\ge0$. Hence existence of $\phi^*\in\Phi$ follows from \citep[][Ch.II, Prop.1.2]{EkelandTemam}. Uniqueness of $\phi^*$ follows from 
the fact that $F$ is strictly convex over $\left\lbrace\Pi^\bot(\phi)=0\right\rbrace$ if $\Pi^\bot(z)\neq 0$ and strictly convex at $\phi=z$ if $\Pi^\bot(z)= 0$ because (\ref{eq:AssumptionOnAlphaPrimeUncertainty}) 
holds. Finally, predictability of $\phi^*$ follows from \citep[][Thm.2.K]{Rockafellar} via Part \ref{ImAndKerCorrespondencesPredictable} of Proposition \ref{pro:ImAndKerCorrespondencesPredictableCharactOfMe}. 

From Proposition \ref{pro:BSDEspiuNuandRhoNuPhi}, for $\phi\in\Phi$ and $\theta\in \Theta$, $Y^{\phi,\theta}=\rho^\theta_\cdot\big(X-\int_\cdot^T\phi_s^{\textrm{tr}}d\widehat{W}^0_s\big)$
is the $Y$-component of the solution to the classical BSDE (\ref{eq:BSDEYPhiNu}). As a consequence for every $\phi\in \Phi$ it holds $\esssup_{\theta\in \Theta} f^{\phi,\theta}(t,Z_t) = f^{\phi,\theta^*(\phi)}(t,Z_t)=f^\phi(t,Z_t),\ P\otimes dt$-a.e..
The generators $f^\phi$ are standard, so that by the comparison theorem for classical BSDEs,  $(Y^\phi,Z^\phi)$ with $Y^\phi_t:=\esssup_{\theta\in \Theta}Y^{\phi,\theta}_t$ is the 
standard solution to the BSDEs (under $P_0$) with parameters $(f^\phi,X)$, for $\phi\in \Phi$. The generator $f^*$ is also standard because $f^*(t,z)=f^{\phi^*,\theta^*(\phi^*)}(t,z)=\essinf_{\phi\in\Phi}f^\phi(t,z)$ for all $z$, $P\otimes dt$-a.e.. 
Now the comparison theorem yields (\ref{eq:RobustGDBoundandHedging}) from (\ref{eq:InfSupNGDBoundRobust}).

The supermartingale property of $R^{\phi^*}(X)$ under any $Q\in\mathcal{P}^{\textrm{ngd}}$ follows from (\ref{eq:RobustGDBoundandHedging}) using arguments as in the proof of Proposition \ref{pro:RobustTrackingError}.
The martingale property under $Q^*$ is proved similarly, noting that the finite variation part under $Q^*$ of $R^{\phi^*}(X)$ vanishes since $f^*(t,Z_t)=f^{\phi^*,\theta^*}(t,Z_t)\ P\otimes dt$-a.e..
\end{proof}
Proposition~\ref{pro:RobustValandHedging} shows that  the tracking error of the hedging strategy $\phi^*$ with respect to valuation according to $\pi^{u,*}_\cdot(X)$ has the supermartingale property
simultaneously  under all measures in $\mathcal{P}^{\textrm{ngd}}=\bigcup_{\theta\in\Theta}\mathcal{P}^{\textrm{ngd}}(P^\theta)$.
The next theorem shows that a minimax identity holds: the sup-inf representation of $\pi^u_\cdot(\cdot)$ in (\ref{eq:SupInfNGDBoundNaive}) is equal to the inf-sup 
representation of $\pi^{u,*}_\cdot(\cdot)$ in (\ref{eq:InfSupNGDBoundRobust}); see also (\ref{eq:MinimaxGen}). 
Moreover, the good-deal hedging strategy $\bar \phi$ with respect to the worst-case model (given by $\bar \theta$) that gives the highest good-deal valuation bound  $\pi^u_\cdot(\cdot)$, is identical with the robust good-deal hedging strategy $\phi^*$
from Proposition~\ref{pro:RobustValandHedging}.

\begin{theorem}\label{thm:CharactOfSaddlePt}
  Assume (\ref{eq:AssumptionOptimizationSimplify}) and (\ref{eq:AssumptionOnAlphaPrimeUncertainty}) hold. For $X\in L^2(P_0)$, let $(Y,Z)$ be standard solution of the BSDE (\ref{eq:BSDEsforPiBaru}). Then 
  $P\otimes dt$-almost everywhere holds
  \begin{equation}\label{eq:MinimaxGen}
f^{\phi^*,\theta^*}(t,Z_t)=\essinf_{\phi\in\Phi}\ \esssup_{\theta\in\Theta}f^{\phi,\theta}(t,Z_t) = \esssup_{\theta\in\Theta}\ \essinf_{\phi\in\Phi}f^{\phi,\theta}(t,Z_t) =f^{\bar{\phi},\bar{\theta}}(t,Z_t),
\end{equation}
with $(\bar{\phi},\bar{\theta})$, $(\phi^*,\theta^*)$ from Theorem \ref{thm:SolveUncertainValuationpb} and Proposition \ref{pro:RobustValandHedging}. Moreover $(Y,Z)$ coincides 
with the standard solution to the BSDE (\ref{eq:BSDEsforPiuunc}) and 
\begin{equation}\label{eq:SaddlPtStard=RobustPhiStar=PhiBar}
\pi^u_t(X)=\pi^{u,*}_t(X)=Y_t,\ t\in[0,T]\quad \textrm{and}\quad \phi^*_t(X)=\bar{\phi}_t(X),\ P\otimes dt\text{-a.e.}.
\end{equation}
\end{theorem}
\begin{proof}
Let $X\in L^2(P_0)$. By an application of Lemma \ref{lem:MinimaxPbSolved}, the generator $f^{\phi,\theta}$ of the BSDE (\ref{eq:BSDEYPhiNu}) for $\theta\in\Theta$ and $\phi\in\Phi$ 
satisfy the minimax relation (\ref{eq:MinimaxGen}). By Theorem \ref{thm:SolveUncertainValuationpb} and Proposition \ref{pro:RobustValandHedging} it holds $f(t,Z_t) = f^{\bar{\phi},\bar{\theta}}(t,Z_t)$ 
and $f^*(t,Z_t) = f^{\phi^*,\theta^*}(t,Z_t)$, $P\otimes dt$-a.e., for $f,f^*$ respectively generators of the BSDEs (\ref{eq:BSDEsforPiuunc}), (\ref{eq:BSDEsforPiBaru}). Also, $\pi^u_t(X)=\pi^{u,*}_t(X)=Y_t,\ t\in[0,T],$ 
since by uniqueness of BSDE solutions $(Y,Z)$ also solves the BSDE (\ref{eq:BSDEsforPiuunc}). Hence $(\bar{\phi},\bar{\theta})$ and $(\phi^*,\theta^*)$ are both saddle points of the function 
$(\phi_t,\theta_t)\mapsto f^{\phi,\theta}(t,Z_t)$. Now for any $\theta\in\Theta$ and $z\in\R^n$, the function $\phi\mapsto F(\phi,\theta):= \theta^\textrm{tr}(z-\phi)-{\xi^0}^\textrm{tr}\phi + h\big((z-\phi)^\textrm{tr}A^{-1}(z-\phi)\big)^{1/2}$ 
is strictly convex over $\{\Pi^\bot(\phi)=0\}$ if $\Pi^\bot(z)\neq 0$, and strictly convex at $\phi=z$ if $\Pi^\bot(z)= 0$, since $\abs{\xi^\theta}<h\sqrt{\alpha'}$. 
\citet[][Ch.VI, Prop.1.5]{EkelandTemam} implies that the $\phi$-components of the saddle points are identical, yielding $\bar{\phi}=\phi^*$.   
\end{proof}

\subsection{The impact of model uncertainty on robust good-deal hedging}\label{subsec:ImpactUncertaintyGDHedging}
In the framework of Section~\ref{sec:GDValandHedgingUnderUncertaintyEllipsoids}, results have so far been stated for an arbitrary standard correspondence $\Theta$ without further structural 
assumptions, and ellipsoidal correspondences were only assumed for the no-good-deal restrictions $C^\theta$, $\theta\in\Theta$. Recall (cf.\ Theorem \ref{thm:FullHedgingProblemSolved} 
and subsequent remarks) that in the absence of uncertainty the good-deal hedging strategy contains a speculative component in the direction of the market price of risk. This already 
indicates that under uncertainty one should expect to see relevant differences by a robust approach to hedging. 
We investigate the effect of uncertainty  (solely) about the market 
prices of risk $\xi^\theta$ on robust good-deal hedging,  assuming in addition (note that $\xi^\theta \in \mathrm{Im }\,\sigma^{\text{tr}}$ is natural) that for all $(t,\omega)\in[0,T]\times \Omega$, 
the set  $\Theta_t(\omega)$ is a subset of $\mathrm{Im\ }\sigma^{\text{tr}}_t(\omega)$, i.e.\ that 

\begin{equation}\label{eq:ThetaEllipsoidal}
\Theta_t(\omega) =\Theta^0_t(\omega)\cap \mathrm{Im\ }\sigma^{\text{tr}}_t(\omega)
\end{equation}
holds for some standard correspondence $\Theta^0$ with $0\in\Theta^0$ satisfying the uniform boundedness Assumption~\ref{asp:BddCorr}. With (\ref{eq:ThetaEllipsoidal}), one clearly has 
$\Pi^\bot(\theta)=0$ for all $\theta\in\Theta$, which implies that $\xi^\theta=\xi^0+\theta$ for all $\theta\in\Theta$. This leads to the following simplified expressions of the BSDE generators $f^{\phi,\theta}$, $f^\theta$: 
\begin{align*}
	f^{\phi,\theta}(t,z)&=\theta_t^\textrm{tr}(\Pi_t(z)-\phi_t)-{\xi^0_t}^\textrm{tr}\phi_t + h_t\big((z-\phi_t)^\textrm{tr}A^{-1}_t(z-\phi_t)\big)^{1/2},\\
f^\theta(t,z)&=-{\xi^0_t}^{\textrm{tr}}\Pi_t(z) + \big(h_t^2-{\xi_t^{\theta}}^{\textrm{tr}}A_t\xi_t^{\theta}\big)^{1/2}\big({\Pi_t^\bot(z)}^{\textrm{tr}}A_t^{-1} \Pi_t^\bot(z)\big)^{1/2}.
\end{align*}
As a consequence, the process $\bar{\theta}=\bar{\theta}(X)$ does actually not depend on the contingent claim $X\in L^2(P_0)$ under consideration, and solves the minimization problem 
\begin{equation}\label{eq:ProjectionMinNormnu}
         	{\xi^{\bar{\theta}}_t}^{\textrm{ tr}}A_t\xi^{\bar{\theta}}_t=  \min_{\theta\in \Theta}{\xi^\theta_t}^{\textrm{tr}}A_t\xi^\theta_t,\quad t\in[0,T].
\end{equation}
In addition in this case, one has $\mathcal{Q}^{\text{ngd}}(P^{\bar{\theta}}) = \bigcup_{\theta\in\Theta}\mathcal{Q}^{\text{ngd}}(P^\theta)=\mathcal{Q}^{\text{ngd}}.$ To obtain even more 
explicit results one may assume e.g.\  ellipsoidal uncertainty  
\begin{equation}\label{eq:ellipsoidalTheta0}
\Theta^0_t(\omega):=\left\lbrace x\in \R^n\ \lvert\ x^{\textrm{tr}}B_t(\omega)x\le \delta^2_t(\omega)\right\rbrace \quad \text{for all } (t,\omega)\in[0,T]\times \Omega, 
\end{equation}
with $\delta$ being a positive bounded and predictable process, and $B$ being a uniformly elliptic and predictable matrix-valued process, satisfying the separability condition (\ref{eq:AssumptionOptimizationSimplify}) 
with respect to $\sigma$. Clearly $f^\phi(t,Z_t)$ from (\ref{eq:GeneratorfPhi}) in this case is $P\otimes dt$-a.e.\ equal to
\begin{equation*}
-{\xi^0_t}^{\textrm{tr}}\phi_t+\delta_t\big((\Pi_t(Z_t)-\phi_t)^{\textrm{tr}}B_t^{-1} (\Pi_t(Z_t)-\phi_t)\big)^{1/2}+h_t\big((Z_t-\phi_t)^{\textrm{tr}}A_t^{-1} (Z_t-\phi_t)\big)^{1/2}.	
\end{equation*}
In terms of $\phi^*$ and the solution $(Y,Z)$ to the BSDE (\ref{eq:BSDEsforPiBaru}), the process $\theta^*=\theta^*(\phi^*)$ of Proposition \ref{pro:RobustValandHedging} is given by 
\begin{equation}\label{eq:ThetaStar}
 \theta^*_t(X)=\delta_t\big((\Pi_t(Z_t)-\phi^*_t)^{\textrm{tr}}B^{-1}_t(\Pi_t(Z_t)-\phi^*_t)\big)^{-1/2}B^{-1}_t(\Pi_t(Z_t)-\phi^*_t).
\end{equation}

\begin{remark}\label{Rem:ellipUnc}
Let us recall Remark~\ref{Rem:C0plusTheta}~b). In the present context of Section~\ref{subsec:ImpactUncertaintyGDHedging}
with  constraints of ellipsoidal type for good-deals (\ref{eq:ellipsoidalC0}) and for model uncertainty (\ref{eq:ellipsoidalTheta0}),
results as explicit as in Section~\ref{subsec:Example1withEllipsoids} can be obtained in particular cases, as elaborated subsequently, but not in general.
Indeed, using $\xi^{\theta}=\xi^0+\theta$ for $\theta\in \Theta$,  to find the minimizer $\bar{\theta}$ (the worst-case) in (\ref{eq:ProjectionMinNormnu}) requires to compute the projection of $-\xi_t^0$ onto the ellipsoid  $\Theta_t$
 with respect to the norm induced by the
matrix $A_t$. In the radial case $A\equiv Id_{\R^n}$ the projection is Euclidian. 
While there is no closed formula for the projection in general,  the solution is described by a parametric formula  in terms of a
Lagrangian multiplier that solves a 1-dimensional equation, and it can be  computed by efficient algorithms \citep[see][]{Kiseliov94}, 
which is relevant if it may be required frequently (as in Monte Carlo simulation, cf.\ Section~\ref{sec:MC}).
\end{remark}

It is instructive to look at the special case where in addition the matrices $A$ and $B$ are related through $B=A/r$ for some scalar $r>0$; in other words, $B$ basically equals $A$  up to a change of $\delta$
to $\sqrt{r}\delta$. In this case (\ref{eq:ProjectionMinNormnu}) is solved by
\begin{equation}\label{eq:PhistatAequalBCaseA}
\bar{\theta}_t=-\xi^0_tI_{\left\lbrace{\xi^0_t}^{\textrm{tr}}A_t\xi^0_t\le r\delta^2_t\right\rbrace}-\frac{\sqrt{r}\delta_t}{({\xi^0_t}^{\textrm{tr}}A_t\xi^0_t)^{1/2}}\xi^0_tI_{\left\lbrace{\xi^0_t}^{\textrm{tr}}A_t\xi^0_t> r\delta^2_t\right\rbrace},\quad t\in[0,T], 
\end{equation}
and replacing $\phi^*=\bar{\phi}$ in the formula of $\theta^*$ in (\ref{eq:ThetaStar}) by its expression from (\ref{eq:PhiBarExpressionGeneralTheta}) in terms of $\bar{\theta}$ one obtains $\theta^*=\bar{\theta}$.
Note that (\ref{eq:PhistatAequalBCaseA}) implies that $\bar{\theta}^{\text{tr}}A\bar{\theta}$ is equal to ${\xi^0}^{\textrm{tr}}A\xi^0$ on $\{{\xi^0}^{\textrm{tr}}A\xi^0\le r\delta^2\}$ and equal  to 
$r\delta^2$ on $\{{\xi^0}^{\textrm{tr}}A\xi^0> r\delta^2\}$. In other words, the worst-case Girsanov kernel $-\bar{\theta}$ is equal to the market price of risk $\xi^0$ of the center $P_0$ of 
the confidence set $\mathcal{R}$ of reference measures, being truncated outside a suitable neighborhood $\{{\xi^0}^{\textrm{tr}}A\xi^0\le r\delta^2\}$.

To obtain an intuition about the impact that model uncertainty may have on robust good-deal hedging, let us look at the behavior of the worst-case Girsanov 
kernel $\bar{\theta}=\theta^*$ obtained in (\ref{eq:PhistatAequalBCaseA}) and the hedging strategy $\bar{\phi}=\phi^*$ in (\ref{eq:PhiBarExpressionGeneralTheta}) for varying scaling 
constant $r$: As $r$ becomes large, the worst-case Girsanov kernel $-\bar{\theta}$ becomes close to the market price of risk $\xi^0$ and $\phi^*=\bar{\phi}$ close to $\Pi(Z)$. This 
shows that as uncertainty becomes overwhelming, the robust good-deal hedging strategy ceases to comprise a speculative component in the direction of the market price of risk. In such a 
situation one can show that the hedging strategy is the risk-minimizing  strategy under the worst-case no-good-deal measure in the worst-case model $P^{\bar{\theta}}$. More precisely, 
for an arbitrary shape of the correspondence $\Theta^0$, if uncertainty is big enough for the confidence set $\mathcal{R}$ of reference measures to contain some risk neutral pricing measure 
from $\mathcal{M}^e$, then robust good-deal hedging does not comprise a speculative component and the holdings  $\phi^*$ of a hedging strategy in risky assets coincide with those of the globally risk-minimizing  strategy by 
\citet{FollmerSondermann} \citep[cf.][Sect.2]{Schweizer} under the reference measure for the worst-case valuation (of the claim). In this sense, the  eventually non-speculative nature of the robust 
good-deal hedging strategy under (large) uncertainty offers new theoretical support for the quadratic hedging objective
of risk minimization, which may be criticized for giving equal weighting to upside and downside risk. 
More broadly, it gives support to a common perception \citep[cf.\ e.g.][]{LiouiPoncet} that speculative objectives should be avoided in 
hedging, in addition to practical considerations like simplification of marking-to-market (uses risk neutral valuation). 
To make this precise, let us consider the classical BSDE
\begin{equation}\label{eq:BSDERiskMin}
 -dY_t= \big(-{\xi^0_t}^{\textrm{tr}}\Pi_t(Z_t) + h_t\big({\Pi_t^\bot(Z_t)}^{\textrm{tr}}A_t^{-1} \Pi_t^\bot(Z_t)\big)^{1/2}\big)dt - Z^{\textrm{tr}}_t dW^{0}_t,
\end{equation}
for $t\in [0,T]$ with $Y_T=X$. First we prove the following 
\begin{proposition}\label{pro:SaddPtCondEquivalence}
 Assume (\ref{eq:AssumptionOptimizationSimplify}) and (\ref{eq:AssumptionOnAlphaPrimeUncertainty}) hold and that $\Theta$ satisfies (\ref{eq:ThetaEllipsoidal}). For any $X\in L^2(P_0)$, 
 let $(Y^X,Z^X)$ denote the standard solution of the BSDE (\ref{eq:BSDERiskMin}). Then 
\begin{align}
 & \phantom{\text{xxxxx}\quad}
\label{eq:equalityPiPhi}
 \pi^{u}_\cdot(X)=Y^X\text{ and } \phi^*(X)=\Pi(Z^X)\text{ for all } X\in L^2(P_0)\\
&\label{eq:UncertaintyBigEnough}
 \text{ holds, if and only if }\qquad  \mathcal{R}\cap\mathcal{M}^e(S)\neq\emptyset.
\end{align}
\end{proposition}
\begin{proof}
 Let $X\in L^2(P_0)$. Recall that for $\Theta$ defined in (\ref{eq:ThetaEllipsoidal}), $\bar{\theta}$ from Theorem \ref{thm:SolveUncertainValuationpb} does not vary with $X$ and solves 
 the minimization  problem (\ref{eq:ProjectionMinNormnu}). Now if (\ref{eq:UncertaintyBigEnough}) holds, then there exists $\theta\in\Theta$ such that $P^\theta\in\mathcal{R}\cap \mathcal{M}^e(S)\neq\emptyset$, 
 i.e.\  $\widehat{Q}^\theta = P^\theta$, and therefore $\xi^\theta=0$.  This implies that $\theta=\bar{\theta}=-\xi^0$ and hence $\xi^0\in \Theta$. As a consequence, the generator $f=f^{\bar{\theta}}$ 
 of the BSDE (\ref{eq:BSDEsforPiuunc}) coincides  with that of the BSDE (\ref{eq:BSDERiskMin}). By uniqueness of standard BSDE solutions follows $\pi^{u}_\cdot(X)=Y^X$. Now from 
 Theorem \ref{thm:SolveUncertainValuationpb} and Theorem \ref{thm:CharactOfSaddlePt} one obtains that $\phi^*=\bar{\phi}=\Pi(Z^X)$.
Conversely, suppose that (\ref{eq:equalityPiPhi}) holds. Then the generator $f=f^{\bar{\theta}}$ for the BSDE (\ref{eq:BSDEsforPiuunc}) and the one for (\ref{eq:BSDERiskMin}) are equal 
everywhere by \citet[Thm.7.1 and Rmk.4.1]{Coquetetal}. This implies (since $\Pi^\bot(\theta)=0$ for all $\theta\in\Theta$) that $\xi^{\bar{\theta}}=0$, i.e.\  
$\widehat{Q}^{\bar{\theta}}=P^{\bar{\theta}}$, and hence $\mathcal{R}\cap \mathcal{M}^e(S)\neq \emptyset$.
\end{proof}
Now we can make the previously described relation between risk minimization and good-deal hedging under large uncertainty  precise. Using Section~\ref{sec:Parameterization} and
$dS/S =\sigma d\widehat{W}^0$,  any Galtchouk-Kunita-Watanabe (GKW) decomposition \citep[cf.][]{Schweizer} of a continuous local $Q$-martingale ($Q$ in $ \mathcal{Q}^{\textrm{ngd}}$) wrt.\  $\sigma\cdot \widehat{W}^0$ gives a GKW decomposition wrt.\  $S$ and vice versa.  Note that $\bar{Q}=\bar{Q}(X)$ depends on $X$ in the
\begin{theorem}\label{thm:RobustGDHedgingRiskMin}
Let the assumptions of Proposition \ref{pro:SaddPtCondEquivalence} and (\ref{eq:UncertaintyBigEnough}) hold. For $X$ in  $L^2(P_0)$, let $(Y,Z)$ be the standard solution of the BSDE (\ref{eq:BSDERiskMin}). 
Then $\pi^{u}_\cdot(X)=Y$  has the GKW decomposition w.r.t.\ $\sigma\cdot \widehat{W}^0$ (and $S=\mathrm{diag}(S) \sigma\cdot \widehat{W}^0$, cf.\ Sect.\ref{sec:Parameterization})
    \begin{equation}\label{eq:LocalRiskMin}
    \pi^{u}_t(X) = \pi^{u}_0(X) + \phi^*\cdot \widehat{W}^0_t+R^{\phi^*}_t,\quad  t\in[0,T],
    \end{equation}
with $\phi^*=\Pi(Z)$. The tracking error $R^{\phi^*}(X)=\Pi^\bot_\cdot(Z)\cdot W^{\bar{Q}}$ is a $\bar{Q}$-martingale orthogonal to $\sigma\cdot \widehat{W}$, for $\bar{Q}(X)\in \mathcal{Q}^{\textrm{ngd}}$ 
given by $d\bar{Q}/dP_0=\mathcal{E}\left((-\xi^0+\bar{\eta})\cdot W^0\right)$ with 
\begin{equation*}
\bar{\eta}_t = h_t\big({\Pi_t^\bot(Z_t)}^{\textrm{tr}}A^{-1}_t\Pi_t^\bot(Z_t)\big)^{-1/2}A^{-1}_t\Pi_t^\bot(Z_t),\quad t\in [0,T].
\end{equation*}
\end{theorem}
\begin{proof}
 From Proposition \ref{pro:SaddPtCondEquivalence} we have $\phi^*=\Pi(Z)$ and $\pi^u_\cdot(X)=Y$. By the definitions of $\bar{\eta}$ and $\bar{Q}$, $Y_t = Y_0+Z\cdot W^{\bar{Q}}_t$ 
 holds for all $t\in[0,T]$. As a consequence, one obtains $ \pi^u_t(X)= \pi^u_0(X)+\phi^*\cdot \widehat{W}^{0}_t+ \Pi_\cdot^\bot(Z_\cdot) \cdot W^{\bar{Q}}_t, t\in [0,T]$. 
 Thus (\ref{eq:LocalRiskMin}) holds with $R^{\phi^*}(X)=\Pi^\bot_\cdot(Z)\cdot W^{\bar{Q}}$ being a $\bar{Q}$-martingale orthogonal to $S=S_0+\int_0^\cdot\mathrm{diag}(S_t)\sigma_tdW^{\bar{Q}}_t$ since $\sigma(\Pi^\bot_\cdot(Z))=0$. 
Furthermore since $\phi^*_t\bot \Pi^\bot_t(Z_t)$, and $\phi^*\cdot \widehat{W}^{0}=\phi^*\cdot W^{\bar{Q}}$, then $R^{\phi*}(X)$ is also orthogonal to 
$\phi^*\cdot \widehat{W}^{0}$ under $\bar{Q}$. Therefore (\ref{eq:LocalRiskMin}) is the GKW decomposition of $\bar{\pi}^u_\cdot(X)$ under $\bar{Q}$. 
\end{proof}

A seminal no-trade result by \citet{DowWerlang} shows that a utility optimizing agent  abstains from taking any position in a risky asset if uncertainty is too large. In comparison, the above 
theorem shows that a good-deal hedger keeps dynamically trading according to the risk minimizing component $\Pi(Z)$ but ceases to accept any speculative component.  
 In a  setting quite different to ours, \citet[][]{BoyarchenkoCerratoCrosbyHodges} demonstrate by numerical computation in an  example that relative benefits of dynamic 
hedging compared to static hedging could decrease if uncertainty increases. This is intuitive, as \citep[cf.\ e.g.][]{Cont}
 static hedges are less exposed to model risk.
Proposition~\ref{pro:SaddPtCondEquivalence} similarly addresses how increasing uncertainty affects dynamic hedging, but is different in that it  offers  theoretical conditions under which dynamic good-deal 
hedging $\phi^*$ ceases to comprise a speculative components in order to compensate for exposures to non-spanned risk. 

\subsection{Example with closed-form solutions under model uncertainty}\label{subsec:OptionOnNonTradedAssetUncertainty}
\label{sec:explUC}
The usual filtration is generated by a $2$-dimensional $P_0$-Brownian motion $W^0:=(W^{0,S},W^{0,H})^{\text{tr}}$. We consider a single traded risky asset with price $S$ and a non-traded asset 
with value $H$ modelled under $P_0$ for $t\in [0,T]$ by
\begin{equation*}
	dS_t=S_t\sigma^S\big(\xi^{0,S}dt+ dW^{0,S}_t\big),\quad dH_t=H_t\big(\gamma dt + \beta (\rho dW^{0,S}_t+\sqrt{1-\rho^2} dW^{0,H}_t)\big)
\end{equation*}
with $S_0,H_0>0$, scalars $\sigma^S,\beta>0$, $\gamma,\xi^{0,S}\in \R$ and correlation coefficient $\rho\in [-1,1]$. We derive robust good-deal bounds and hedging strategies in 
closed-form, for European call options on the non-traded asset and for no-good-deal restriction and uncertainty modelled (as in Section \ref{subsec:ImpactUncertaintyGDHedging}) 
using the radial sets  $C^0 = \{x\in\R^2:\lvert x\rvert\le h\}$ and $\Theta^0 = \{x\in\R^2:\abs{x}\le \delta\}$ for scalars $h,\delta\ge0$. Here one has $\Theta=\Theta^0\cap\mathrm{Im}\ \sigma = [-\delta,\delta]\times\{0\}$, 
for $\sigma=(\sigma^S,0)$, and hence $\xi^\theta=(\xi^{\theta,S},0)^{\text{tr}}:=(\xi^{0,S}+\theta^{S},0)^{\text{tr}}\in\mathrm{Im}\ \sigma,$ for models $P^{\theta}$ with $\theta=(\theta^S,0)^{\text{tr}}$, 
where $\theta^S\in[-\delta,\delta]$. From \eqref{eq:PhistatAequalBCaseA}, with $A= B\equiv \mathrm{Id}_{\R^2}$ and $r=1$, the worst-case model $P^{\bar{\theta}}$ corresponds to
\begin{equation}\label{eq:WorstCaseThetaExa}
\bar{\theta}^S = -\xi^{0,S}I_{\{\abs{\xi^{0,S}}\le \delta\}} -\delta\frac{\xi^{0,S}}{\abs{\xi^{0,S}}}I_{\{\abs{\xi^{0,S}}> \delta\}}	.
\end{equation}
By Theorems \ref{thm:SolveUncertainValuationpb}-\ref{thm:CharactOfSaddlePt}, the robust good-deal bound and hedging strategy for call option $X:=(H_T-K)^+$ are given by $\pi^u_\cdot(X)=Y$ and  
$\bar{\phi}(X)=(Z^1+\frac{\lvert Z^2\rvert}{\sqrt{h^2-\lvert\xi^{\bar{\theta},S}\rvert^2}}\xi^{\bar{\theta},S},0)^{\text{tr}}$, for $(Y,Z:=(Z^1,Z^2)^{\text{tr}})$ solving the BSDE (\ref{eq:BSDEsforPiuunc}), equaling the BSDE (\ref{eq:BSDEPiuNu}) for $\theta=\bar{\theta}$:
\begin{equation}\label{eq:BSDEvalExCloseFormSolUncertainty}
 -dY_t= \big(-\xi^{\bar{\theta},S}Z^1_t+(h^2-\lvert \xi^{\bar{\theta},S}\rvert^2)^{1/2}\lvert Z^2_t\rvert \big)dt - Z^{\textrm{tr}}_t dW^{\bar{\theta}}_t\quad\text{and}\quad Y_T=X,
\end{equation}
with $W^{\bar{\theta}}_t := (W^{0,S}_t - \bar{\theta}^St,W^{0,H}_t)^{\text{tr}},\ t\in [0,T]$. Writing (\ref{eq:BSDEvalExCloseFormSolUncertainty}) under $\widehat{Q}^{\bar{\theta}}$ and using 
(\ref{eq:WorstCaseThetaExa}), arguments analogous to those in the derivation of (\ref{eq:CloFormUpperGDBCall}) yield
\begin{align*}
	\pi^u_t(X) &=N(d_{+})H_te^{\tilde{\alpha}_+ (T-t)}-KN(d_{-})\\
		    &=:\text{B/S-call-price}\big(\text{time: } t,\ \text{spot: }H_te^{\tilde{\alpha}_+ (T-t)},\ \text{strike: } K, \text{vol: } \beta\big),\\
		   \pi^l_t(X) &=\text{B/S-call-price}\big(\text{time: } t,\ \text{spot: }H_te^{\tilde{\alpha}_- (T-t)},\ \text{strike: } K, \text{vol: } \beta\big) ,
\end{align*}
with $d_{\pm}:=\big(\ln \big(H_t/K\big)+\big(\tilde{\alpha}_+\pm \frac{1}{2}\beta ^2\big)(T-t)\big)/\big(\beta\sqrt{T-t}\big)$, $\tilde{\alpha}_{\pm} := \gamma+\beta\big(-\rho\xi^{0,S}\pm \tilde{h}\sqrt{1-\rho^2}\big)$ 
and $\tilde{h}:=hI_{\{\abs{\xi^{0,S}}\le \delta\}}+\big(h^2-\big\lvert \xi^{0,S}-\delta \xi^{0,S}/\lvert \xi^{0,S}\rvert\big\rvert^2\big)^{1/2}I_{\{\abs{\xi^{0,S}}> \delta\}}$.
Analogously to the derivation of (\ref{eq:CloFormHedgStraUpperGDBCall}), note that $Z=e^{\tilde{\alpha}_+ (T-t)}N(d_+)H_t\beta(\rho,\sqrt{1-\rho^2})^{\text{tr}}$. Hence the (seller's) robust good-deal hedging strategy is obtained $P\otimes dt$-a.e.\  as 
\begin{equation*}
\bar{\phi}_t(X) = e^{\tilde{\alpha}_+ (T-t)}N(d_+)H_t \beta\ \Big(\rho+\frac{\sqrt{1-\rho^2}\xi^{0,S}}{\widetilde{h}\xi^{0,S}}\big(\lvert\xi^{0,S}\rvert-\delta\big)\mathds{1}_{\{\lvert\xi^{0,S}\rvert>\delta\}},0\Big)^{\textrm{tr}},\ t\in [0,T].
\end{equation*}
For $\lvert \xi^{0,S}\rvert>\delta$, the speculative nature of $\bar{\phi}(X)$ is reflected by the presence of the second summand in the first component of $\bar{\phi}(X)$ above.
For $\lvert \xi^{0,S}\rvert\le \delta$, this summand vanishes and the function $\delta\mapsto\tilde{\alpha}_+$ is constant on $\delta\in[\lvert \xi^{0,S}\rvert,\infty]$. In this case
robust good-deal hedging is then risk-minimizing and non-speculative as proved in Theorem \ref{thm:RobustGDHedgingRiskMin}. Note that for $\delta= \xi^{0,S}=0$ (i.e.\  risk-neutral setting 
under $P_0=\widehat{Q}^0$ in absence of uncertainty), we recover formulas of Section \ref{subsec:OptionOnNonTradedAsset} for $n=2$ and $d=1$.
\begin{figure}
        \centering
        \begin{subfigure}[b]{0.42\textwidth}
                \includegraphics[scale = 0.45]{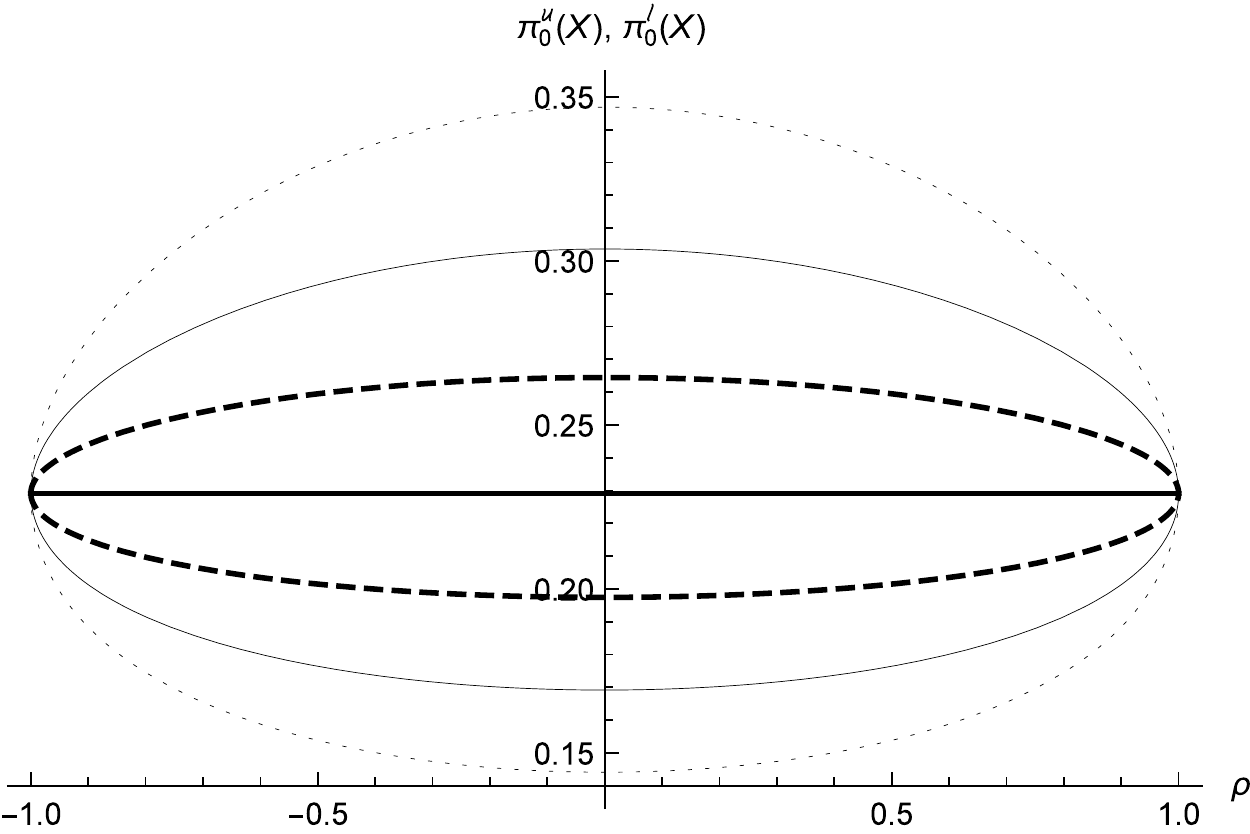}
                \caption{{\tiny $h_0=\xi^{0,S}=0$, $\ \delta=0$}}
                \label{fig:plotRhoAllZeroMPRVh}
        \end{subfigure}
        ~\ \ 
        \begin{subfigure}[b]{0.42\textwidth}
                \includegraphics[scale = 0.6]{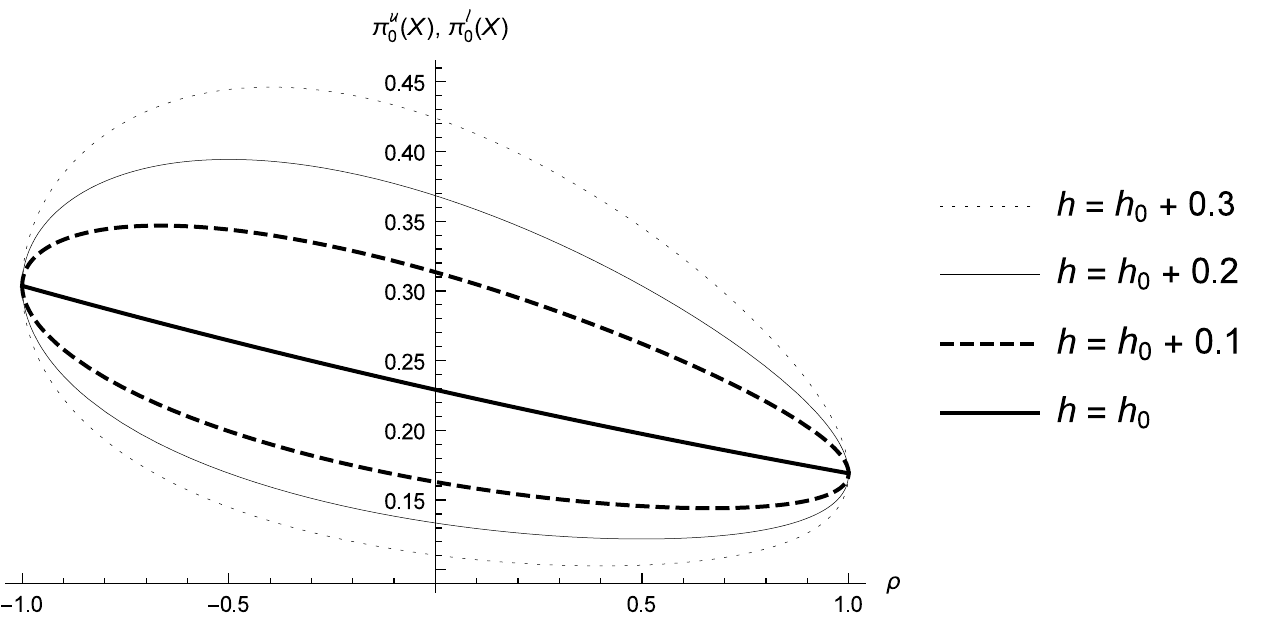}
                \caption{{\tiny $h_0=\xi^{0,S}=0.2$, $\ \delta=0$}}
                \label{fig:plotRhoAllPosMPRVh}
        \end{subfigure}
        \\
        \begin{subfigure}[b]{0.4\textwidth}
                \includegraphics[scale=0.55]{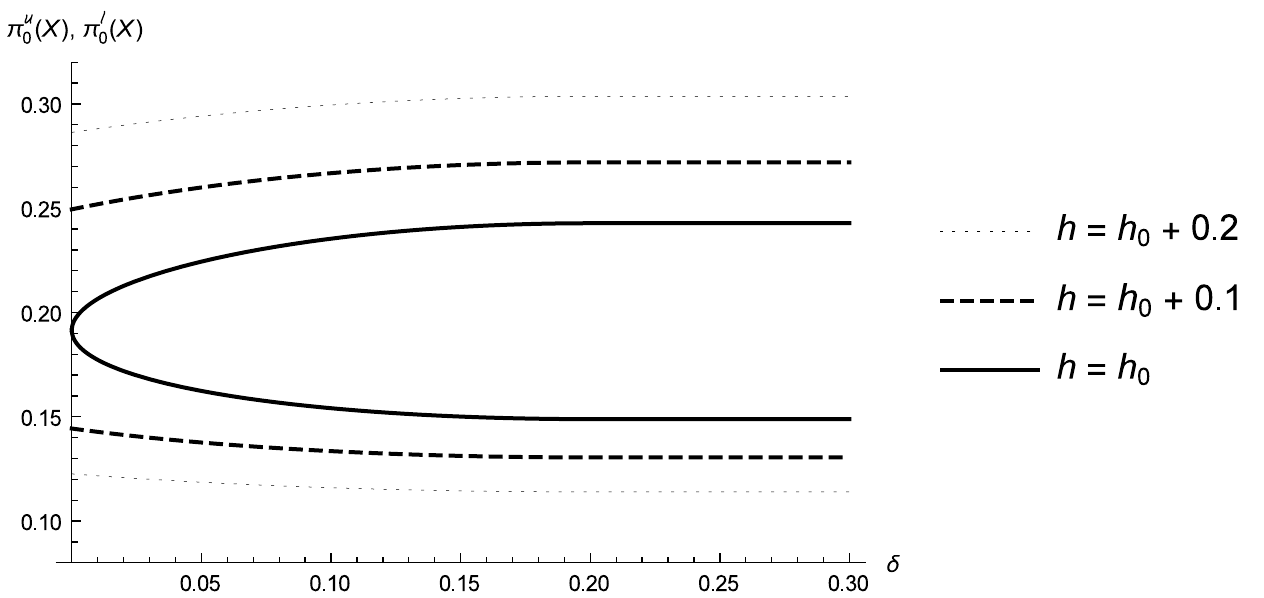}
                \caption{{\tiny $h_0=\xi^{0,S}=0.2,\ \rho=0.6$}}
                \label{fig:GDBsatZeroAgainstdeltaXiBig}
        \end{subfigure}
        ~\qquad\qquad\ \ 
        \begin{subfigure}[b]{0.4\textwidth}
                \includegraphics[scale = 0.6]{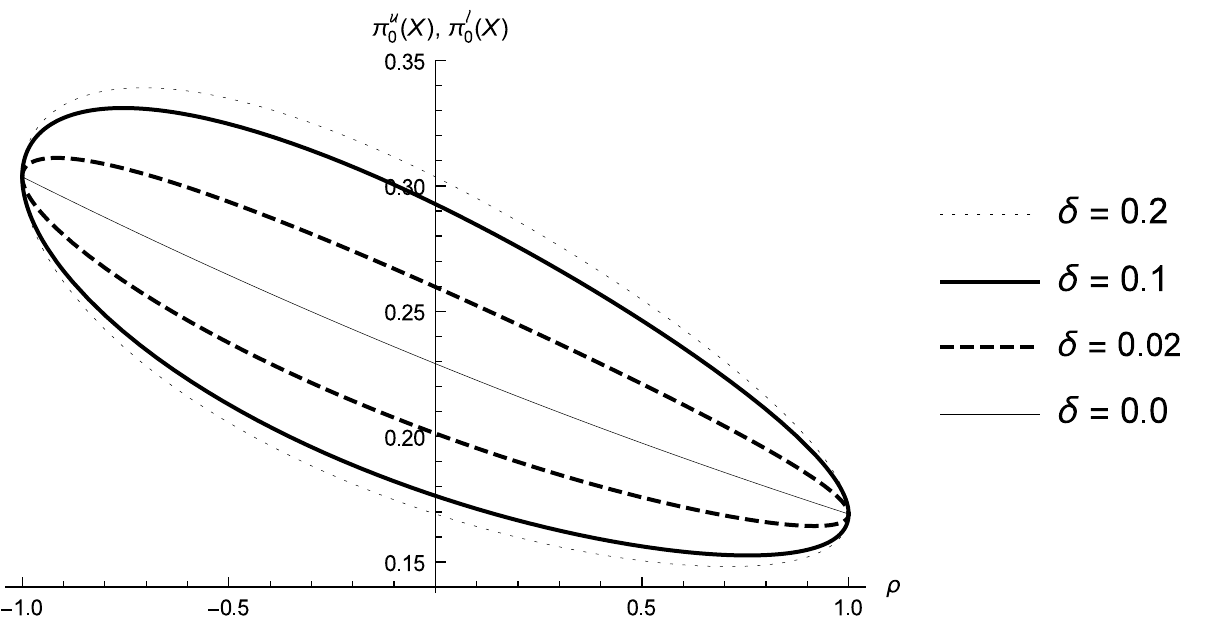}
                \caption{{\tiny $h=\xi^{0,S}=0.2$}}
                \label{fig:plotRhoAllPosMPRVDelta}
        \end{subfigure}
        \caption{Dependence of $\pi^u_0(X), \pi^l_0(X)$ on $\rho,\ h$ and/or $\delta$}
        \label{fig:GDBsatZeroAgainstRhoVhDelta}
\end{figure}
Figure \ref{fig:GDBsatZeroAgainstRhoVhDelta} illustrates the dependence of the bounds $\pi^u_0(X)$ and $\pi^l_0(X)$ in the presence of uncertainty, on the correlation coefficient $\rho$, 
uncertainty size $\delta$ and no-good-deal constraint (optimal growth rate bound) $h$, and for global parameters $\gamma=0.05$, $\beta=0.5$, $K=1$, $H_0=1$ and $T=1$. Figures \ref{fig:plotRhoAllZeroMPRVh},\ref{fig:plotRhoAllPosMPRVh} 
are plots of $\pi^u_0(X)$ and $\pi^l_0(X)$ as functions of $\rho$ for fixed $\delta=0$ (i.e.\  absence of uncertainty) and $\xi^{0,S}\in\{0,0.2\}$, showing how the good-deal bounds vary 
for different values of $h$. Figure \ref{fig:plotRhoAllPosMPRVDelta} contains a similar plot for fixed $h=\xi^{0,S}=0.2$, showing how the bounds vary with $\rho$ for different values of $\delta$. 
One can observe that the maximum of $\pi^u_0(X)$ and minimum of $\pi^l_0(X)$ are attained at $\rho=0$ only for $\xi^{0,S}=0$ (cf.\  Figure \ref{fig:plotRhoAllZeroMPRVh}). In other words, if the 
market price of risk $\xi^{0,S}$ is zero (hence $P_0=\widehat{Q}^0$), then the largest good-deal bounds are obtained when the traded and non-traded assets are uncorrelated (i.e.\ $\rho=0$). On the other 
hand if $\xi^{0,S}>0$ (as e.g.\  in Figures \ref{fig:plotRhoAllPosMPRVh},\ref{fig:plotRhoAllPosMPRVDelta}), the plots are tilted so that the maximum of $\pi^u_0(X)$ (resp.\ minimum of $\pi^l_0(X)$) is 
reached at $\rho<0$ (resp. $\rho>0$). For $\pi^u_0(X)$, this is explained by the fact that if the market price of risk $\xi^{0,S}$ is positive, the supremum in (\ref{eq:definitionPiuandPil}) is 
maximized by the no-good-deal measure $\bar{Q}=Q^{\bar{\lambda}}\in\mathrm{Q}^{\text{ngd}}(P_0)$ with Girsanov kernel $\bar{\lambda}:=\big(-\xi^{0,S},(h^2-\big\lvert \xi^{0,S}-\delta \xi^{0,S}/\lvert \xi^{0,S}\rvert\big\rvert^2)^{1/2}\big)^{\text{tr}}$ 
under which the upward drift $\tilde{\alpha}_+$ of the underlying price process $H$ is maximized, clearly at a negative correlation $\rho$. The explanation for 
$\pi^u_0(X)$ is similar, with $\tilde{\alpha}_-$ being minimal at a positive correlation, for $\xi^{0,S}>0$. For $\xi^{0,S}<0$ the tilt of the plots occurs in the other direction.
That the good-deal bounds in Figures \ref{fig:plotRhoAllZeroMPRVh},\ref{fig:plotRhoAllPosMPRVh},\ref{fig:plotRhoAllPosMPRVDelta} coincide for  perfect correlation $\rho=\pm 1$ is clear, because in this case 
derivatives $X$ on $H$ are attainable and admit unique no-arbitrage prices, implying $\pi^u_\cdot(X)=\pi^l_\cdot(X)$. Finally, Figure \ref{fig:GDBsatZeroAgainstdeltaXiBig} illustrates 
the evolution  with respect to $\delta$ of the good-deal bounds at time $t=0$ for $\rho=0.6$, $\xi^{0,S}=0.2$  and different values of $h$, with $\lvert\xi^{0,S}\rvert$ chosen as the smallest value $h_0$ of 
$h$. One observes that for each given $h$, the good-deal bound curves become flat for $\delta\ge\lvert\xi^{0,S}\rvert$ (as predicted by Proposition \ref{pro:SaddPtCondEquivalence}), and 
match (i.e.\  $\pi^u_0(X)=\pi^l_0(X)$) for $\delta=\lvert\xi^{0,S}\rvert-h_0=0$ (as might be expected in the absence of uncertainty for a degenerate expected growth rate bound $h=\lvert\xi^{0,S}\rvert$).

\section{Appendix}\label{app:AppendixB}\label{app:AppendixA}
This section includes  lemmas and derivations omitted from the paper's main body.

\begin{lemma}\label{lem:OptimizationWithPhiSolved}
 For $d< n$, let $\sigma \in \R^{d\times n}$ be of full-rank, $A\in\R^{n\times n}$ be symmetric and positive definite, and $h>0$, $Z\in\R^n$, $\xi\in\mathrm{Im}\ \sigma^{\textrm{tr}}$. 
 Let $\alpha'>0$ be a constant of ellipticity of $A^{-1}$ and assume that $\abs{\xi}<h\sqrt{\alpha'}$ and $A^{-1}(\mathrm{Ker}\ \sigma) = \mathrm{Ker}\ \sigma$. Then 
 $\bar{\phi}:= \Pi(Z)+\big(\Pi^\bot(Z)^{\textrm{tr}}A^{-1}\Pi^\bot(Z)\big)^{1/2}\big(h^2-\xi^{\textrm{tr}}A\xi\big)^{-1/2}A\xi$ is the unique minimizer of the function 
 $\phi\mapsto F(\phi):=-\xi^*\phi +h\big(\left(Z-\phi\right)^{\textrm{tr}}A^{-1}\left(Z-\phi\right)\big)^{1/2}$ on $\mathrm{Im}\ \sigma^{\textrm{tr}}$.
\end{lemma}
\begin{proof}
 The proof is an application of the classical Kuhn-Tucker Theorem \citep[cf.][Sect.28]{RockafellarConvexAnalysis}. For details see \citep[][Lemma 3.35]{KentiaPhD} and proof.
\end{proof}
\begin{lemma}\label{lem:MinimaxPbSolved}
 Let $d< n$, $h>0$ be constant, $Z\in\R^n$, $A\in \R^{n\times n}$  a symmetric positive definite matrix,  $\sigma\in\R^{d\times n}$ a full ($d$)-rank matrix, and $\xi^0\in \Phi := \text{Im}\ \sigma^{\text{tr}}$. 
 Let $\Theta\subset \R^n$ be a convex-compact set, and $F: \R^n\times\R^n\ni(\phi,\theta)\mapsto \theta^\textrm{tr}(Z-\phi)-{\xi^0}^\textrm{tr}\phi + h\big((Z-\phi)^\textrm{tr}A^{-1}(Z-\phi)\big)^{1/2}$.
 Then the minmax identity 
\(\label{eq:MinmaxIdentity}
\inf_{\phi\in\Phi}\ \sup_{\theta\in\Theta}F(\phi,\theta) = \sup_{\theta\in\Theta}\ \inf_{\phi\in\Phi}F(\phi,\theta). 
\)
holds.
\end{lemma}
\begin{proof}
 For all $\phi\in\R^n$, the function $\theta\mapsto F(\phi,\theta)$ is concave, continuous. For all $\theta\in\R^n$ the function $\phi\mapsto F(\phi,\theta)$ is convex and 
 continuous. As $\Theta\subset\R^n$ is convex and compact, and  $\Phi = \text{Im}\ \sigma^{\text{tr}}$ is convex and closed,  a minimax theorem \citep[][Ch.VI, Prop.2.3]{EkelandTemam}
 applies and the minmax identity holds. 
\end{proof}
\begin{proof}[Proof of Lemma \ref{lem:PropDynamicRiskMeasMeandQngdareMstable}]
Part a) is classical (see \citealt{Delbaen} and cf.\ previously given other references). As for Part b), m-stability and convexity of $\mathcal{M}^e$ follow from \citep[][Prop.5]{Delbaen}. 
Convexity of $\mathcal{Q}^{\mathrm{ngd}}$ follows from that of $\mathcal{M}^e$ and the values of $C$. To show m-stability of $\mathcal{Q}^{\mathrm{ngd}}$, let $Z^i=\mathcal{E}(\lambda^i\cdot W)\in \mathcal{Q}^{\mathrm{ngd}},\ i=1,2$, $\tau\le T$ 
be a stopping time and $Z=I_{[0,\tau]}Z^1_\cdot+ I_{]\tau,T]}Z^1_\tau Z^2_\cdot /Z^2_\tau$. Since $\mathcal{M}^e$ is m-stable, then $Z\in \mathcal{M}^e$ and one has $Z=\mathcal{E}(\lambda\cdot W)$ 
for some predictable process $\lambda$. It remains to show that $\lambda$ is bounded and that $\lambda\in C$. From the expression of $Z$, writing the densities $Z,Z^1,Z^2$ as ordinary exponentials by distinguishing $t\le \tau$ 
and $t\ge \tau$, and taking the logarithm yields $(\lambda-I_{[0,\tau]}\lambda^1-I_{]\tau,T]}\lambda^2)\cdot W= \frac{1}{2}\int_0^\cdot\big(\abs{\lambda_s}^2-I_{[0,\tau]}(s)\abs{\lambda^1_s}^2-I_{]\tau,T]}(s)\abs{\lambda^2_s}^2\big)ds.$
Since $\filt$ is the augmented Brownian filtration, then $[0,\tau]$ and $]\tau,T]$ are predictable and so is $\lambda-I_{[0,\tau]}\lambda^1-I_{]\tau,T]}\lambda^2$. Hence $\left(\lambda-I_{[0,\tau]}\lambda^1-I_{]\tau,T]}\lambda^2\right)\cdot W$ 
is a continuous local martingale of finite variation and is thus equal to zero. As a consequence $\lambda=I_{[0,\tau]}\lambda^1+I_{]\tau,T]}\lambda^2$ is bounded since $\lambda^1,\lambda^2$ are, and satisfies $\lambda\in C$
since $C$ is convex-valued. 
\end{proof}

\begin{proof}[Proof of Theorem~\ref{thm:MinimalSupersolGDB}]
As shown in \citep[][Thm.3.7]{KentiaPhD}, $\pi^u_\cdot(X)$ admits under $\widehat{Q}$ the Doob-Meyer decomposition $\pi^{u}_\cdot(X) = \pi^{u}_0(X)+Z\cdot \widehat{W} - A=\pi^{u}_0(X)+Z\cdot W + \int_0^\cdot \xi_t^{\textrm{tr}}Z_tdt- A,$
 where $Z\in\h^2(\widehat{Q})$ and $A$ is a non-decreasing predictable process with $A_0=0$. Alternatively one rewrites $\displaystyle -d\pi^u_t(X) =g_t(Z_t)dt - Z_t^{\textrm{tr}}dW_t+dK_t,$
with $ K:= A -\int_0^\cdot \xi_t^{\textrm{tr}}Z_tdt - \int_0^\cdot\esssup_{\lambda\in \Lambda}\lambda_t^{\textrm{tr}}Z_t\ dt$ being finite-valued and predictable. For  
$(\pi^u_\cdot(X),Z,K)$ to be a supersolution to the BSDE with parameters $(g,X)$ it suffices to show that $K$ is non-decreasing. For any $\lambda=-\xi+\eta \in \Lambda$, one can construct 
the sequence of $\lambda^n=-\xi+\eta^n\in \Lambda$ Girsanov kernels of measures $Q^n\in \mathcal{Q}^{\mathrm{ngd}}$ with $\eta^n=\eta I_{\{\lvert\eta\rvert\le n\}}$ such that 
$\lambda^n\to \lambda\ P\otimes dt\text{-a.e.}$ as $n\to \infty$. For each $Q^n$ it holds  $ \pi^{u}_\cdot(X) = \pi^{u}_0(X)+Z\cdot W^{Q^n} + \int_0^\cdot {Z_t}^{\textrm{tr}}\eta^n_tdt- A.$ 
Since $\pi^u_\cdot(X)$ is a bounded $Q^n$-supermartingale, then $dA_t - \xi_t^{\textrm{tr}}Z_tdt \ge Z_t^{\textrm{tr}}\lambda^n_tdt,\ \text{for all } n\in\N.$ Taking the limit as 
$n\to \infty$ and using dominated convergence one obtains $dA_t - \xi_t^{\textrm{tr}}Z_tdt \ge Z_t^{\textrm{tr}}\lambda_tdt.$ Now taking the essential supremum over all $\lambda\in\Lambda$ 
yields $dK_t\ge0$.

To show that the supersolution $(\pi^u_\cdot(X),Z,K)$ is minimal, it suffices by the dynamic principle \citep[cf.][Lem.3.6]{KentiaPhD} to show that the $Y$-component of any other 
supersolution is a c\`adl\`ag $Q$-supermartingale for any $Q\in \mathcal{Q}^{\mathrm{ngd}}$. Let $(\bar{Y},\bar{Z},\bar{K})$ be a supersolution of the BSDE for parameters $(g,X)$, with 
$\bar{Y}\in \mathcal{S}^\infty$. By change of measure, under a $Q\in \mathcal{Q}^{\mathrm{ngd}}$ with Girsanov kernel  $\lambda^Q\in \Lambda$ we have
\begin{equation}\label{eq:ForYbarandKbar}
 -d\bar{Y}_t =\Big(\esssup_{\lambda\in \Lambda}\lambda_t^{\textrm{tr}}\bar{Z}_t - \bar{Z}_t^{\textrm{tr}}\lambda^Q_t\Big)dt - \bar{Z}_t^{\textrm{tr}}dW^Q_t+d\bar{K}_t,\quad t\in [0,T],
\end{equation}
\begin{equation}
\label{eq:KbarandPositive}
\text{and get} \quad  d\bar{K}_t+ \Big(\esssup_{\lambda\in \Lambda}\lambda_t^{\textrm{tr}}\bar{Z}_t - \bar{Z}_t^{\textrm{tr}}\lambda^Q_t\Big)dt\ge0,  \quad t\in [0,T],
\end{equation}
by using that  $\bar{K}$ is non-decreasing. 
 By (\ref{eq:ForYbarandKbar}-\ref{eq:KbarandPositive}) and $\bar{Y}\in\mathcal{S}^\infty$, the local martingale $\bar{Z}\cdot W^Q$ is bounded from below, and thus is a supermartingale. 
As $\bar{Y}\in\mathcal{S}^\infty$, the integral of (\ref{eq:KbarandPositive}) on $[0,T]$ is in $L^1(Q)$ and so $\bar{Y}$ is a $Q$-supermartingale. 
\end{proof}

\begin{proof}[Derivation of (\ref{eq:GoodDealUpperForPut}),(\ref{eq:GoodDealHedgingForPut})]
The stochastic exponential $\mathcal{E}\big((\varepsilon/\sqrt{\nu})\cdot W^{\nu}\big)$ is a uniformly integrable martingale which defines a measure $\bar{Q}\in\overline{\mathcal{Q}^{\textrm{ngd}}}\supseteq \mathcal{Q}^{\textrm{ngd}}$
(see (\ref{eq:QngdBARDefinition}) for definition of $\overline{\mathcal{Q}^{\textrm{ngd}}}\subset \mathcal{M}^e$)
with Girsanov kernel $\bar{\lambda}:=\varepsilon/\sqrt{\nu}$, i.e.\  $d\bar{Q}/dP = \mathcal{E}\big((\varepsilon/\sqrt{\nu})\cdot W^{\nu}\big)$. Indeed, applying  \citep[][Thm.2.4 and Sect.6]{CheriditoFilipovicYor} one gets that 
$\mathcal{E}\big((\varepsilon/\sqrt{\nu})\cdot W^{\nu}\big)$ and $S=S_0\mathcal{E}\left(\sqrt{\nu}\cdot W^S\right)$ are uniformly integrable $P$- respectively $Q$-martingales. 
The variance process $\nu$ under $\bar{Q}$ is again a CIR process with parameters $(\bar{a},b,\beta,\rho)$ where $\bar{a}:=a+\beta\varepsilon\sqrt{1-\rho^2}>a$ and the Feller condition 
$\beta^2\le 2\bar{a}$ still holds. For a put option $X=(K-S_T)^+ \in L^\infty$, $\bar{Y}_t:=E^{\bar{Q}}_t[X]$ are given by the Heston formula \citep[cf.][]{Heston}, applied 
under $\bar{Q}$ (instead of $P$). Since  the Heston price is  non-decreasing in the mean reversion level of the variance process \citep[][Prop.5.3.1]{OuldAly} one expects that 
$\pi^u_t\left(X\right) =\bar{Y}_t= E^{\bar{Q}}_t[X]$. Let us make this precise. For $Q\in\overline{\mathcal{Q}^{\mathrm{ngd}}}$ with Girsanov kernel $\lambda$ satisfying 
$\abs{\lambda}\le \varepsilon/\sqrt{\nu}$, one  has $Y^Q_T=\bar{Y}_T=X$ with $Y^Q_t=E^Q_t[X]$. Using Feynman-Kac, $\bar{Y}_t=u(t,S_t,\nu_t)$ for a function 
$u\in \mathcal{C}^{1,2,2}([0,T]\times\R^+\times\R^+)$ with $\frac{\partial u}{\partial \nu}\ge 0$ \citep[see][Thm.5.3.1, Cor.5.3.1]{OuldAly}. By It\^o's formula and change of measure 
follows
\begin{align}
	d\bar{Y}_t = &\beta\sqrt{1-\rho^2}\sqrt{\nu_t}\big(\lambda_t-\frac{\varepsilon}{\sqrt{\nu_t}}\big)\frac{\partial u}{\partial \nu}(t,S_t,\nu_t)dt+\beta\sqrt{1-\rho^2}\sqrt{\nu_t}\frac{\partial u}{\partial \nu}(t,S_t,\nu_t)dW^{Q,\nu}_t\notag\\
		     &+\Big(S_t\sqrt{\nu_t}\frac{\partial u}{\partial S}(t,S_t,\nu_t)+\beta\rho\sqrt{\nu_t}\frac{\partial u}{\partial \nu}(t,S_t,\nu_t)\Big)dW^{S}_t,\quad t\in[0,T].\label{eq:ItogivesdPbarunderQ}
\end{align}
Since $X$ is bounded, then $\bar{Y}$ is in $\mathcal{S}^\infty(Q)$ and a $Q$-supermartingale  by (\ref{eq:ItogivesdPbarunderQ}) . Hence  $Y^Q_t \le \bar{Y}_t$ for all $Q\in\overline{\mathcal{Q}^{\textrm{ngd}}}$, which by Part 1. of \citet[][Thm.3.7]{KentiaPhD} 
implies the claim and thus we obtain the Heston type formula (\ref{eq:GoodDealUpperForPut}). 

Since $\bar{Q}\in\overline{\mathcal{Q}^{\text{ngd}}}$ and $\pi^u_0(X)=E^{\bar{Q}}[X]$ with $X\in L^\infty$,
Corollary \ref{cor:GDBsolvesBSDEunderOptimalMeas} implies that the good-deal bound is the $Y$-component of the minimal 
solution $(\bar{Y},\bar{Z})\in\mathcal{S}^\infty\times\h^2 $ (note $P=\widehat{Q}$) of the BSDE (\ref{eq:BSDEmINsOL}) with generator $g_t(z)=\bar{\lambda}_tz^2=\varepsilon z^2/\sqrt{\nu_t},$ for $z=(z^1,z^2)$, and 
terminal condition $X$. Now consider the strategy 
\begin{equation*}
 \bar{\phi}_t=\bar{Z}^1_t=S_t\sqrt{\nu_t}\frac{\partial u}{\partial S}(t,S_t,\nu_t)+\beta\rho\sqrt{\nu_t}\frac{\partial u}{\partial \nu}(t,S_t,\nu_t)=S_t\sqrt{\nu_t}\ \Delta_t+\frac{\beta\rho}{2}\mathcal{V}_t.
\end{equation*}
Clearly $\bar{\phi}$ is in the set $\Phi = \h^2(\R)$ of permitted trading strategies since $\bar{Z}\in \h^2 (\R^2)$.  Recall that  $\mathcal{P}^{\text{ngd}}$ consists of 
$dQ/dP=\mathcal{E}\big((\lambda^S,\lambda^\nu)\cdot W\big)$ such that $\big\lvert(\lambda^S,\lambda^\nu)\big\rvert\le \varepsilon/\sqrt{\nu}$ with $(\lambda^S,\lambda^\nu)$ being bounded. 
For $Q\in \mathcal{P}^{\text{ngd}}$, any  wealth process $\phi\cdot W^S$, $\phi\in\Phi$,  is thus in $\mathcal{S}^1(Q)$. 
As $\mathcal{Q}^{\text{ngd}}\subseteq \mathcal{P}^{\text{ngd}}$ holds,  clearly $\pi^u_t(X) \le \rho_t(X-\int_t^T\phi_sdW^S_s)$ for any strategy $\phi \in \Phi$.
To prove that $\bar{\phi}$ is a good-deal hedging strategy, we show the reverse inequality $\pi^u_t(X) \ge E^Q_t\big[X-\int_t^T\bar{\phi}_sdW^S_s\big]$ for all $Q\in\mathcal{P}^{\text{ngd}}$. Let $Q\in\mathcal{P}^{\text{ngd}}$ with Girsanov kernel $(\lambda^S,\lambda^\nu)$.
Like in (\ref{eq:ItogivesdPbarunderQ}), we obtain for any stopping time $\tau$ that
\begin{align}
	\bar{Y}_{\tau\wedge T}- \int_{\tau\wedge t}^{\tau\wedge T}\bar{\phi}_sdW^S_s &= \bar{Y}_{\tau\wedge t} + \int_{\tau\wedge t}^{\tau\wedge T}\beta\sqrt{1-\rho^2}\sqrt{\nu_s}\big(\lambda^\nu_s-\frac{\varepsilon}{\sqrt{\nu_s}}\big)\frac{\partial u}{\partial \nu}(s,S_s,\nu_s)ds\notag\\
						  &\quad\ \qquad + L_{\tau\wedge T}-L_{\tau\wedge t},\label{eq:ToConcludeHedStraHeston}
\end{align}
for the local $Q$-martingale $L := \int_0^\cdot\beta\sqrt{1-\rho^2}\sqrt{\nu_s}\frac{\partial u}{\partial \nu}(s,S_s,\nu_s)dW^{Q,\nu}_s$. 
By $\frac{\partial u}{\partial \nu}\ge0$ and $\big\lvert\lambda^\nu\big\rvert\le \varepsilon/\sqrt{\nu}$
follows that
\(
	\bar{Y}_{\tau\wedge T}- \int_{\tau\wedge t}^{\tau\wedge T}\bar{\phi}_sdW^S_s 
\) 
is less than $ \bar{Y}_{\tau\wedge t} + L_{\tau\wedge T}-L_{\tau\wedge t}$.
Localizing $L$ along a sequence of stopping times $\tau_n\uparrow \infty$ and taking conditional $Q$-expectations 
yields $E^Q_t[\bar{Y}_{\tau_n\wedge T}- \int_{\tau_n\wedge t}^{\tau_n\wedge T}\bar{\phi}_sdW^S_s] \le \bar{Y}_{\tau_n\wedge t}$. Using $X\in L^\infty$ and $\bar{\phi}\cdot W^S\in S^1(Q)$, the claim then follows by 
dominated convergence. 
Hence (\ref{eq:GoodDealHedgingForPut}) holds for $\mathcal{V}_t:=\frac{\partial u}{\partial \sigma}(t,S_t,\nu_t) = 2\sigma_t\frac{\partial u}{\partial \nu}(t,S_t,\nu_t)$ and volatility $\sigma_t=\sqrt{\nu_t}$.
\end{proof}

\bibliographystyle{abbrvnat}
\bibliography{Hedging-under-generalized-good-deal-bounds-and-model-uncertainty_Becherer_Kentia}

\end{document}